\documentclass[10pt,reqno,oneside]{amsart}

\usepackage{cancel,bm, amssymb}
\usepackage{color}
\usepackage{color}
\usepackage{graphicx,framed}
\usepackage[colorlinks=true, pdfstartview=FitV, linkcolor=blue, citecolor=blue, urlcolor=blue]{hyperref}
\definecolor{shadecolor}{rgb}{0.95, 0.95, 0.86}

\renewcommand{\d}{{\mathrm d}}

\newcommand{\im}{\mathrm{i}}
\newcommand{\e}{\mathrm{e}}
\newcommand{\z}{\zeta}

\def\eq{\begin{equation}}
\def\endeq{\end{equation}}

\numberwithin{equation}{section}

\newtheorem{theo}{Theorem}[section]

\newtheorem{lem}[theo]{Lemma}
\newtheorem{rem}[theo]{Remark}
\newtheorem{problem}[theo]{Riemann-Hilbert Problem}
\newtheorem{remark}[theo]{Remark}
\newtheorem{prop}[theo]{Proposition} 
\newtheorem{cor}[theo]{Corollary}

\begin{document}

\title[Deformed Tracy-Widom distributions I.]{Large deformations of the Tracy-Widom distribution I. Non-oscillatory asymptotics}

\author{Thomas Bothner}
\address{Department of Mathematics, University of Michigan, 2074 East Hall, 530 Church Street, Ann Arbor, MI 48109-1043, United States}
\email{bothner@umich.edu}

\author{Robert Buckingham}
\address{Department of Mathematical Sciences, University of Cincinnati, P.O. Box 210025, Cincinnati, OH 45221-0025, United States}
\email{buckinrt@uc.edu}



\thanks{T.B. is grateful to Alexander Its for stimulating discussions about this project.  R.B.'s work was supported by the National Science Foundation through grants DMS-1312458 and DMS-1615718, and by a Faculty Release Fellowship from the Charles Phelps Taft Research Center.}

\begin{abstract}
We analyze the left-tail asymptotics of deformed Tracy-Widom distribution 
functions describing the fluctuations of the largest eigenvalue in 
invariant random matrix ensembles after removing each soft edge eigenvalue 
independently with probability $1-\gamma\in[0,1]$.  As $\gamma$ 
varies, a transition from Tracy-Widom statistics ($\gamma=1$) to classical 
Weibull statistics ($\gamma=0$) was observed in the physics literature by Bohigas, de Carvalho, and Pato 
\cite{BohigasCP:2009}.  We provide a description of this transition by rigorously computing the leading-order left-tail asymptotics of the thinned GOE, GUE and GSE Tracy-Widom distributions. In this paper, we obtain the asymptotic behavior in the non-oscillatory region with $\gamma\in[0,1)$ fixed (for the GOE, GUE, and GSE distributions) and $\gamma\uparrow 1$ at a controlled rate (for the GUE distribution).  This is the first step in an ongoing program to completely describe the transition between Tracy-Widom and Weibull statistics. As a corollary to our results, we obtain a new total-integral formula involving the Ablowitz-Segur solution to the second Painlev\'e equation.
%
%
\end{abstract}

\date{\today}
\maketitle

\section{Introduction}  The Tracy-Widom distribution functions $F_{\beta}$ are universal probability distributions describing extremal behavior of, among a host of other applications, eigenvalues in the Gaussian invariant ensembles \cite{TracyW:1994,TW}, the increasing subsequences of a random permutation \cite{BaikDJ:1999}, last-passage percolation, randomly growing Young diagrams, and vicious random walkers \cite{Johansson:2000,BaikR:2001}, the KPZ growth model \cite{PrahoferS:2000}, and growing interfaces in liquid crystals \cite{TakeuchiS:2010}.  Specifically, consider the three classical random matrix ensembles GOE ($\beta=1$), GUE ($\beta=2$) and GSE ($\beta=4$), i.e. we choose a Hermitian matrix {\bf X} with real ($\beta=1$), complex ($\beta=2$), or real quaternion ($\beta=4$) entries and underlying eigenvalue probability density function of the form
\begin{equation*}
	\frac{1}{Z_{N,\beta}}\prod_{1\leq j<k\leq N}|\lambda_k-\lambda_j|^{\beta}\exp\left(-\frac{\beta}{2}\sum_{j=1}^N\lambda_j^2\right)
\end{equation*}
with a normalization constant $Z_{N,\beta}$. The Tracy-Widom functions $F_{\beta}(s)$ are the distribution functions of the (properly centered and scaled) largest eigenvalue $\lambda_{\max}({\bf X})$ as the matrix size $N$ tends to infinity:
\begin{equation}\label{TW:form}
	\lim_{N\rightarrow\infty}\mathbb{P}\left(\frac{\lambda_{\max}({\bf X})-\sqrt{2N}}{2^{-\frac{1}{2}}N^{-\frac{1}{6}}}\leq s\right)=F_{\beta}(s),\ \ s\in\mathbb{R},\ \ \ \ \beta=1,2,4.
\end{equation}
Quite remarkably, the three distribution functions admit the representations \cite{TracyW:1994,TW}
\begin{equation*}
	F_2(s)=\exp\left(-\int_s^{\infty}(t-s)u_{_{\textnormal{HM}}}^2(t)\,\d t\right),\ \ \ \big(F_1(s)\big)^2=F_2(s)\e^{-\mu(s)},\ \ \ \big(F_4(s)\big)^2=F_2(s)\cosh^2\left(\frac{1}{2}\mu(s)\right)
\end{equation*}
in terms of the Hastings-McLeod \cite{HM} solution 
$u_{_\textnormal{HM}}=u_{_{\textnormal{HM}}}(x)$ to the Painlev\'e-II equation
\begin{equation*}
	u''_{_{\textnormal{HM}}}=xu_{_{\textnormal{HM}}}+2u^3_{_{\textnormal{HM}}},\ \ \ \ \ (')=\frac{\d}{\d x};\ \ \ \ \ \ \ u_{_{\textnormal{HM}}}(x)=\frac{x^{-\frac{1}{4}}}{2\sqrt{\pi}}\e^{-\frac{2}{3}x^{\frac{3}{2}}}\big(1+o(1)\big),\ \ x\rightarrow+\infty
\end{equation*}
and its antiderivative
\begin{equation*}
	\mu(s):=\int_s^{\infty}u_{_\textnormal{HM}}(t)\,\d t.
\end{equation*}
The distribution functions $F_{\beta}(s)$ are central to modern integrable probability but they are very different from the classical normal distribution;  in particular we note that
\begin{equation*}
	F_{\beta}(s)=1-c_{\beta}s^{-\frac{3\beta}{4}}\e^{-\frac{2\beta}{3}s^{\frac{3}{2}}}\big(1+o(1)\big),\ \ \ s\rightarrow+\infty;\ \ \ \ c_1=\frac{1}{4\sqrt{\pi}},\ c_2=\frac{1}{16\pi},\ c_4=\frac{1}{512\pi},
\end{equation*}
and \cite{TracyW:1994,DeiftIK:2008,BaikBD:2008,BaikBDI:2009,BorotN:2012}, as $s\rightarrow-\infty$,
\begin{align}
	F_1(s)=&\,\,\tau_1(-s)^{-\frac{1}{16}}\e^{\frac{1}{24}s^3-\frac{1}{3\sqrt{2}}(-s)^{\frac{3}{2}}}\big(1+o(1)\big),\ \ \ \ \ \ F_2(s)=\tau_2(-s)^{-\frac{1}{8}}\e^{\frac{1}{12}s^3}\big(1+o(1)\big),\label{F2tail}\\
	F_4(s)=&\,\,\tau_4(-s)^{-\frac{1}{16}}\e^{\frac{1}{24}s^3+\frac{1}{3\sqrt{2}}(-s)^{\frac{3}{2}}}\big(1+o(1)\big);\ \ \ \ \tau_1=2^{-\frac{11}{48}}\e^{\frac{1}{2}\zeta'(-1)},\tau_2=2^{\frac{1}{24}}\e^{\zeta'(-1)},\  \tau_4=2^{-\frac{35}{48}}\e^{\frac{1}{2}\zeta'(-1)}\nonumber
\end{align}
in terms of the Riemann-zeta function $\z(z)$. Our focus lies on the distribution of the largest eigenvalue $\lambda_{\max}({\bf X},\gamma)$ in the following {\it thinned} process: let $\lambda_1({\bf X})<\ldots<\lambda_{\max}({\bf X})$ denote the eigenvalues of a GOE, GUE, or GSE matrix ${\bf X}$ and fix a number $\gamma\in[0,1]$. Now, discard each eigenvalue independently with probability $1-\gamma$ and define $\lambda_{\max}({\bf X},\gamma)$ as the largest observed eigenvalue after thinning. In the large-$N$ limit its distribution function leads to a one-parameter generalization of \eqref{TW:form}:
\begin{equation}\label{TW:formAS}
	F_\beta(s,\gamma):=\lim_{N\rightarrow\infty}\mathbb{P}\left(\frac{\lambda_{\max}({\bf X},\gamma)-\sqrt{2N}}{2^{-\frac{1}{2}}N^{-\frac{1}{6}}}\leq s\right),\ \ \ \ \ s\in\mathbb{R},\ \ \ \beta=1,2,4.
\end{equation}
This again admits a Painlev\'e representation, as proven in Subsection \ref{detproof} below.  
\begin{prop}\label{detformu} Let $u_{_\textnormal{AS}}=u_{_{\textnormal{AS}}}(x,\gamma)$ denote the Ablowitz-Segur \cite{AblowitzS:1976} solution to the Painlev\'e-II equation
\begin{equation}\label{AS:behavior}
	u''_{_{\textnormal{AS}}}=xu_{_{\textnormal{AS}}}+2u^3_{_{\textnormal{AS}}},\ \ \ \ \ (')=\frac{\d}{\d x};\ \ \ \ \ \ \ u_{_{\textnormal{AS}}}(x,\gamma)=\sqrt{\gamma}\,\frac{x^{-\frac{1}{4}}}{2\sqrt{\pi}}\e^{-\frac{2}{3}x^{\frac{3}{2}}}\big(1+o(1)\big),\ \ x\rightarrow+\infty.
\end{equation}
For any $s\in\mathbb{R}$ and $\gamma\in[0,1]$, the deformed Tracy-Widom distribution functions \eqref{TW:formAS} equal
\begin{equation}\label{th:1}
	F_2(s,\gamma)=\exp\left(-\int_s^{\infty}(t-s)u_{_{\textnormal{AS}}}^2(t,\gamma)\,\d t\right),\ \ \ \ \big(F_4(s,\gamma)\big)^2=F_2(s,\gamma)\cosh^2\left(\frac{1}{2}\mu(s,\gamma)\right),
\end{equation}
and
\begin{equation}\label{th:2}
	\big(F_1(s,\gamma)\big)^2=F_2(s,\overline{\gamma})\frac{\gamma-1-\cosh\mu(s,\overline{\gamma})+\sqrt{\overline{\gamma}}\sinh\mu(s,\overline{\gamma})}{\gamma-2},\ \ \ \overline{\gamma}:=2\gamma-\gamma^2\in[0,1],
\end{equation}
where
\begin{equation}
\label{mu-s-gamma}
	\mu(s,\gamma):=\int_s^{\infty}u_{_\textnormal{AS}}(t,\gamma)\,\d t.
\end{equation}
\end{prop}
\noindent The thinning operation is well known in the theory of point processes, see e.g.\ \cite{IPS}, but was studied in a random matrix theory context only recently by Bohigas and Pato 
\cite{BohigasP:2004,BohigasP:2006} and later by Bohigas, de Carvalho, and 
Pato \cite{BohigasCP:2009}.  Their 
motivation was a fundamental question in nuclear physics:  how can one 
recover a missing energy level (i.e.\ eigenvalue) from an otherwise 
complete set of measurements?  As observed by Dyson, if the energy levels
are completely correlated (such as evenly spaced levels, also referred to 
as picket fence statistics), then missing levels can be easily identified, 
while if there is no correlation (Poisson statistics), there is no 
possibility of finding the missing level.  The scattering resonances of 
neutrons scattered off a heavy nucleus are known to obey random matrix 
eigenvalue statistics, which lie in an intermediate correlation regime.  
The thinning process reduces the amount of level repulsion (a correlating 
effect), so the thinned Gaussian ensembles interpolate between moderate correlations when 
$\gamma=1$ (random matrix statistics) and no correlations when $\gamma=0$ 
(classical extreme value statistics).  The extreme value theorem states that 
the limiting distribution of the maximum of a sequence of independent 
identically distributed random variables must be a Gumbel (type I), Fr\'echet 
(type II), or Weibull (type III) distribution, provided the limit exists.  
Since GOE, GUE, and GSE eigenvalues obey the semicircle law \cite{F}, i.e.\ we 
have the following weak convergence for the empirical measure:
\begin{equation*}
	\frac{1}{N}\sum_{j=1}^N\delta_{\lambda_j({\bf X})}(y)\rightharpoonup \frac{1}{\pi}\sqrt{2N-y^2}\,\chi_{(-\sqrt{2N},\sqrt{2N})}(y),\ \ \ \chi_A(y)=\begin{cases}1,&y\in A\\ 0,&y\notin A,\end{cases}
\end{equation*}
it is reasonable to expect type III behavior in the 
$\gamma=0$ limit. Indeed, following the approach of Bohigas, de Carvalho, and Pato \cite{BohigasCP:2009}, we first note that \eqref{TW:form} is equivalent to a gap probability for the rescaled eigenvalues $\mu_j({\bf X}):=2^{\frac{1}{2}}N^{\frac{1}{6}}(\lambda_j({\bf X})-\sqrt{2N})$, i.e.
\begin{equation*}
	F_{\beta}(s)=\lim_{N\rightarrow\infty}\mathbb{P}\big(\#\{j:\ \mu_j({\bf X})\in(s,\infty)\}=0\big).
\end{equation*}
But once we thin out the soft-edge scaled eigenvalues $\mu_j({\bf X})$ the average number of eigenvalues in the interval $(s,\infty)$ is reduced, so we want to scale $s$ with $\gamma$ accordingly in order to keep the level density of the remaining fraction $\gamma$ constant. In order to find the correct scale we note that, as $N\rightarrow\infty$, to leading order
\begin{equation*}
	N=\frac{1}{\pi}\int_{-\sqrt{2N}}^{\sqrt{2N}}\sqrt{2N-y^2}\,\d y\sim \frac{1}{\pi}\int_{-4N^{\frac{2}{3}}}^0\sqrt{-\mu}\,\d\mu,\ \ 
\end{equation*}
so the average number of eigenvalues in $(s,\infty)$ equals
\begin{equation*}
	\frac{1}{\pi}\int_s^{\infty}\sqrt{-\mu}\,\chi_{(-\infty,0)}(\mu)\,\d\mu=\frac{2}{3\pi}(-s)^{\frac{3}{2}},\ s<0,
\end{equation*}
i.e. we should replace $s$ in \eqref{TW:formAS} by $\gamma^{-\frac{2}{3}}s$ in order to keep it invariant. Now, return to \eqref{TW:formAS} and use \eqref{AS:behavior} and \eqref{th:1} to see 
\begin{equation*}
	F_{\beta}(s,\gamma)=1-c_{\beta}\gamma^{\frac{\beta}{2}}s^{-\frac{3\beta}{4}}\e^{-\frac{2\beta}{3}s^{\frac{3}{2}}}\big(1+o(1)\big),\ \ \ s\rightarrow+\infty,
\end{equation*}
which holds uniformly for $\gamma\in[0,1]$. Thus, 
\begin{equation}\label{Weib:1}
	\lim_{\gamma\downarrow 0}F_{\beta}\big(\gamma^{-\frac{2}{3}}s,\gamma\big)=1,\ \ s>0.
\end{equation}
The left ($s\rightarrow-\infty$) tail behavior of $F_{\beta}(s,\gamma)$ is more subtle, and it is here where the derivation in \cite{BohigasCP:2009} becomes non-rigorous: combine
\begin{equation*}
	\frac{\partial}{\partial s}\ln F_2(s,\gamma)=\int_s^{\infty}u^2_{_{\textnormal{AS}}}(x,\gamma)\,\d x=\big(u_{_{\textnormal{AS}}}'(s,\gamma)\big)^2-su_{_\textnormal{AS}}^2(s,\gamma)-u_{_\textnormal{AS}}^4(s,\gamma)
\end{equation*}
with the known asymptotic expansion \cite{AblowitzS:1976,AS2}
\begin{equation*}
	u_{_\textnormal{AS}}(x,\gamma)=(-x)^{-\frac{1}{4}}\sqrt{\frac{v}{\pi}}\cos\left(\frac{2}{3}(-x)^{\frac{3}{2}}-\frac{v}{2\pi}\ln\big(8(-x)^{\frac{3}{2}}\big)+\phi\right)+\mathcal{O}\left((-x)^{-1}\right),\ \ x\rightarrow-\infty,
\end{equation*}
valid for fixed $v=-\ln(1-\gamma)\in[0,\infty)$ where $\phi=\frac{\pi}{4}-\textnormal{arg}\,\Gamma(\frac{v}{2\pi\im})$. Thus, back in \eqref{th:1} and \eqref{th:2} we have
\eq
\label{F2gamma-conjecture}
F_{\beta}(s,\gamma) = \tau_{\beta}(\gamma)(-s)^{\frac{3v^2}{4\beta\pi^2}}\exp\left(-vg_{\beta}(-s)^{\frac{3}{2}}\right)\big(1+o(1)\big),\ \ s\rightarrow-\infty;\ \ \ \ g_1=g_2=\frac{2}{3\pi},\ \ g_4=\frac{1}{3\pi}
\endeq
with undetermined integration constants $\tau_{\beta}(\gamma)$. Ignoring these {\it constant factors} and all error terms we then find from \eqref{F2gamma-conjecture}, together with \eqref{Weib:1},
\eq\label{BDP}
\lim_{\gamma\downarrow 0}F_{\beta}\big(\gamma^{-\frac{2}{3}}s,\gamma\big)=\begin{cases}\exp\left(-g_{\beta}(-s)^{\frac{3}{2}}\right), & s\leq 0 \\ 1, & s>0, \end{cases}
\endeq
which is a simple transformation of the Weibull distribution function with scale parameter $g_{\beta}$ and Weibull slope $\frac{3}{2}$, see Figure \ref{Weibplots} below.
\begin{figure}[tbh]
\includegraphics[width=5.4cm,height=3.9cm]{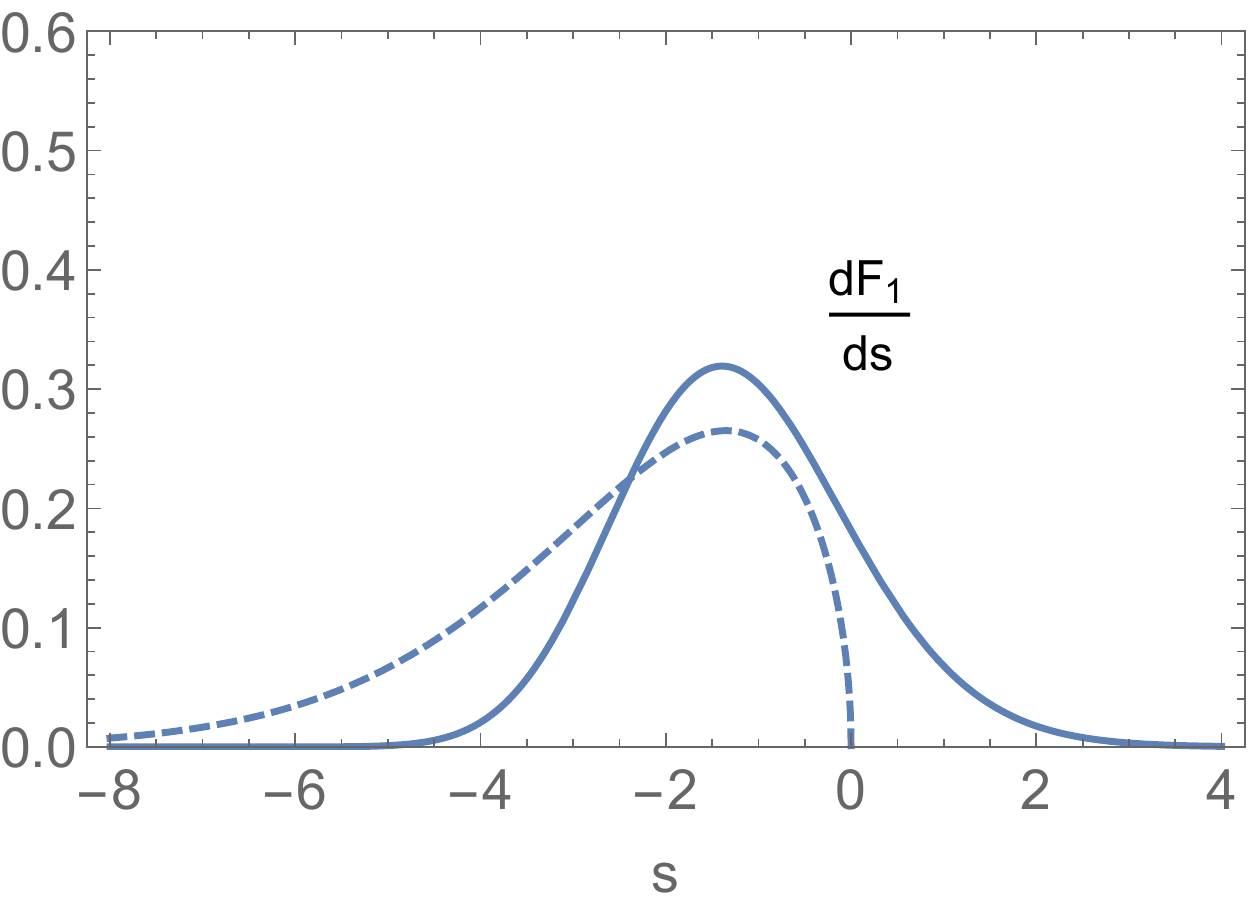}
\includegraphics[width=5.4cm,height=3.9cm]{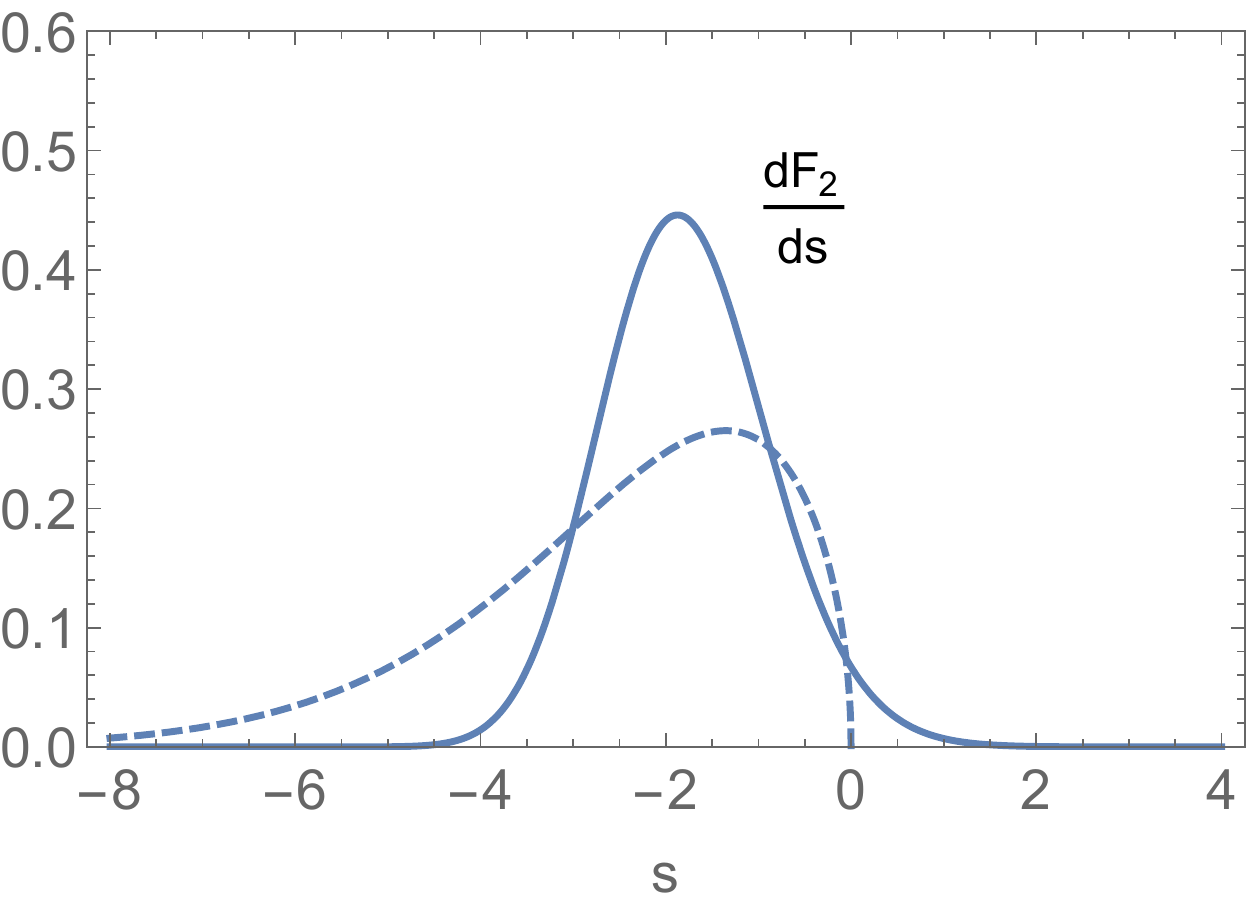}
\includegraphics[width=5.4cm,height=3.9cm]{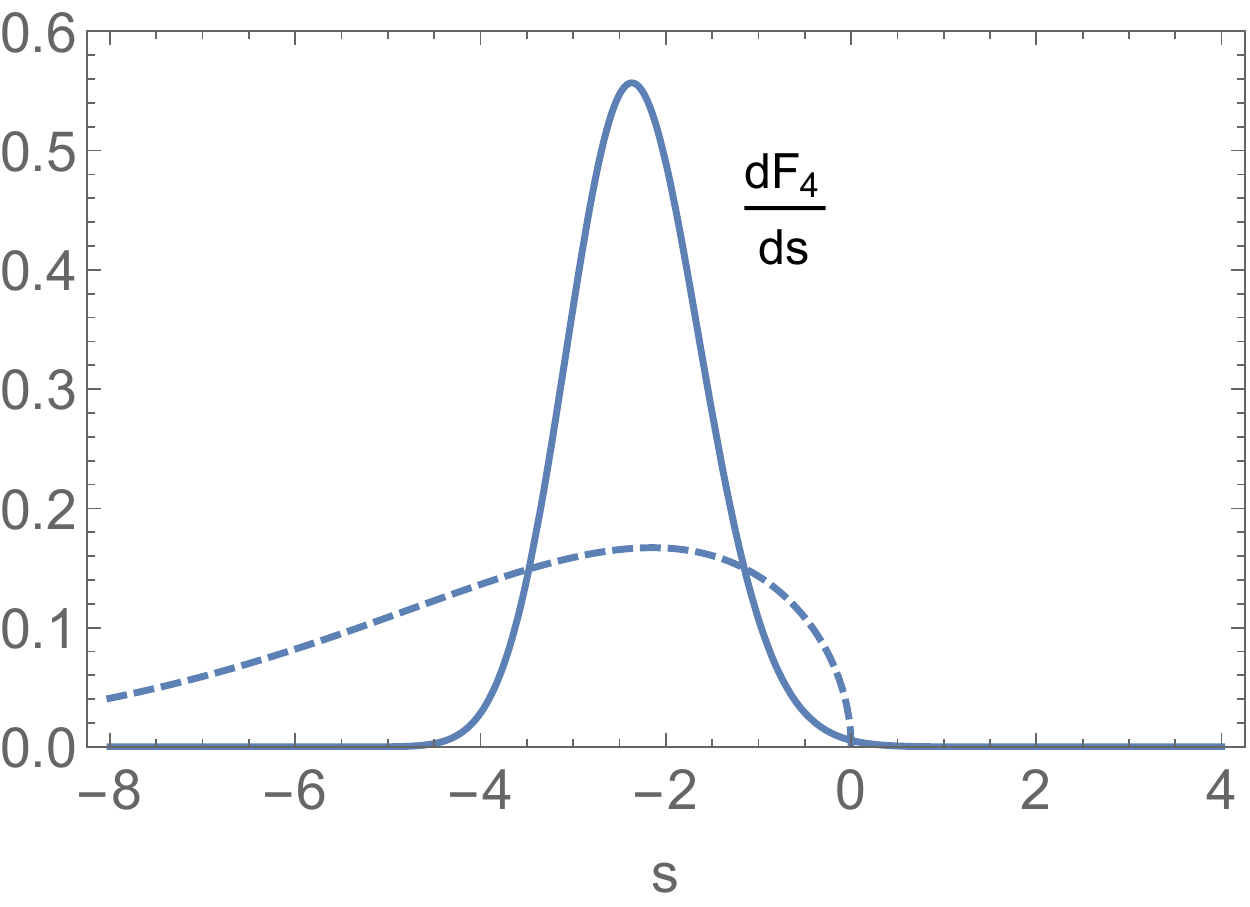}
\caption{The three Tracy-Widom densities $\frac{\d F_{\beta}}{\d s}$ as solid lines for $\beta=1,2,4$ from left to right. The dashed lines are the corresponding limiting Weibull densities from \eqref{BDP}.}
\label{Weibplots}
\end{figure}
We can confirm \eqref{BDP} rigorously using Theorems \ref{Eres:3} and \ref{OStheo} below which are the main results of the paper.
\subsection{Statement of results and discussion}
\begin{theo}\label{Eres:3} For any fixed $0<\epsilon<\frac{1}{2}$, there exist constants $s_0=s_0(\epsilon)>0$ and $c_j=c_j(\epsilon)>0$, $j=1,2$, so that
\begin{equation}\label{Eexp:3}
	\ln F_2(s,\gamma)=-\frac{2v}{3\pi}(-s)^{\frac{3}{2}}+\frac{v^2}{4\pi^2}\ln\big(8(-s)^{\frac{3}{2}}\big)+\ln\left(G\Big(1+\frac{\im v}{2\pi}\Big)G\Big(1-\frac{\im v}{2\pi}\Big)\right)+r_2(s,v)
\end{equation}
for $-s\geq s_0$ and $0\leq v=-\ln(1-\gamma)<(-s)^{\frac{1}{2}-\epsilon}$. Here $G(z)$ is the Barnes $G$-function (see \cite{NIST}), the error term $r_2(s,v)$ is differentiable with respect to $s$, and we have
\begin{equation*}
	\big|r_2(s,v)\big|\leq c_1\frac{v^3}{|s|^{\frac{3}{2}}}+c_2\frac{v}{|s|} \ \ \ \ \ \forall\,(-s)\geq s_0,\ \ 0\leq v<(-s)^{\frac{1}{2}-\epsilon}.
\end{equation*}
\end{theo}
\begin{theo}\label{OStheo} For any fixed $\gamma\in[0,1)$, there exist positive constants $s_0=s_0(v)$ and $c_j=c_j(v)$ such that
\begin{equation*}
	\ln F_1(s,\gamma)=-\frac{2v}{3\pi}(-s)^{\frac{3}{2}}+\frac{v^2}{2\pi^2}\ln\big(8(-s)^{\frac{3}{2}}\big)+\frac{1}{2}\ln\left(G\left(1+\frac{\im v}{\pi}\right)G\left(1-\frac{\im v}{\pi}\right)\right)+\frac{1}{2}\ln\left(\frac{2}{1+\e^{v}}\right)+r_1(s,v)
\end{equation*}
and
\begin{align*}
	\ln F_4(s,\gamma)=&\,-\frac{v}{3\pi}(-s)^{\frac{3}{2}}+\frac{v^2}{8\pi^2}\ln\big(8(-s)^{\frac{3}{2}}\big)+\frac{1}{2}\ln\left(G\left(1+\frac{\im v}{2\pi}\right)G\left(1-\frac{\im v}{2\pi}\right)\right)\\
	&+\ln\left(\frac{1}{2}\left(\frac{1+\sqrt{1-\e^{-v}}}{1-\sqrt{1-\e^{-v}}}\right)^{\frac{1}{4}}+\frac{1}{2}\left(\frac{1-\sqrt{1-\e^{-v}}}{1+\sqrt{1-\e^{-v}}}\right)^{\frac{1}{4}}\right)+r_4(s,v)
\end{align*}
for $-s\geq s_0$ and $0\leq v=-\ln(1-\gamma)<+\infty$. The function $G(z)$ is again the Barnes G-function, the error terms $r_1(s,v)$ and $r_4(s,v)$ are differentiable with respect to $s$, and
\begin{equation*}
	\big|r_j(s,v)\big|\leq\frac{c_j(v)}{|s|^{\frac{3}{4}}}\ \ \ \ \ \forall\,(-s)\geq s_0,\ \ \ j=1,4.
\end{equation*}
\end{theo}
\noindent While the results of Theorem \ref{Eres:3} and \ref{OStheo} settle the deformation from Tracy-Widom $F_1$, $F_2$, and $F_4$ statistics ($\gamma=1$) to classical Weibull statistics \eqref{BDP} ($\gamma=0$) rigorously, they do not capture the full transition regime! In fact, comparing \eqref{BDP} to \eqref{F2tail} we might ask the important question:\smallskip
\begin{quote}
How is the exponential decay $\exp(-vg_{\beta}(-s)^{\frac{3}{2}})$ in \eqref{F2gamma-conjecture} changed to the super-exponential decay $\exp(s^3/24)$ or $\exp(s^3/12)$ in \eqref{F2tail} as $\gamma\uparrow 1$, or equivalently as $v=-\ln(1-\gamma)\rightarrow+\infty$?\smallskip
\end{quote}
This question is the central reason behind the growth condition placed on $v$ in \eqref{Eexp:3}. As long as $v$ does not grow faster as $(-s)^{\frac{1}{2}-\epsilon}$, all leading terms in the left-tail GUE expansion are unchanged from $v\in[0,+\infty)$ fixed, i.e. we are dealing with quasi-Weibull tails. On the other hand, we have
\begin{theo}[\cite{B2}, Theorem $1.4$]
Given $\chi\in\mathbb{R}$, let $p=p(\chi)\in\mathbb{Z}_{\geq 0}$ be such that $p=0$ for $\chi<-\frac{1}{2}$ and $\chi+\frac{1}{2}<p\leq\chi+\frac{3}{2}$ for $\chi\geq -\frac{1}{2}$. Then, as $s\rightarrow-\infty$,
\begin{equation}\label{eigen}
	\ln F_2(s,\gamma)=\frac{s^3}{12}-\frac{1}{4}\ln(-s)+\ln\tau_2+\sum_{j=0}^{p-1}\ln\left(1+\frac{j!}{\sqrt{\pi}}2^{-\frac{7}{2}j-\frac{9}{4}}(-s)^{-\frac{3}{2}j-\frac{3}{4}}\e^{\frac{2}{3}\sqrt{2}(-s)^{\frac{3}{2}}-v}\right)+o(1),
\end{equation}
which holds uniformly for $v\geq\frac{2}{3}\sqrt{2}\,(-s)^{\frac{3}{2}}-\chi\ln\big((-s)^{\frac{3}{2}}\big)$.
\end{theo}
\noindent 
So, once $v$ grows at least as fast as 
$\frac{2}{3}\sqrt{2}(-s)^{\frac{3}{2}}$, we observe already Tracy-Widom 
tails.  Combining \eqref{Eexp:3} and \eqref{eigen}, we, thus far, only have an answer to our question for the GUE in the disjoint regions
\begin{equation*}
	(t,v)\in\mathbb{R}^2:\ \ \ t=(-s)^{\frac{3}{2}}\geq t_0,\ \ \ \ \ \ \varkappa:=\frac{v}{t}\in\left[0,t^{-\frac{1}{3}-\frac{2}{3}\epsilon}\right)\bigsqcup\left[\frac{2}{3}\sqrt{2}-\chi\frac{\ln t}{t},+\infty\right].  
\end{equation*}
See Figure \ref{transition-regions} below.  These regions form the non-oscillatory part of the transition asymptotics of $F_2(s,\gamma)$ as $s\rightarrow-\infty$. Our next step is the analysis in the outstanding parameter domain which is ongoing work \cite{BB2} for the thinned GOE/GUE and GSE Tracy-Widom distributions. There, for faster growing $\gamma$, the Riemann-Hilbert analysis carried out below changes significantly and requires the use of elliptic functions. These elliptic functions degenerate at both ends ($\varkappa\downarrow 0$ and $\varkappa\uparrow\frac{2}{3}\sqrt{2}$) and oscillations vanish: at one end ($\varkappa\downarrow 0$) oscillations die out via decreasing amplitudes and fixed periods, at the other end ($\varkappa\uparrow\frac{2}{3}\sqrt{2}$) via fixed amplitudes and increasing periods.
\begin{figure}[tbh]
\begin{center}
\includegraphics[height=2.2in]{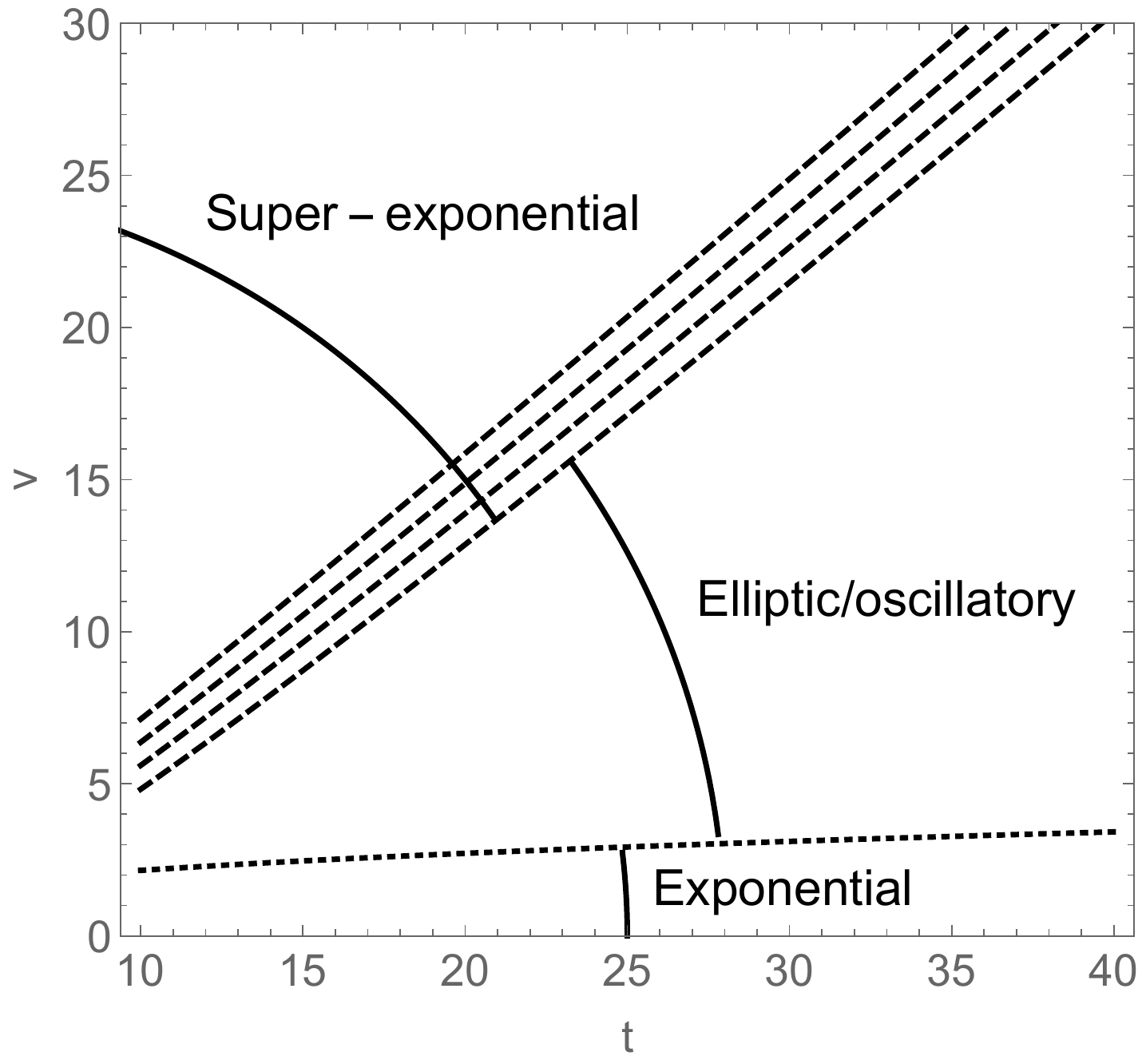}
\caption{Left-tail behavior of $F_2(s,\gamma)$ as $s\to-\infty$, $\gamma\to 1$ in terms of $v=-\ln(1-\gamma)$ and $t=(-s)^{\frac{3}{2}}$.  This work determines the asymptotics in the exponential region (see Theorem \ref{Eres:3}).  For the super-exponential region see Equation \eqref{F2tail}. Dotted line:  $v=t^{\frac{1}{3}}$.  Dashed lines:  $v=\frac{2}{3}\sqrt{2}t-\chi\ln t$ for various values of $\chi$.}
\label{transition-regions}
\end{center}
\end{figure}

\begin{rem}
We emphasize that the appearance of an oscillatory intermediate regime in the left tail expansion of $F_{\beta}(s,\gamma)$ was first observed numerically in the paper by Bohigas, de Carvalho and Pato \cite{BohigasCP:2009}. The similarity between $F_2(s)$ and $F_2(s,\gamma)$ with $\gamma\neq 1$ in \eqref{th:1} suggests an (alternate) approach to understanding these oscillations by analyzing the transition between $u_{_\textnormal{AS}}(x,\gamma)$ and $u_{_\textnormal{HM}}(x)$ as $\gamma\uparrow 1$.  See \cite{B1} for recent progress on this front.
\end{rem}
\noindent At this point we reformulate the result of Theorem \ref{Eres:3} in terms of the underlying Painlev\'e transcendent \eqref{AS:behavior}, which leads us to the following total integral formula.
\begin{cor}\label{cor:1} Let $u_{_{\textnormal{AS}}}(x,\gamma)$ denote the Ablowitz-Segur Painlev\'e-II solution as defined in \eqref{AS:behavior}. Then, for any fixed $\gamma\in[0,1)$,
\begin{equation*}
	\lim_{s\rightarrow-\infty}\left(-\int_s^{\infty}(t-s)u_{_{\textnormal{AS}}}^2(t,\gamma)\,\d t+\frac{2v}{3\pi}(-s)^{\frac{3}{2}}-\frac{3v^2}{8\pi^2}\ln(-s)\right)=\frac{3v^2}{4\pi^2}\ln 2+\ln\left(G\Big(1+\frac{\im v}{2\pi}\Big)G\Big(1-\frac{\im v}{2\pi}\Big)\right).
\end{equation*}
\end{cor}
\begin{rem} In recent work computing asymptotics of orthogonal polynomials with a discontinuous Gaussian weight, Bogatskiy, Claeys, and Its \cite{BogatskiyCI:2016} in fact conjectured the above total integral, or equivalently the thinned GUE constant factor
\begin{equation*}
	\ln\tau_2(\gamma) = \frac{3v^2}{4\pi^2}\ln 2 + \ln\left(G\left(1+\frac{iv}{2\pi}\right)G\left(1-\frac{iv}{2\pi}\right)\right)
\end{equation*}
in \eqref{F2gamma-conjecture} for any fixed $\gamma\in\mathbb{C}\backslash[1,+\infty)$.
\end{rem}
\begin{rem} The total integral formula in Corollary \ref{cor:1} is analogous to another total integral formula for the aforementioned Hastings-McLeod solution to the Painlev\'e-II equation (see \cite{BaikBD:2008}).  For any $c<0$,
\begin{equation*}
\int_c^\infty(t-s)u_{_{\textnormal{HM}}}^2(t)\,\d t + \int_{-\infty}^c\left((t-s)u_{_{\textnormal{HM}}}^2(t)-\frac{1}{4}t^2+\frac{1}{8t}\right)\d t=-\frac{1}{24}\ln 2-\zeta'(-1)-\frac{c^3}{12}+\frac{1}{8}\ln(-c).
\end{equation*}
Computation of the asymptotics for the GOE and GSE Tracy-Widom laws also 
leads to a complete integral for the function $u_{_\textnormal{HM}}(x)$ itself \cite{BaikBD:2008,BaikBDI:2009}.
\end{rem}
\noindent The computation of undetermined constants of integration, e.g.\ $\tau_2$ in \eqref{F2tail} and $\tau_2(\gamma)$ in \eqref{F2gamma-conjecture}, in the asymptotic expansion of tails of distribution functions and gap probabilities is a standard but notoriously difficult problem in integrable probability.  In the interior of the eigenvalue bulk, Dyson \cite{Dyson:1976} conjectured that 
\begin{equation}\label{Dyson}
P(s):=\mathbb{P}\left[\textnormal{No bulk GUE eigenvalues in}\left(-\frac{s}{\pi},\frac{s}{\pi}\right)\right] = c_0s^{-\frac{1}{4}}\e^{-\frac{1}{2}s^2}\left(1+\mathcal{O}\left(s^{-1}\right)\right) \ \ \textnormal{as}\ s\to\infty
\end{equation}
with
\eq\nonumber
\ln c_0 = \frac{1}{12}\ln 2 + 3\zeta'(-1).
\endeq
This expansion was proven up to determining the constant $c_0$ by Widom \cite{Widom:1995}.  The constant was finally proven ten years later by Krasovsky 
\cite{Krasovsky:2004}, with subsequent alternate proofs given by Ehrhardt \cite{Ehrhardt:2006} and Deift, Its, Krasovsky, and Zhou \cite{DeiftIKZ:2007}.  
The corresponding bulk constants for the GOE and GSE were first proven by 
Ehrhardt \cite{Ehrhardt:2007}, with a subsequent alternate proof given by 
Baik, Buckingham, DiFranco, and Its \cite{BaikBDI:2009}.  
Interestingly, in the bulk the analogous result for the thinned GUE process is 
simpler than the result for the non-thinned GUE process 
and follows immediately from work of Basor and Widom \cite{BasorW:1983}, Budylin and Buslaev \cite{BB}, and Bothner, Deift, Its, and Krasovsky \cite{BDIK2}:
\begin{eqnarray}
P(s,\gamma)&:=&\mathbb{P}\left[\textnormal{No bulk GUE eigenvalues in}\left(-\frac{s}{\pi},\frac{s}{\pi}\right)\ \textnormal{after thinning}\right]\nonumber\\
 &=& c_0(\gamma)\left(4s\right)^{\frac{v^2}{2\pi^2}}\e^{-\frac{2v}{\pi}s}\left(1+o(1)\right)\ \ \textnormal{as}\ \ s\rightarrow\infty\ \ \textnormal{with}\ \ 0<v<s^{\frac{1}{3}},\ \label{Poisson}
\end{eqnarray}
where
\begin{equation*}
c_0(\gamma)=\left\{ G\left(1+\frac{iv}{2\pi}\right)G\left(1-\frac{iv}{2\pi}\right)\right\}^2.
\end{equation*}
\begin{rem}
As an immediate consequence of \eqref{Poisson} we emphasize that the thinning operation, when applied to bulk scaled GUE eigenvalues, interpolates between random matrix theory statistics, see \eqref{Dyson}, and a particle system obeying classical Poisson statistics.  Indeed,
\begin{equation*}
	\lim_{\gamma\downarrow 0}P(s\gamma^{-1},s)=\begin{cases}\e^{-\frac{2s}{\pi}},&s\geq 0\\ 0,&s<0,\end{cases}
\end{equation*}
which is the gap probability of a Poisson particle system after removing a fraction $1-\gamma$ of particles. The full transition between \eqref{Dyson} and \eqref{Poisson} was described rigorously in the recent works \cite{BDIK1,BDIK2}. Previously, Dyson \cite{Dys} had already established an oscillatory transition regime for $P(s,\gamma)$ as $s\rightarrow+\infty$ and $\gamma\uparrow 1$ based on a non-rigorous log-gas interpretation of $P(s,\gamma)$.
\end{rem}
\noindent Towards the end of our discussion we would like to mention a few other recent works on thinned random matrix ensembles, here in the context of Haar-distributed random matrices, i.e. circular ensembles. For instance, gap and conditional probabilities for the thinned CUE have been computed via Toeplitz determinants and orthogonal polynomials on the unit circle in \cite{CC}. For all three classical circular ensembles, \cite{BFM} deals with the computation of the spacing distributions in the large-N limit through the use of Fredholm determinants. As an intriguing application of the thinned CUE, the Odlyzko data set of Riemann zeros is analyzed in the paper \cite{FM}, and two-point correlation functions as well as nearest neighbor spacings are computed for the thinned data set. Finally, \cite{BD} is devoted to the analysis of mesoscopic fluctuations in the thinned CUE.
\subsection{Determinantal \texorpdfstring{formul\ae\,for}{formulae for} \texorpdfstring{$F_{\beta}(s,\gamma)$}{F}}
\label{detproof} 
We now prove Proposition \ref{detformu}.
\begin{proof} It is well known (see \cite{F}) that the GUE Tracy-Widom distribution can be written as the Fredholm determinant
\begin{equation}
\label{F2-determinant}
	F_2(s)=\det(1-K_{\textnormal{Ai}}\upharpoonright_{L^2(s,\infty)}),\ \ s\in\mathbb{R},
\end{equation}
where $K_{\textnormal{Ai}}:L^2\big((s,\infty);\d\lambda\big)\circlearrowleft$ denotes the trace-class integral operator
\begin{equation}\label{Airyker}
\big(K_{\textnormal{Ai}}f\big)(\lambda):=\int_s^{\infty}K_{\textnormal{Ai}}(\lambda,\mu)f(\mu)\d\mu;\ \ \ \ K_{\textnormal{Ai}}(\lambda,\mu):=\frac{\textnormal{Ai}(\lambda)\textnormal{Ai}'(\mu)-\textnormal{Ai}'(\lambda)\textnormal{Ai}(\mu)}{\lambda-\mu},\ \ \ \lambda,\mu\in(s,\infty)
\end{equation}
constructed in terms of the classical Airy function $\textnormal{Ai}(z)$. Moreover, the probability that there are exactly $m\in\mathbb{Z}_{\geq 0}$ edge-scaled eigenvalues $\mu_j({\bf X})$ in the interval $(s,\infty)$ in the large-N limit equals \cite{TracyW:1994,D,F2}
\begin{equation}\label{mprob}
	E_{\beta}(m,(s,\infty))=\frac{(-1)^m}{m!}\frac{\partial^m}{\partial\xi^m}E_{\beta}((s,\infty),\xi)\Big|_{\xi=1},
	\end{equation}
using the generating functions
\begin{equation*}
	E_2((s,\infty),\xi)=\det(1-\xi K_{\textnormal{Ai}}\upharpoonright_{L^2(s,\infty)})=\exp\left(-\int_s^{\infty}(t-s)u_{_{\textnormal{AS}}}^2(t,\xi)\,\d t\right),
\end{equation*}
and
\begin{equation*}
	\big(E_4((s,\infty),\xi)\big)^2=E_2((s,\infty),\xi)\cosh^2\left(\frac{1}{2}\mu(s,\xi)\right),
\end{equation*}
as well as
\begin{equation*}
	\big(E_1((s,\infty),\xi)\big)^2=E_2((s,\infty),\overline{\xi})\frac{\xi-1-\cosh\mu(s,\overline{\xi})+\sqrt{\overline{\xi}}\sinh\mu(s,\overline{\xi})}{\xi-2},\ \ \ \ \ \overline{\xi}:=2\xi-\xi^2
\end{equation*}
in terms of the Ablowitz-Segur transcendent $u_{_\textnormal{AS}}(x,\gamma)$ from \eqref{AS:behavior} and its antiderivative $\mu(s,\xi)$ (see \eqref{mu-s-gamma}).  
Hence,
\begin{equation}\label{det:1}
	F_{\beta}(s,\gamma)=\lim_{N\rightarrow\infty}\mathbb{P}(\#\{j:\ \mu_j({\bf X},\gamma)\in(s,\infty)=0)=\sum_{m=0}^{\infty}E_{\beta}(m,(s,\infty))(1-\gamma)^m,
\end{equation}
since each edge-scaled eigenvalue $\mu_j({\bf X})$ is removed independently with probability $1-\gamma$. Substituting \eqref{mprob} into \eqref{det:1}, we then find
\begin{equation*}
	F_{\beta}(s,\gamma)=\sum_{m=0}^{\infty}\frac{1}{m!}\frac{\partial^m}{\partial\xi^m}E_{\beta}((s,\infty),\xi)\Big|_{\xi=1}(\gamma-1)^m=E_{\beta}((s,\infty),\gamma)
\end{equation*}
by Taylor's theorem.
\end{proof}
\begin{rem} The $\gamma$-deformed distribution functions $F_{\beta}(s,\gamma)$ are of interest in their own right aside from 
the thinned processes. They appear, for instance, in the formula for the distribution $F_{\beta}(s|m)$ of 
the $m^\text{th}$ largest eigenvalue \cite{TracyW:1994,D}:
\begin{equation*}
	F_2(s|m+1)-F_2(s|m)=\frac{(-1)^m}{m!}\frac{\partial^m}{\partial\xi^m}F_2(s,\xi)\Big|_{\xi=1},\ F_{\beta}(s|m+1)-F_{\beta}(s|m)=\frac{(-1)^m}{m!}\frac{\partial^m}{\partial\xi^m}\big(F_{\beta}(s,\xi)\big)^{\frac{1}{2}}\Big|_{\xi=1}
\end{equation*}
with $\beta=1,4$ in the last case and in general $F_{\beta}(s|0)=0$.
Similarly, the function $F_2(s,\gamma)$ appears in determinantal form in 
the law for the length of the second row of a random Young diagram 
\cite{BaikDJ:2000}.
\end{rem}
\noindent As mentioned earlier, the Tracy-Widom distributions are 
\emph{universal} in the sense that they describe the limiting behavior of 
a wide variety of seemingly unrelated processes, much like the normal, 
Gumbel, Fr\'echet, and Weibull distributions of classical probability.  
The thinning process makes sense for most processes with extremal statistics 
given by Tracy-Widom laws, and our results also apply. For one example, let $Y_n$ denote the set of all Young diagrams of size $n$.  
For $\mu\in Y_n$, define Plancherel measure to be 
\eq
\mathbb{P}(\mu) = \frac{d_\mu^2}{n!},
\endeq
where $d_\mu$ is the number of standard Young tableaux (filled Young 
diagrams) with shape $\mu$.  Also define $\ell_n^{(1)}(\mu)$ as the number 
of boxes in the first row of $\mu$.  Baik, Deift, and Johansson 
\cite{BaikDJ:1999} showed that 
\eq
\lim_{n\to\infty}\mathbb{P}\left(\frac{\ell_n^{(1)}-2\sqrt{n}}{n^{\frac{1}{6}}}\leq s\right) = F_2(s).
\endeq
Now remove each row of $\mu$ independently with probability $1-\gamma$, thus 
obtaining a possibly different diagram.  Let $\ell_n^{(1,\gamma)}$ 
be the length of the longest observed row.  Then 
\eq
\lim_{n\to\infty}\mathbb{P}\left(\frac{\ell_n^{(1,\gamma)}-2\sqrt{n}}{n^{\frac{1}{6}}}\leq s\right) = F_2(s,\gamma),
\endeq
and so our asymptotic results apply to this problem as well.
\subsection{Overview and outline}
The overall strategy is straightforward. The Airy kernel \eqref{Airyker} displays a well-known integrable structure \cite{IIKS} which allows us to analyze the tail asymptotics with a nonlinear steepest-descent Riemann-Hilbert approach \cite{DZ}. In more detail, we compute 
an asymptotic expansion of $\partial_{\gamma}\ln F_2(s,\gamma)$ 
(see equation \eqref{partial-gamma-F2} below) and integrate this expansion definitely with respect to $\gamma$ 
(see \eqref{F2final}). This way, knowing that $F_2(s,0)=1$ for $s\in\mathbb{R}$, we are able to determine the integration constant $\tau_2(\gamma)$ explicitly. Central to this chosen approach is a local identity for $\partial_{\gamma}\ln F_2(s,\gamma)$ in terms of the solution of the Riemann-Hilbert Problem \ref{master} which we derive in Proposition \ref{diff2}. We also use a simpler identity for $\partial_s\ln F_2(s,\gamma)$ to double check our previous computations.\smallskip

All necessary steps in the nonlinear steepest descent analysis for fixed 
$\gamma\in[0,1)$ are carried out in Section \ref{sec:rhp-analysis}.  We 
introduce the basic Riemann-Hilbert problem in Subsection \ref{sec:rhp-def}
and first carry out a series of changes of variables in Subsection 
\ref{sec:rhp-prep} to arrive at a Riemann-Hilbert problem with constant jumps 
on four rays emanating from the origin.  The $g$-function, a standard 
technique for regularizing Riemann-Hilbert problems, is introduced afterwards 
in Subsection \ref{sec:rhp-gfun}, and lenses are opened in Subsection 
\ref{sec:rhp-lenses} to ensure the jump matrices decay to the identity except 
on a single band.  In Subsection \ref{sec:rhp-model} we solve the outer model 
problem that results from discarding all decaying jumps. The solution to this 
problem is expected to be a good approximation of the true solution to the 
Riemann-Hilbert problem, except near the two band endpoints at which the jump 
matrices decay sub-exponentially.  This necessitates the construction of two 
local parametrices near these two endpoints.  The Riemann-Hilbert analysis 
concludes in Subsection \ref{sec:rhp-error} by controlling the error of our 
approximate solution.\smallskip

\noindent In Section \ref{diffidentity} we show how 
$\partial_{\gamma}\ln F_2(s,\gamma)$ and $\partial_s\ln F_2(s,\gamma)$
can be constructed from the solution of the Riemann-Hilbert problem.  The 
explicit asymptotic expansion for $\partial_s\ln F_2(s,\gamma)$ is computed in 
Section \ref{sderivsec}, along with its indefinite integral with respect to 
$s$.  The explicit formula for $\partial_{\gamma}\ln F_2(s,\gamma)$ is found 
in Section \ref{gamma-deriv-sec}, which establishes Theorem \ref{Eres:3} for 
fixed $\gamma\in[0,1)$.  Section \ref{sec:extension} is dedicated to extending 
the result so $\gamma$ can approach $1$ at a controlled rate.  The major 
technical difference in this case is we must work in shrinking neighborhoods 
of the band endpoints when we build the local parametrices. Finally, in 
Section \ref{OSsec} we prove Theorem \ref{OStheo} for the GOE and GSE 
distributions by revisiting and extending the computations in 
\cite{BaikBDI:2009} on the total integral formula for the Ablowitz-Segur 
transcendent \eqref{AS:behavior}.

\section{The Riemann-Hilbert problem (RHP) and nonlinear steepest descent analysis}
\label{sec:rhp-analysis}

\subsection{The basic RHP}
\label{sec:rhp-def}

We begin with the following RHP which is central to the 
integrable structure of the Airy kernel (see \cite{IIKS}).  We will show in 
Propositions \ref{diff1} and \ref{diff2} how 
$\frac{\partial}{\partial s}\ln F_2(s,\gamma)$ and 
$\frac{\partial}{\partial \gamma}\ln F_2(s,\gamma)$ can be written in terms of 
the solution ${\bf Y}(z)$ and a function ${\bf X}(z)$ satisfying the 
equivalent RHP \ref{unifRHP}.  
\begin{problem}\label{master} Determine ${\bf Y}(z)={\bf Y}(z;s,\gamma)\in\mathbb{C}^{2\times 2}$, a matrix-valued piecewise analytic function uniquely characterized by the following four properties.
\begin{enumerate}
	\item ${\bf Y}(z)$ is analytic for $z\in\mathbb{C}\backslash[s,\infty)$ with $s\in\mathbb{R}$. We orient $[s,\infty)\subset\mathbb{R}$ from left to right.
	\item The limiting values ${\bf Y}_{\pm}(z)=\lim_{\varepsilon\downarrow 0}{\bf Y}(z\pm\im\varepsilon)$ from either side of the cut $[s,\infty)$ are square integrable and related via the jump condition
	\begin{equation*}
		{\bf Y}_+(z)={\bf Y}_-(z)\begin{pmatrix}
		1-2\pi\im\gamma\textnormal{Ai}(z)\textnormal{Ai}'(z) & 2\pi\im\gamma\textnormal{Ai}^2(z)\\
		-2\pi\im\gamma\big(\textnormal{Ai}'(z)\big)^2 & 1+2\pi\im\gamma\textnormal{Ai}(z)\textnormal{Ai}'(z)
		\end{pmatrix},\ \ \ z\in(s,\infty).
	\end{equation*}
	\item Near the endpoint $z=s$, 
	\begin{equation*}
		{\bf Y}(z)=\mathcal{O}\big(\ln|z-s|\big),\ \ \ z\rightarrow s.
	\end{equation*}
	\item As $z\rightarrow\infty$, in a full vicinity of infinity,
	\begin{equation}\label{masterASY}
		{\bf Y}(z)={\bf I}+{\bf Y}_1z^{-1}+{\bf Y}_2z^{-2}+\mathcal{O}\left(z^{-3}\right),\ \ \ \ \ {\bf Y}_{\ell}=\big(Y_{\ell}^{jk}\big)_{j,k=1}^2\in\mathbb{C}^{2\times 2}.
	\end{equation}
\end{enumerate}
\end{problem}
In order to solve this problem asymptotically, we follow the Deift-Zhou nonlinear steepest descent roadmap \cite{DZ} and carry out several explicit and invertible transformations.

\subsection{Preliminary steps}
\label{sec:rhp-prep}
Our first simplification of RHP \ref{master} is the following ``undressing transformation'' also used in \cite{CIK,B2}. Consider the entire, unimodular function
\begin{equation}\label{e:2}
	\Phi_0(\z):=\sqrt{2\pi}\,\e^{-\im\frac{\pi}{4}}\begin{pmatrix}
	\textnormal{Ai}(\z) & \e^{\im\frac{\pi}{3}}\textnormal{Ai}\left(\e^{-\im\frac{2\pi}{3}}\z\right)\\
	\textnormal{Ai}'(\z) & \e^{-\im\frac{\pi}{3}}\textnormal{Ai}'\left(\e^{-\im\frac{2\pi}{3}}\z\right)
	\end{pmatrix},\ \ \ \ \ \z\in\mathbb{C}
\end{equation}
and define an {\it Airy parametrix},
\begin{figure}[tbh]
\begin{center}
\resizebox{0.35\textwidth}{!}{\includegraphics{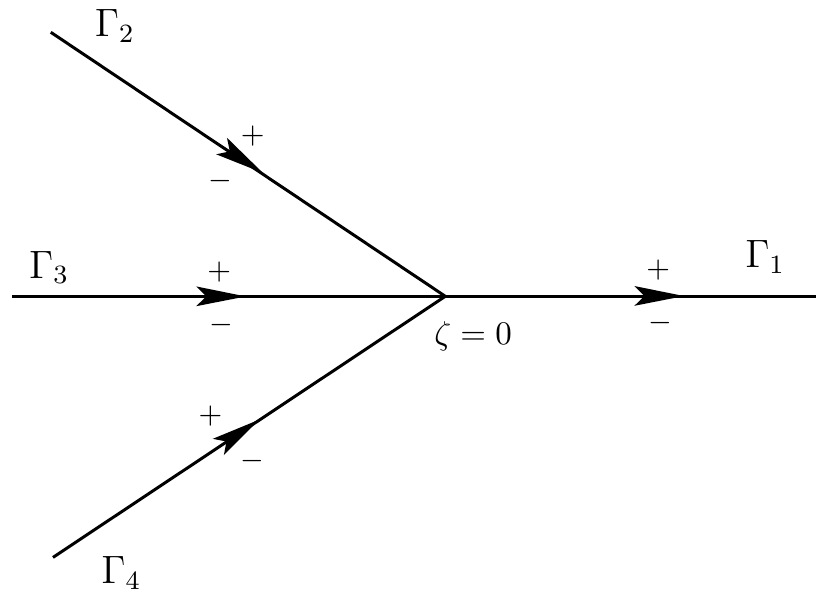}}
\caption{The oriented jump contours for the Airy parametrix $\Phi(\z)$ in the complex $\z$-plane.}
\label{figure1}
\end{center}
\end{figure}
\begin{equation}\label{e:3}
	\Phi(\z):=\Phi_0(\z)\begin{cases}
	{\bf I},&\textnormal{arg}\,\z\in(0,\frac{2\pi}{3})\\
	\bigl(\begin{smallmatrix}
	1 & 0\\
	-1 & 1
	\end{smallmatrix}\bigr),&\textnormal{arg}\,\z\in(\frac{2\pi}{3},\pi)\smallskip\\
	\bigl(\begin{smallmatrix}
	1 & -1\\
	0 & 1
	\end{smallmatrix}\bigr),&\textnormal{arg}\,\z\in(-\frac{2\pi}{3},0)\smallskip\\
	\bigl(\begin{smallmatrix}
	1 & -1\\
	0 & 1
	\end{smallmatrix}\bigr)\bigl(\begin{smallmatrix}
	1 & 0\\
	1 & 1
	\end{smallmatrix}\bigr),&\textnormal{arg}\,\z\in(-\pi,-\frac{2\pi}{3}).
	\end{cases}
\end{equation}
This matrix-valued function solves a well-known model problem:
\begin{problem}\label{Airypara} The Airy parametrix $\Phi(\z)$ has the following properties.
\begin{enumerate}
	\item $\Phi(\z)$ is analytic for $\z\in\mathbb{C}\backslash\bigcup_{j=1}^4\Gamma_j$ with 
	\begin{equation*}
		\Gamma_1:=[0,\infty),\ \ \ \Gamma_2:=\e^{-\im\frac{\pi}{3}}(-\infty,0],\ \ \ \Gamma_3:=(-\infty,0],\ \ \ \Gamma_4:=\e^{\im\frac{\pi}{3}}(-\infty,0]
	\end{equation*}
	and we orient all four rays ``from left to right" as shown in Figure \ref{figure1}.
	\item On the jump contours the limiting values obey the jump conditions
	\begin{equation*}
\begin{split}
		 \Phi_+(\z)&=\Phi_-(\z)\bigl(\begin{smallmatrix}
		1 & 1\\
		0 & 1
		\end{smallmatrix}\bigr),\ \ \ \ \z\in\Gamma_1;\\
		\Phi_+(\z)&=\Phi_-(\z)\bigl(\begin{smallmatrix}
		1 & 0\\
		1 & 1
		\end{smallmatrix}\bigr),\ \ \ \ \z\in\Gamma_2\cup\Gamma_4;\\
		\Phi_+(\z)&=\Phi_-(\z)\bigl(\begin{smallmatrix}
		0 & 1\\
		-1 & 0
		\end{smallmatrix}\bigr),\ \ \z\in\Gamma_3.
\end{split}
	\end{equation*}
	\item $\Phi(\z)$ is bounded at $\z=0$.
	\item As $\z\rightarrow\infty,\z\notin\bigcup_{j=1}^4\Gamma_j$,
	\begin{equation*}
		\Phi(\z)=\z^{-\frac{1}{4}\sigma_3}\frac{1}{\sqrt{2}}\begin{pmatrix}
		1 & 1\\
		-1 & 1
		\end{pmatrix}\e^{-\im\frac{\pi}{4}\sigma_3}\left\{{\bf I}+\frac{1}{48\z^{\frac{3}{2}}}\begin{pmatrix}
		1 & 6\im\\
		6\im & -1
		\end{pmatrix}+\mathcal{O}\left(\z^{-3}\right)\right\}\e^{-\frac{2}{3}\z^{\frac{3}{2}}\sigma_3},
	\end{equation*}
	with $\z^{\alpha}$ defined and analytic for $\z\in\mathbb{C}\backslash(-\infty,0]$ such that $\z^{\alpha}>0$ for $\z>0$.
\end{enumerate}
\end{problem}
\begin{remark}\label{gauge} Conditions (1)--(4) in RHP \ref{Airypara} characterize $\Phi(\z)$ uniquely up to left multiplication with a lower triangular $\z$-independent matrix,
\begin{equation*}
	\eta\in\mathbb{C}:\ \ \ \begin{pmatrix}
	1 & 0\\
	\eta & 1
	\end{pmatrix}\Phi(\z)=\z^{-\frac{1}{4}\sigma_3}\frac{1}{\sqrt{2}}\begin{pmatrix}
	1 & 1\\
	-1 & 1
	\end{pmatrix}\e^{-\im\frac{\pi}{4}\sigma_3}\left\{{\bf I}+\mathcal{O}\left(\z^{-\frac{1}{2}}\right)\right\}\e^{-\frac{2}{3}\z^{\frac{3}{2}}\sigma_3},\ \ \z\rightarrow\infty.
\end{equation*}
\end{remark}
\begin{remark}\label{AiryConn} The jump matrices in RHP \ref{Airypara} satisfy the cyclic relation
\begin{equation*}
  \bigl(\begin{smallmatrix}
         1 & -1\\ 0 & 1
        \end{smallmatrix}\bigr)\bigl(\begin{smallmatrix}
         1 & 0\\ 1 & 1
        \end{smallmatrix}\bigr)\bigl(\begin{smallmatrix}
         0 & 1\\ -1 & 0
        \end{smallmatrix}\bigr)\bigl(\begin{smallmatrix}
         1 & 0\\ 1 & 1
        \end{smallmatrix}\bigr)={\bf I},
\end{equation*} 
which paraphrases in particular that the matrix entries $\Phi^{11}(\z)$ and $\Phi^{21}(\z)$ are entire functions, i.e.\ they
admit analytic extensions from the sector $\textnormal{arg}\,\z\in(0,\frac{2\pi}{3})$ to the full complex plane. This observation allows
us to write the Airy kernel solely in terms of RHP \ref{Airypara},
\begin{equation*}
  K_{\textnormal{Ai}}(\lambda,\mu)=\frac{\im}{2\pi}\frac{\Phi^{11}(\lambda)\Phi^{21}(\mu)-\Phi^{11}(\mu)\Phi^{21}(\lambda)}{\lambda-\mu},
\end{equation*}
a definition that is independent of the gauge transformation outlined in Remark \ref{gauge}.
\end{remark}
\begin{rem} We record the following factorization property of the jump matrix in RHP \ref{master} using \eqref{e:2}. This factorization is at the heart of the upcoming transformation leading to RHP \ref{unifRHP} below.
\begin{equation*}
	z\in\mathbb{C}:\ \ \ \ \ \Phi_0^{-1}(z)\begin{pmatrix}1-2\pi\im\gamma\textnormal{Ai}(z)\textnormal{Ai}'(z)&2\pi\im\gamma \textnormal{Ai}^2(z)\\ -2\pi\im\gamma\big(\textnormal{Ai}'(z)\big)^2& 1+2\pi\im\gamma\textnormal{Ai}(z)\textnormal{Ai}'(z)\end{pmatrix}\Phi_0(z)=\begin{pmatrix}1&-\gamma\\ 0&1\end{pmatrix}.
\end{equation*}
\end{rem}
We now undress RHP \ref{master} and reduce it to a problem with constant 
jumps.  For $s<0$, define (see Figure \ref{figure2})
\begin{equation*}
	{\bf X}(z):={\bf Y}(z)\Phi(z)\begin{cases}
	I,&z\in\Omega_1\cup\Omega_2\cup\Omega_3\\
	\bigl(\begin{smallmatrix}
	1 & 0\\
	1 & 1
	\end{smallmatrix}\bigr),&z\in\Omega_4\smallskip\\
	\bigl(\begin{smallmatrix}
	1 & 0\\
	-1 & 1
	\end{smallmatrix}\bigr),&z\in\Omega_5
	\end{cases}
\end{equation*}
and in case $s>0$ (see Figure \ref{figure3}), set
\begin{equation*}
	{\bf X}(z):={\bf Y}(z)\Phi(z)\begin{cases}
	I,&z\in\Omega_1\cup\Omega_2\cup\Omega_3\\
	\bigl(\begin{smallmatrix}
	1 & 0\\
	-1 & 1
	\end{smallmatrix}\bigr),&z\in\Omega_4\smallskip\\
	\bigl(\begin{smallmatrix}
	1 & 0\\
	1 & 1
	\end{smallmatrix}\bigr),&z\in\Omega_5.
	\end{cases}
\end{equation*}
\begin{figure}[tbh]
\begin{minipage}{0.4\textwidth} 
\begin{center}
\resizebox{0.85\textwidth}{!}{\includegraphics{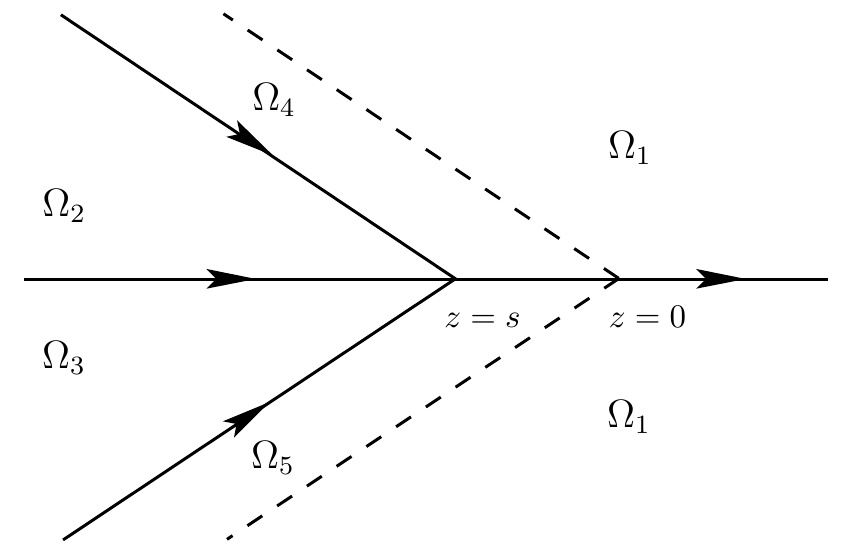}}
\caption{``Undressing" of RHP \ref{master} for $s<0$ with jump contours of ${\bf X}(z)$ as solid lines.}
\label{figure2}
\end{center}
\end{minipage}
\begin{minipage}{0.4\textwidth}
\begin{center}
\resizebox{0.8\textwidth}{!}{\includegraphics{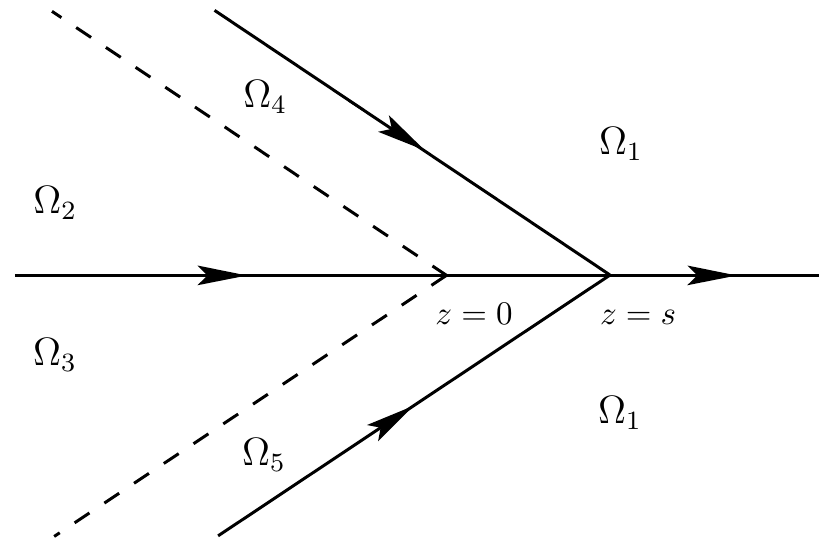}}
\caption{``Undressing" of RHP \ref{master} for $s>0$ with jump contours of ${\bf X}(z)$ as solid lines.}
\label{figure3}
\end{center}
\end{minipage}
\end{figure}

\begin{problem}\label{unifRHP} Determine ${\bf X}(z)\in\mathbb{C}^{2\times 2}$ such that
\begin{enumerate}
	\item ${\bf X}(z)$ is analytic for $z\in\mathbb{C}\backslash\bigcup_{j=1}^4\Gamma_j^{(s)}$ with
	\begin{equation*}
		\Gamma_1^{(s)}:=(s,\infty),\ \ \ \Gamma_2^{(s)}:=s+\e^{-\im\frac{\pi}{3}}(-\infty,0),\ \ \ \ \Gamma_3^{(s)}:=(-\infty,s),\ \ \ \ \Gamma_4^{(s)}:=s+\e^{\im\frac{\pi}{3}}(-\infty,0)
	\end{equation*}
	and the rays $\Gamma_j^{(s)}$ are shown in Figures \ref{figure2} and \ref{figure3}.
	\item The following jump conditions hold true:
	\begin{equation*}
\begin{split}
		{\bf X}_+(z)&={\bf X}_-(z)\bigl(\begin{smallmatrix}
		1 & 1-\gamma\\
		0 & 1
		\end{smallmatrix}\bigr),\ \ z\in\Gamma_1^{(s)};\\
                {\bf X}_+(z)&={\bf X}_-(z)\bigl(\begin{smallmatrix}
		1 & 0\\
		1 & 1
		\end{smallmatrix}\bigr),\ \ \ \ \ z\in\Gamma_2^{(s)}\cup\Gamma_4^{(s)};\\
		{\bf X}_+(z)&={\bf X}_-(z)\bigl(\begin{smallmatrix}
		0 & 1\\
		-1 & 0
		\end{smallmatrix}\bigr),\ \ \ \, z\in\Gamma_3^{(s)}.
\end{split}
	\end{equation*}
	\item Near $z=s$,
	\begin{equation}\label{Xsing}
		{\bf X}(z)=\widehat{{\bf X}}(z)\begin{pmatrix}
		1 & \frac{\gamma}{2\pi\im}\ln(z-s)\\
		0 & 1
		\end{pmatrix}\begin{cases}
		{\bf I},&\textnormal{arg}(z-s)\in(0,\frac{2\pi}{3})\\
		\bigl(\begin{smallmatrix}
		1 & 0\\
		-1 & 1
		\end{smallmatrix}\bigr),&\textnormal{arg}(z-s)\in(\frac{2\pi}{3},\pi)\\
		\bigl(\begin{smallmatrix}
		1 & -1\\
		0 & 1
		\end{smallmatrix}\bigr)\bigl(\begin{smallmatrix}
		1 & 0\\
		1 & 1
		\end{smallmatrix}\bigr),&\textnormal{arg}(z-s)\in(\pi,\frac{4\pi}{3})\\
		\bigl(\begin{smallmatrix}
		1 & -1\\
		0 & 1
		\end{smallmatrix}\bigr),&\textnormal{arg}(z-s)\in(\frac{4\pi}{3},2\pi)
		\end{cases}
	\end{equation}
	where $\widehat{{\bf X}}(z)$ is analytic at $z=s$ and we fix the branch of the logarithm with $\textnormal{arg}(z-s)\in(0,2\pi)$.
	\item As $z\rightarrow\infty$,
	\begin{equation*}
		{\bf X}(z)=z^{-\frac{1}{4}\sigma_3}\frac{1}{\sqrt{2}}\begin{pmatrix}
		1 & 1\\
		-1 & 1\\
		\end{pmatrix}\e^{-\im\frac{\pi}{4}\sigma_3}\left\{{\bf I}+{\bf X}_1z^{-\frac{1}{2}}+{\bf X}_2z^{-1}+\mathcal{O}\left(z^{-\frac{3}{2}}\right)\right\}\e^{-\frac{2}{3}z^{\frac{3}{2}}\sigma_3},
	\end{equation*}
	again with principal branches for all fractional exponents. The matrices ${\bf X}_1$ and ${\bf X}_2$ are $z$-independent, with (see \eqref{masterASY})
\begin{equation*}
  {\bf X}_1=\frac{1}{2}\begin{pmatrix}
                                 -1 & \im \\ \im  & 1
                                \end{pmatrix}Y_1^{12};\ \ \ \ {\bf X}_2=\frac{1}{2}\begin{pmatrix}
                                0 & \im\\ -\im & 0
                                \end{pmatrix}\big(Y_1^{11}-Y_1^{22}\big).
\end{equation*}
\end{enumerate}
\end{problem}
Assume from now on that $s<0$ is negative and define
\begin{equation}\label{Tdef}
	{\bf T}(z):={\bf X}(|s|z+s),\ \ z\in\mathbb{C}\backslash(\Sigma_T\cup\{0\}).
\end{equation}
This transformation ``centers" the problem at the origin $z=0$, so we have jumps on the contour
\begin{equation*}
	\Sigma_T:=\bigcup_{j=1}^4\Gamma_j
\end{equation*}
shown in Figure \ref{figure1}. In more detail
\begin{problem}\label{centerRHP} Determine a function ${\bf T}(z)={\bf T}(z;s,\gamma)\in\mathbb{C}^{2\times 2}$ uniquely characterized by the following properties:
\begin{enumerate}
	\item ${\bf T}(z)$ is analytic for $z\in\mathbb{C}\backslash(\Sigma_T\cup\{0\})$.
	\item ${\bf T}(z)$ has the jumps
\begin{equation*}
\begin{split}
	{\bf T}_+(z)&={\bf T}_-(z)\bigl(\begin{smallmatrix}
	1 & 1-\gamma\\
	0 & 1
	\end{smallmatrix}\bigr),\ \, z\in\Gamma_1\backslash\{0\};\\
	{\bf T}_+(z)&={\bf T}_-(z)\bigl(\begin{smallmatrix}
	1 & 0\\
	1 & 1
	\end{smallmatrix}\bigr),\ \ \ \ \, z\in(\Gamma_2\cup\Gamma_4)\backslash\{0\};\\
        {\bf T}_+(z)&={\bf T}_-(z)\bigl(\begin{smallmatrix}
	0 & 1\\
	-1 & 0
	\end{smallmatrix}\bigr),\ \ \ z\in\Gamma_3\backslash\{0\}.
\end{split}
\end{equation*}
	\item Near $z=0$,
	\begin{equation*}
		{\bf T}(z)=\widehat{{\bf T}}(z)\begin{pmatrix}
		1 & \frac{\gamma}{2\pi\im}\ln z\\
		0 & 1
		\end{pmatrix}\begin{cases}
		{\bf I},&\textnormal{arg}\,z\in(0,\frac{2\pi}{3})\\
		\bigl(\begin{smallmatrix}
		1 & 0\\
		-1 & 1
		\end{smallmatrix}\bigr),&\textnormal{arg}\,z\in(\frac{2\pi}{3},\pi)\\
		\bigl(\begin{smallmatrix}
		1 & -1\\
		0 & 1
		\end{smallmatrix}\bigr)\bigl(\begin{smallmatrix}
		1 & 0\\
		1 & 1
		\end{smallmatrix}\bigr),&\textnormal{arg}\,z\in(\pi,\frac{4\pi}{3})\\
		\bigl(\begin{smallmatrix}
		1 & -1\\
		0 & 1
		\end{smallmatrix}\bigr),&\textnormal{arg}\,z\in(\frac{4\pi}{3},2\pi)\\
		\end{cases},
	\end{equation*}
	where $\widehat{{\bf T}}(z)$ is analytic at $z=0$ and $\textnormal{arg}\,z\in(0,2\pi)$.
	\item As $z\rightarrow\infty$,
\begin{equation*}
	{\bf T}(z)=\big(|s|z\big)^{-\frac{1}{4}\sigma_3}\frac{1}{\sqrt{2}}\begin{pmatrix}
	1 & 1\\
	-1 & 1
	\end{pmatrix}\e^{-\im\frac{\pi}{4}\sigma_3}\left\{{\bf I}+{\bf X}_1(|s|z)^{-\frac{1}{2}}+\mathcal{O}\left(z^{-1}\right)\right\}\e^{-\frac{2}{3}(|s|z+s)^{\frac{3}{2}}\sigma_3}.
\end{equation*}
\end{enumerate}
\end{problem}
This concludes the first steps in the Deift-Zhou nonlinear steepest descent 
road map.  We point out that so far we have not used that $\gamma\in[0,1)$ is fixed. This feature enters our analysis in the next transformation.

\subsection{Normalization through the g-function transformation.}
\label{sec:rhp-gfun}
We choose to work with the function
\begin{equation*}
  g(z):=\frac{2}{3}(z-1)^{\frac{3}{2}},\ \ \ z\in\mathbb{C}\backslash(-\infty,1],
\end{equation*}
that is defined and analytic off the cut $(-\infty,1]\subset\mathbb{R}$ such that $(z-1)^{\frac{3}{2}}>0$ for $z>1$. The transformation
\begin{equation}\label{Sdef}
  {\bf S}(z):={\bf T}(z)\e^{tg(z)\sigma_3},\ \ \ z\in\mathbb{C}\backslash\Sigma_T;\ \ \ \ t=(-s)^{\frac{3}{2}}
\end{equation}
leads us then to the following problem.
\begin{problem}\label{normRHP} The normalized function ${\bf S}(z)={\bf S}(z;s,\gamma)\in\mathbb{C}^{2\times 2}$ is characterized by the following properties:
 \begin{enumerate}
  \item ${\bf S}(z)$ is analytic for $z\in\mathbb{C}\backslash(\Sigma_T\cup\{0\})$.
  \item The limiting values ${\bf S}_{\pm}(z)$, $z\in\Sigma_T$, from either side of the oriented contours are related by the equations
    \begin{equation*}
\begin{split}
      {\bf S}_+(z)&={\bf S}_-(z)\bigl(\begin{smallmatrix}
                    1 & 0\\
		    \e^{2tg(z)} & 1
                   \end{smallmatrix}\bigr), \hspace{1.65in} z\in(\Gamma_2\cup\Gamma_4)\backslash\{0\};\\
{\bf S}_+(z)& ={\bf S}_-(z)\bigl(\begin{smallmatrix}
	  0 & 1\\
-1 & 0
\end{smallmatrix}\bigr) \hspace{1.9in} z\in\Gamma_3\backslash\{0\};\\
    {\bf S}_+(z)&={\bf S}_-(z)\begin{pmatrix}
                  \e^{t(g_+(z)-g_-(z))} & 1-\gamma\\
		  0 & \e^{-t(g_+(z)-g_-(z))}
                 \end{pmatrix},\ \ z\in(0,1);\\ 
  {\bf S}_+(z)&={\bf S}_-(z)\begin{pmatrix}
		  1 & \e^{-t(\varkappa+2g(z))}\\
		  0 & 1
		  \end{pmatrix}, \hspace{1.1in} z\in(1,\infty),
\end{split}
  \end{equation*}
where we have introduced the abbreviation
\begin{equation}\label{kapdef}
	\varkappa:=-\frac{1}{t}\ln(1-\gamma)\in[0,+\infty)\ \ \ \ \Leftrightarrow\ \ \ \ 1-\gamma=\e^{-\varkappa t}.
\end{equation}
  \item Near $z=0$, with $\textnormal{arg}\,z\in(0,2\pi)$,
  \begin{equation}\label{l:0}
    {\bf S}(z)\e^{-tg(z)\sigma_3} = \widehat{{\bf T}}(z)\begin{pmatrix}
		1 & \frac{\gamma}{2\pi\im}\ln z\\
		0 & 1
		\end{pmatrix}\begin{cases}
		{\bf I},&\textnormal{arg}\,z\in(0,\frac{2\pi}{3})\\
		\bigl(\begin{smallmatrix}
		1 & 0\\
		-1 & 1
		\end{smallmatrix}\bigr),&\textnormal{arg}\,z\in(\frac{2\pi}{3},\pi)\\
		\bigl(\begin{smallmatrix}
		1 & -1\\
		0 & 1
		\end{smallmatrix}\bigr)\bigl(\begin{smallmatrix}
		1 & 0\\
		1 & 1
		\end{smallmatrix}\bigr),&\textnormal{arg}\,z\in(\pi,\frac{4\pi}{3})\\
		\bigl(\begin{smallmatrix}
		1 & -1\\
		0 & 1
		\end{smallmatrix}\bigr),&\textnormal{arg}\,z\in(\frac{4\pi}{3},2\pi).
		\end{cases}
  \end{equation}
  \item At $z=\infty$, we have the normalized behavior
    \begin{equation}\label{l:00}
	{\bf S}(z)= \big(|s|z\big)^{-\frac{1}{4}\sigma_3}\underbrace{\frac{1}{\sqrt{2}}\begin{pmatrix}
	                                                                          1 & 1\\
    -1 & 1
	                                                                         \end{pmatrix}\e^{-\im\frac{\pi}{4}\sigma_3}}_{=: {\bf N}}
  \left\{{\bf I}+{\bf X}_1(|s|z)^{-\frac{1}{2}}+\mathcal{O}\left(z^{-1}\right)\right\}.
    \end{equation}
 \end{enumerate}
\end{problem}
At this point we make three observations.
\begin{prop} Observation 1:
\begin{equation}\label{es:1}
 \Re\big(g(z)\big)<0,\ \ \ z\in (\Gamma_2\cup\Gamma_4)\backslash\{0\}.
\end{equation}
Observation 2:
\begin{equation}\label{es:2}
 \varkappa+2g(z)=\varkappa+\frac{4}{3}(z-1)^{\frac{3}{2}}>0,\ \ \ z\in(1,\infty).
\end{equation}
Observation 3:
\begin{equation*}
  \Pi(z)=g_+(z)-g_-(z)=-\frac{4\im}{3}(1-z)^{\frac{3}{2}},\ \ \ \ z\in(-\infty,1).
\end{equation*}
\end{prop}
We now define $\phi(z):=\frac{4}{3}(z-1)^{\frac{3}{2}}=2g(z)$, $z\in\mathbb{C}\backslash(-\infty,1]$ and note that
\begin{equation}\label{es:3}
	z\in(0,1):\ \ \phi_+(z)=\Pi(z)=-\phi_-(z);\ \ \ \ \ \ \ \ \Re\big(\phi(z)\big)<0\ \ \textnormal{if}\ \ \Im z\lessgtr 0,\ \ \Re z\in(0,1).
\end{equation}
This allows us in turn to perform the following transformation.
\subsection{Factorization and opening of lens} 
\label{sec:rhp-lenses}
Observe that with \eqref{kapdef}
\begin{equation*}
  z\in(0,1):\ \ \begin{pmatrix}
   \e^{t\Pi(z)} & 1-\gamma\\
  0 & \e^{-t\Pi(z)}
  \end{pmatrix}=\begin{pmatrix}
  1 & 0\\
  \e^{t(\phi_-(z)+\varkappa)} & 1
  \end{pmatrix}\begin{pmatrix}
0 & \e^{-t\varkappa}\\
-\e^{t\varkappa} & 0
\end{pmatrix}\begin{pmatrix}
1 & 0\\
\e^{t(\phi_+(z)+\varkappa)} & 1
\end{pmatrix}.
\end{equation*}
Now notice Figure \ref{figlens} below where $\gamma^{\pm}$ denote the boundaries of the lens-shaped regions $\Omega_{\pm}$ and define
\begin{equation}\label{l:1}
  {\bf L}(z):={\bf S}(z)\begin{cases}
	    \bigl(\begin{smallmatrix}
	           1 & 0\\
  		-\e^{t(\phi(z)+\varkappa)} & 1
	          \end{smallmatrix}\bigr),&z\in\Omega_+\\
	    \bigl(\begin{smallmatrix}
	           1 & 0\\
		  \e^{t(\phi(z)+\varkappa)} & 1
	          \end{smallmatrix}\bigr),&z\in\Omega_-\\
	    {\bf I},&z\in\mathbb{C}\backslash(\Sigma_T\cup\Omega_+\cup\Omega_-).
           \end{cases}
\end{equation}
\begin{figure}[tbh]
\begin{center}
\resizebox{0.5\textwidth}{!}{\includegraphics{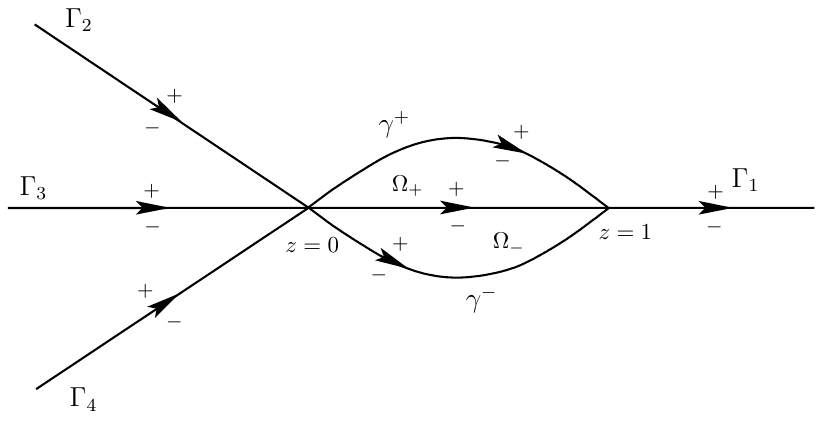}}
\caption{The oriented jump contour $\Sigma_L$ in the complex $z$-plane.}
\label{figlens}
\end{center}
\end{figure}

This transforms the previous RHP \ref{normRHP} for ${\bf S}(z)$ to the following problem.
\begin{problem}\label{lensRHP}
 Determine ${\bf L}(z)={\bf L}(z;s,\gamma)\in\mathbb{C}^{2\times 2}$ characterized by the following properties:
\begin{enumerate}
 \item 	${\bf L}(z)$ is analytic for $z\in\mathbb{C}\backslash(\Sigma_L\cup\{0\})$.
  \item Along the contour $\Sigma_L$ shown in Figure \ref{figlens}, the limiting values ${\bf L}_{\pm}(z)$ are related via
\begin{equation*}
\begin{split}
   {\bf L}_+(z)& ={\bf L}_-(z)\bigl(\begin{smallmatrix}
                    1 & 0\\
		    \e^{2tg(z)} & 1
                   \end{smallmatrix}\bigr), \hspace{.39in} z\in(\Gamma_2\cup\Gamma_4)\backslash\{0\};\\
   {\bf L}_+(z)& ={\bf L}_-(z)\bigl(\begin{smallmatrix}
	  0 & 1\\
-1 & 0
\end{smallmatrix}\bigr), \hspace{.56in} z\in\Gamma_3\backslash\{0\}; \\
  {\bf L}_+(z) & ={\bf L}_-(z)\bigl(\begin{smallmatrix}
                0 & \e^{-t\varkappa}\\
  -\e^{t\varkappa} & 0
               \end{smallmatrix}\bigr), \hspace{.27in} z\in(0,1);\\
{\bf L}_+(z) & ={\bf L}_-(z)\bigl(\begin{smallmatrix}
1 & \e^{-t(\varkappa+2g(z))}\\
0 & 1
\end{smallmatrix}\bigr),\ \ z\in(1,\infty);\\
  {\bf L}_+(z) & ={\bf L}_-(z)\bigl(\begin{smallmatrix}
                1 & 0\\
	    \e^{t(\phi(z)+\varkappa)} & 1
               \end{smallmatrix}\bigr), \hspace{.21in} z\in\gamma^+\cup\gamma^-.
\end{split}
\end{equation*}
  \item The singular behavior near $z=0$ needs to be adjusted according to \eqref{l:1}, which amounts to 
  the multiplication of \eqref{l:0} by the corresponding lower triangular matrices from the right.
  \item As $z\rightarrow\infty$, the behavior of $L(z)$ is unchanged from \eqref{l:00}.
\end{enumerate}
\end{problem}
The importance of transformation \eqref{l:1} comes from the fact that, because of \eqref{es:1}, \eqref{es:2}, and \eqref{es:3},
we have now the following behavior for the jump matrix ${\bf G}_L(z;s,\gamma)$ in the problem for ${\bf L}(z)$:
\begin{equation*}
  {\bf G}_L(z;s,\gamma)\rightarrow {\bf I},\ \ \ \ \ s\rightarrow-\infty,\ \ \gamma\in[0,1)\ \textnormal{fixed}
\end{equation*}
uniformly for $z\in\Sigma_L$ away from the line segment $(-\infty,1)$ and small neighborhoods of $z=0$ and $z=1$. 
We have thus reached the point at which we need to focus on the local model problems.
\subsection{Local analysis} 
\label{sec:rhp-model}
We first construct the outer parametrix, which satisfies the problem below.
\begin{problem}\label{out:1}
  Find ${\bf P}^{(\infty)}(z)={\bf P}^{(\infty)}(z;s,\gamma)\in\mathbb{C}^{2\times 2}$ such that
\begin{enumerate}
 \item ${\bf P}^{(\infty)}(z)$ is analytic for $z\in\mathbb{C}\backslash(-\infty,1]$.
  \item The function ${\bf P}^{(\infty)}(z)$ assumes square-integrable limiting values on $(-\infty,1]$ that are related
  by the jump conditions
  \begin{equation*}
\begin{split}
    {\bf P}_+^{(\infty)}(z)&={\bf P}_-^{(\infty)}(z)\bigl(\begin{smallmatrix}
                                        0 & 1\\
					-1 & 0
                                       \end{smallmatrix}\bigr), \hspace{.87in} z\in(-\infty,0);\\
  {\bf P}_+^{(\infty)}(z)&={\bf P}_-^{(\infty)}(z)
\e^{-\frac{t}{2}\varkappa\sigma_3}\bigl(\begin{smallmatrix}
0 & 1\\
-1 & 0
\end{smallmatrix}\bigr)\e^{\frac{t}{2}\varkappa\sigma_3},\ \ z\in(0,1).
\end{split}
\end{equation*}
  \item As $z\rightarrow\infty$ with $\textnormal{arg}\,z\in(-\pi,\pi)$, compare \eqref{l:00},
  \begin{equation*}
    {\bf P}^{(\infty)}(z)=(|s|z)^{-\frac{1}{4}\sigma_3}{\bf N}\left\{{\bf I}+\frac{2\im\nu}{z^{\frac{1}{2}}}\sigma_3+\frac{1}{4z}\begin{pmatrix}
                                                                                                        -8\nu^2 & \im\\ -\im & -8\nu^2
                                                                                                       \end{pmatrix}+
\frac{\im\nu}{3z^{\frac{3}{2}}}\begin{pmatrix}
                               1-4\nu^2 & -\im\frac{3}{2}\\ -\im\frac{3}{2} & -(1-4\nu^2)
                              \end{pmatrix}+
\mathcal{O}\left(z^{-2}\right)\right\}.
  \end{equation*}
\end{enumerate}
\end{problem}
As can be seen from a direct computation, the choice
\begin{equation}\label{o:1}
	{\bf P}^{(\infty)}(z)=\big(|s|(z-1)\big)^{-\frac{1}{4}\sigma_3}\frac{1}{\sqrt{2}}\begin{pmatrix}
	1 & 1\\
	-1 & 1
	\end{pmatrix}\e^{-\im\frac{\pi}{4}\sigma_3}\big(\mathcal{D}(z)\big)^{-\sigma_3},\ \ \ z\in\mathbb{C}\backslash(-\infty,1]
\end{equation}
with the scalar Szeg\H{o}-type function
\begin{align}
	\mathcal{D}(z)&:=\exp\left[(z-1)^{\frac{1}{2}}\frac{t\varkappa}{2\pi}\int_0^1\frac{1}{\sqrt{1-w}}\frac{\d w}{w-z}\right]
	=\left(\frac{(z-1)^{\frac{1}{2}}-\im}{(z-1)^{\frac{1}{2}}+\im}\right)^{\nu},\ \ 
	z\in\mathbb{C}\backslash(-\infty,1],\label{o:2}\\ 
	\nu&:=\frac{t\varkappa}{2\pi\im},\hspace{1cm} \mathcal{D}_+(z)\mathcal{D}_-(z)=\begin{cases}
	1,&z\in(-\infty,0)\\
	\e^{-t\varkappa},&z\in(0,1)
	\end{cases}\nonumber
\end{align}
provides a solution to RHP \ref{out:1}. All branches of fractional exponents in \eqref{o:1} and \eqref{o:2} are principal ones
such that $\z^{\alpha}>0$ for $\z>0$.\bigskip

Next, we consider a small neighborhood of $z=1$ in which we require a solution to the following model problem.
\begin{problem}\label{Airy:1} Find ${\bf P}^{(1)}(z)={\bf P}^{(1)}(z;s,\gamma)$ such that
 \begin{enumerate}
  \item ${\bf P}^{(1)}(z)$ is analytic for $z\in \mathbb{D}_{\frac{1}{4}}(1)\backslash\Sigma_L$ with $\mathbb{D}_r(z_0):=\{z\in\mathbb{C}:\,
  |z-z_0|<r\}$.
  \item The model function displays the following local jump behavior (see Figure \ref{figlens} for contour orientation):
  \begin{equation*}
\begin{split}
    {\bf P}^{(1)}_+(z)&={\bf P}^{(1)}_-(z)\e^{-\frac{t}{2}\varkappa\sigma_3}\bigl(\begin{smallmatrix}
                                                                0 & 1\\ -1 & 0
                                                               \end{smallmatrix}\bigr)\e^{\frac{t}{2}\varkappa\sigma_3}, \hspace{.35in} z\in(0,1)\cap\mathbb{D}_{\frac{1}{4}}(1);\\
   {\bf P}^{(1)}_+(z)&={\bf P}^{(1)}_-(z)\e^{-\frac{t}{2}\varkappa\sigma_3}\bigl(\begin{smallmatrix}
                                                                 1 & \e^{-2tg(z)}\\ 0 & 1
                                                                \end{smallmatrix}\bigr)\e^{\frac{t}{2}\varkappa\sigma_3},\ \
  z\in(1,\infty)\cap\mathbb{D}_{\frac{1}{4}}(1);\\
     {\bf P}^{(1)}_+(z)&={\bf P}^{(1)}_-(z)\e^{-\frac{t}{2}\varkappa\sigma_3}\bigl(\begin{smallmatrix}
                                                                 1 & 0\\ \e^{t\phi(z)} & 1
                                                                \end{smallmatrix}\bigr)\e^{\frac{t}{2}\varkappa\sigma_3}, \hspace{.21in} z\in\big(\gamma^+\cup\gamma^-\big)\cap\mathbb{D}_{\frac{1}{4}}(1).
\end{split}
  \end{equation*}
  \item As $s\rightarrow-\infty$ with $\gamma\in[0,1)$ fixed, we have a matching between ${\bf P}^{(\infty)}(z)$ and 
  ${\bf P}^{(1)}(z)$ of the form
  \begin{equation}\label{Aimatch}
    {\bf P}^{(1)}(z)={\bf P}^{(\infty)}(z)\left\{{\bf I}+\frac{1}{48t(z-1)^{\frac{3}{2}}}
  \begin{pmatrix}
      1 & 6\im\e^{-t\varkappa}\\ 6\im\e^{t\varkappa} & -1
   \end{pmatrix}+\mathcal{O}\left(t^{-2}\right)\right\},\ \ \ \ \ t=(-s)^{\frac{3}{2}},
  \end{equation}
which holds uniformly in $0<r_1\leq|z-1|\leq r_2<\frac{1}{4}$ for any fixed $r_1,r_2$.
 \end{enumerate}
\end{problem}
A solution to this problem is most easily constructed in terms of the function $\Phi(\z)$ introduced in \eqref{e:3}. To be precise, we have
\begin{equation}\label{o:3}
  {\bf P}^{(1)}(z)={\bf E}^{(1)}(z)\Phi\big(\z(z)\big)\e^{\frac{2}{3}\z^{\frac{3}{2}}(z)\sigma_3}
  \e^{\frac{t}{2}\varkappa\sigma_3},\ \ \ \ \z(z):=t^{\frac{2}{3}}(z-1),\ \ \ \ z\in\mathbb{D}_{\frac{1}{4}}(1)\backslash\Sigma_L,
\end{equation}
where 
\begin{equation*}
  {\bf E}^{(1)}(z):={\bf P}^{(\infty)}(z)\e^{-\frac{t}{2}\varkappa\sigma_3}\big(\mathcal{D}(z)\big)^{-\sigma_3}
\big({\bf P}^{(\infty)}(z)\big)^{-1},\ \ \ z\in\mathbb{D}_{\frac{1}{4}}(1)
\end{equation*}
is analytic at $z=1$. Using RHP \ref{Airypara} it is straightforward to verify the required 
properties of \eqref{o:3}.
\begin{remark} The following Taylor expansion of ${\bf E}^{(1)}(z)$ as $z\rightarrow 1$ is used later:
\begin{equation*}
{\bf E}^{(1)}(z)=|s|^{-\frac{1}{4}\sigma_3}\left\{\begin{pmatrix}1&2\im\nu\\ 0&1\end{pmatrix}-(z-1)\begin{pmatrix}2\nu^2&\frac{2}{3}\im\nu(1+2\nu^2)\\ -2\im\nu & 2\nu^2\end{pmatrix}+\mathcal{O}\left((z-1)^2\right)\right\}|s|^{\frac{1}{4}\sigma_3}.
\end{equation*}
\end{remark}

Finally, we turn towards a vicinity of the origin $z=0$. We require ${\bf P}^{(0)}(z)$ satisfying the following.
\begin{problem}\label{oripara} Determine ${\bf P}^{(0)}(z)={\bf P}^{(0)}(z;s,\gamma)\in\mathbb{C}^{2\times 2}$ such that
\begin{enumerate}
 \item ${\bf P}^{(0)}(z)$ is analytic for $z\in\mathbb{D}_{\frac{1}{4}}(0)\backslash(\Sigma_L\cup\{0\})$.
  \item The model function has these jumps, with contour orientation near $z=0$ as shown in Figure \ref{figlens}:
  \begin{equation*}
\begin{split}
    {\bf P}_+^{(0)}(z)&={\bf P}_-^{(0)}(z)\bigl(\begin{smallmatrix}
                              0 & 1\\ -1 & 0
                             \end{smallmatrix}\bigr), \hspace{.27in} z\in\big(\Gamma_3\backslash\{0\}\big)\cap\mathbb{D}_{\frac{1}{4}}(0);\\
     {\bf P}_+^{(0)}(z)&={\bf P}_-^{(0)}(z)\bigl(\begin{smallmatrix}
                              1 & 0\\ \e^{2tg(z)} & 1
                             \end{smallmatrix}\bigr),\ \ z\in\big((\Gamma_2\cup\Gamma_4)\backslash\{0\}\big)\cap\mathbb{D}_{\frac{1}{4}}(0),
\end{split}
  \end{equation*}
and in the right half plane,
\begin{equation*}
\begin{split}
    {\bf P}_+^{(0)}(z)&={\bf P}_-^{(0)}(z)\e^{-\frac{t}{2}\varkappa\sigma_3}\bigl(\begin{smallmatrix}
                              0 & 1\\ -1 & 0
                             \end{smallmatrix}\bigr)\e^{\frac{t}{2}\varkappa\sigma_3}, \hspace{.23in} z\in\big(\Gamma_1\backslash\{0\}\big)\cap\mathbb{D}_{\frac{1}{4}}(0);\\
   {\bf P}_+^{(0)}(z)&={\bf P}_-^{(0)}(z)\e^{-\frac{t}{2}\varkappa\sigma_3}\bigl(\begin{smallmatrix}
                              1 & 0\\ \e^{t\phi(z)} & 1
                             \end{smallmatrix}\bigr)\e^{\frac{t}{2}\varkappa\sigma_3},\ \ z\in\big(\gamma^{\pm}\backslash\{0\}\big)\cap\mathbb{D}_{\frac{1}{4}}(0).
\end{split}
  \end{equation*}
  \item As $z\rightarrow 0$, the parametrix ${\bf P}^{(0)}(z)$ matches the singular behavior of the function ${\bf L}(z)$ as
  outlined in RHP \ref{lensRHP}, condition (3).
  \item We have the matching, as $s\rightarrow-\infty$ with $\gamma\in[0,1)$ fixed,
\begin{equation}\label{confmatch}
  {\bf P}^{(0)}(z)=\left\{{\bf I}+\frac{\im}{\z(z)}{\bf E}^{(0)}(z)\e^{\im\frac{\pi}{2}\nu\sigma_3}\begin{pmatrix}
                                                                               \nu^2 & -\frac{\Gamma(1-\nu)}{\Gamma(\nu)}\\ \frac{\Gamma(1+\nu)}{\Gamma(-\nu)} & -\nu^2
                                                                              \end{pmatrix}\e^{-\im\frac{\pi}{2}\nu\sigma_3}\big({\bf E}^{(0)}(z)\big)^{-1}
+\mathcal{O}\left(t^{-\frac{5}{3}}\right)\right\}{\bf P}^{(\infty)}(z),
\end{equation}
which holds uniformly for $0<r_1\leq |z|\leq r_2<\frac{1}{4}$ with $r_1,r_2$ fixed. Here, $\nu=\frac{t\varkappa}{2\pi\im}$ and $\z(z)$ as well as ${\bf E}^{(0)}(z)$ are defined in
\eqref{o:7} and \eqref{o:8} below.
\end{enumerate}
\end{problem}
The construction of this model function is achieved in terms of the confluent hypergeometric function $U(a,\z)\equiv U(a,1,\z)$ (see \cite{NIST}) and differs
only marginally from the ones given, for example, in \cite{BI,IK,BogatskiyCI:2016}. We define
\begin{equation*}
  \Psi_0(\z):=\begin{pmatrix}
              U(\nu,\e^{\im\frac{\pi}{2}}\z)\e^{2\pi\im\nu} & -U(1-\nu,\e^{-\im\frac{\pi}{2}}\z)\e^{\im\pi\nu}\frac{\Gamma(1-\nu)}{\Gamma(\nu)}\\
-U(1+\nu,\e^{\im\frac{\pi}{2}}\z)\e^{\im\pi\nu}\frac{\Gamma(1+\nu)}{\Gamma(-\nu)} & U(-\nu,\e^{-\im\frac{\pi}{2}}\z)
             \end{pmatrix}\e^{-\frac{\im}{2}\z\sigma_3},\ \ \ \textnormal{arg}\,\z\in\left(-\frac{\pi}{2},\frac{3\pi}{2}\right)
\end{equation*}
where $\nu=\frac{t\varkappa}{2\pi\im}\in\im\mathbb{R}$ and assemble
\begin{equation}\label{o:4}
  \Psi(\z)=\Psi_0(\z)\begin{cases}
                      {\bf I},&\textnormal{arg}\,\z\in(0,\frac{\pi}{3})\\
		      \bigl(\begin{smallmatrix} 
		             1 & 0\\ \e^{\im\pi\nu} & 1
		            \end{smallmatrix}\bigr),&\textnormal{arg}\,\z\in(\frac{\pi}{3},\frac{\pi}{2})\\
\bigl(\begin{smallmatrix}
       1 & 0 \\ 2\im\sin\pi\nu & 1
      \end{smallmatrix}\bigr)\bigl(\begin{smallmatrix}
       1 & 0\\ \e^{-\im\pi\nu} & 1
      \end{smallmatrix}\bigr),&\textnormal{arg}\,\z\in(\frac{\pi}{2},\frac{2\pi}{3})\\
\bigl(\begin{smallmatrix}
       1 & 0\\ 2\im\sin\pi\nu & 1
      \end{smallmatrix}\bigr),&\textnormal{arg}\,\z\in(\frac{2\pi}{3},\pi)
  \end{cases}\ \ \
\begin{cases}
\bigl(\begin{smallmatrix}
       1 & 0\\ 2\im\sin\pi\nu & 1
      \end{smallmatrix}\bigr)\bigl(\begin{smallmatrix}
       0 & -\e^{\im\pi\nu}\\ \e^{-\im\pi\nu} & 0
      \end{smallmatrix}\bigr),&\textnormal{arg}\,\z\in(\pi,\frac{4\pi}{3})\smallskip\\
\bigl(\begin{smallmatrix}
       1 & 0\\ 2\im\sin\pi\nu & 1
      \end{smallmatrix}\bigr)\bigl(\begin{smallmatrix}
       1 & -\e^{\im\pi\nu} \\ \e^{-\im\pi\nu} & 0
      \end{smallmatrix}\bigr),&\textnormal{arg}\,\z\in(\frac{4\pi}{3},\frac{3\pi}{2})\\
    \bigl(\begin{smallmatrix}
       1 & -\e^{-\im\pi\nu} \\ \e^{\im\pi\nu} & 0
      \end{smallmatrix}\bigr),&\textnormal{arg}\,\z\in(-\frac{\pi}{2},-\frac{\pi}{3})\\
  \bigl(\begin{smallmatrix}
       0 & -\e^{-\im\pi\nu} \\ \e^{\im\pi\nu} & 0
      \end{smallmatrix}\bigr),&\textnormal{arg}\,\z\in(-\frac{\pi}{3},0).
\end{cases}
\end{equation}
With the standard properties of confluent hypergeometric functions in mind (see \cite{NIST}), we obtain the {\it confluent hypergeometric-type parametrix}
\begin{problem}\label{conflu} The function $\Psi(\z)\in\mathbb{C}^{2\times 2}$ has the following properties:
\begin{enumerate}
 \item $\Psi(\z)$ is analytic for $\z\in\mathbb{C}\backslash\big(\{\textnormal{arg}\,\z=0,\pm\frac{\pi}{3},\frac{2\pi}{3},\pi,\frac{4\pi}{3}\}\cup\{0\}\big)$
  and the six rays are oriented as shown locally near $z=0$ in Figure \ref{figlens}.
  \item Along the jump contours, we have
\begin{align*}
  \Psi_+(\z)=&\,\Psi_-(\z)\e^{-\im\frac{\pi}{2}\nu\sigma_3}\bigl(\begin{smallmatrix}
       0 & 1\\ -1 & 0
      \end{smallmatrix}\bigr)\e^{\im\frac{\pi}{2}\nu\sigma_3},\ \ \textnormal{arg}\,\z=0;\ \ \, \ \ \Psi_+(\z)=\Psi_-(\z)\e^{\im\frac{\pi}{2}\nu\sigma_3}\bigl(\begin{smallmatrix}
       0 & 1\\ -1 & 0
      \end{smallmatrix}\bigr)\e^{-\im\frac{\pi}{2}\nu\sigma_3},\ \ \textnormal{arg}\,\z=\pi;\\
   \Psi_+(\z)=&\,\Psi_-(\z)\e^{-\im\frac{\pi}{2}\nu\sigma_3}\bigl(\begin{smallmatrix}
       1 & 0\\ 1 & 1
      \end{smallmatrix}\bigr)\e^{\im\frac{\pi}{2}\nu\sigma_3},\ \ \textnormal{arg}\,\z=\pm\frac{\pi}{3};\ \ \ 
  \Psi_+(\z)=\Psi_-(\z)\e^{\im\frac{\pi}{2}\nu\sigma_3}\bigl(\begin{smallmatrix}
       1 & 0\\ 1 & 1
      \end{smallmatrix}\bigr)\e^{-\im\frac{\pi}{2}\nu\sigma_3},\ \ \textnormal{arg}\,\z=\frac{2\pi}{3},\frac{4\pi}{3};
\end{align*}
and there are no jumps on the vertical axis $\textnormal{arg}\,\z=\pm\frac{\pi}{2}$.
  \item Near $\z=0$,
  \begin{equation}\label{o:5}
    \Psi(\z)=\widehat{\Psi}(\z)\begin{pmatrix}
                                1 & \frac{\gamma}{2\pi\im}\ln\zeta\\
				0 & 1
                               \end{pmatrix}\left.\begin{cases}
                                             \bigl(\begin{smallmatrix}
       1 & 0\\ -\e^{2\pi\im\nu} & 1
      \end{smallmatrix}\bigr),&\textnormal{arg}\,\z\in(0,\frac{\pi}{3})\\
      {\bf I},&\textnormal{arg}\,\z\in(\frac{\pi}{3},\frac{2\pi}{3})\\
      \bigl(\begin{smallmatrix}
       1 & 0\\ -1 & 1
      \end{smallmatrix}\bigr),&\textnormal{arg}\,\z\in(\frac{2\pi}{3},\pi)\\
       \bigl(\begin{smallmatrix}
       1 & -1\\ 0 & 1
      \end{smallmatrix}\bigr)\bigl(\begin{smallmatrix}
      1 & 0\\ 1 & 1 \end{smallmatrix}\bigr),&\textnormal{arg}\,\z\in(\pi,\frac{4\pi}{3})\\
      \bigl(\begin{smallmatrix}
             1 & -1\\ 0 & 1
            \end{smallmatrix}\bigr),&\textnormal{arg}\,\z\in(\frac{4\pi}{3},\frac{5\pi}{3})\\
      \bigl(\begin{smallmatrix}
       1 & -1\\ 0 & 1
      \end{smallmatrix}\bigr)\bigl(\begin{smallmatrix}
       1 & 0\\ \e^{2\pi\im\nu} & 1
      \end{smallmatrix}\bigr),&\textnormal{arg}\,\z\in(\frac{5\pi}{3},2\pi)
                                            \end{cases}\right\}\,\times\,\e^{-\im\frac{\pi}{2}\nu\sigma_3},
  \end{equation}
  where $\widehat{\Psi}(\z)$ is analytic at $\z=0$ and the branch of 
  the logarithm in \eqref{o:5} is such that $\textnormal{arg}\,\z\in(0,2\pi)$.
  \item As $\z\rightarrow\infty$,
\begin{equation*}
\begin{split}
\Psi(\z)= & \left[{\bf I}+\frac{\im}{\z}\e^{\im\frac{\pi}{2}\nu\sigma_3} \begin{pmatrix} \nu^2 & -\frac{\Gamma(1-\nu)}{\Gamma(\nu)} \\ \frac{\Gamma(1+\nu)}{\Gamma(-\nu)} & -\nu^2 \end{pmatrix}\e^{-\im\frac{\pi}{2}\nu\sigma_3}+ \mathcal{O}\left(\z^{-2}\right)\right] \\
 & \hspace{1in} \times \z^{-\nu\sigma_3}\e^{-\frac{\im}{2}\z\sigma_3}\begin{cases}  \Bigl(\begin{smallmatrix} \e^{\im\frac{3\pi}{2}\nu} & 0\\ 0 & \e^{-\im\frac{\pi}{2}\nu}  \end{smallmatrix}\Bigr),& \textnormal{arg}\,\z\in(0,\pi)\smallskip\\ \Bigl(\begin{smallmatrix} 0 & -\e^{\im\frac{\pi}{2}\nu}\\ \e^{\im\frac{\pi}{2}\nu} & 0 \end{smallmatrix}\Bigr),&\textnormal{arg}\,\z\in(-\pi,0)\end{cases}.
\end{split}
\end{equation*}

\end{enumerate}

\end{problem}
\begin{remark} Using the local expansions of $U(a,\z)$ near $\z=0$, we have in fact
\begin{equation*}
    \widehat{\Psi}(\z) =\e^{\im\frac{\pi}{2}\nu\sigma_3}\begin{pmatrix}
                                                         \widehat{\Psi}_{11}(\z) & \widehat{\Psi}_{12}(\z)\\
							  \widehat{\Psi}_{21}(\z) & \widehat{\Psi}_{22}(\z)
                                                        \end{pmatrix}\e^{-\frac{\im}{2}\z\sigma_3},\ \ \ |\z|<r,
  \end{equation*}
where the entries equal
  \begin{align*}
    \widehat{\Psi}_{11}(\z)&=-\frac{2\pi\im}{\gamma}\,U_2(\nu,\e^{\im\frac{\pi}{2}}\z),\ \ 
    \widehat{\Psi}_{12}(\z)=-\Big[U_1(1-\nu,\e^{-\im\frac{\pi}{2}}\z)+\frac{\pi}{2}\e^{-\im\frac{\pi}{2}}\,U_2(1-\nu,\e^{-\im\frac{\pi}{2}}\z)\Big]
  \frac{\Gamma(1-\nu)}{\Gamma(\nu)},\\
    \widehat{\Psi}_{21}(\z)&=\frac{2\pi\im}{\gamma}\,U_2(1+\nu,e^{\im\frac{\pi}{2}}\z)\frac{\Gamma(1+\nu)}{\Gamma(-\nu)},\ \ 
  \widehat{\Psi}_{22}(\z)=U_1(-\nu,\e^{-\im\frac{\pi}{2}}\z)+\frac{\pi}{2}\e^{-\im\frac{\pi}{2}}\,U_2(-\nu,\e^{-\im\frac{\pi}{2}}\z)
  \end{align*}
  in terms of the locally analytic functions
  \begin{equation*}
    U_1(a,\z):=-\frac{1}{\Gamma(a)}\sum_{k=0}^{\infty}\frac{(a)_k}{(k!)^2}\big(\psi(a+k)-2\psi(1+k)\big)\z^k,\ \ \ 
    U_2(a,\z):=-\frac{1}{\Gamma(a)}\sum_{k=0}^{\infty}\frac{(a)_k}{(k!)^2}\z^k,\ \ |\z|<r.
  \end{equation*}
  These expressions imply in particular that
  \begin{align*}
	\widehat{\Psi}\big(\z(z)\big)=&\,\e^{\im\frac{\pi}{2}\nu\sigma_3}\Bigg\{\begin{pmatrix}\frac{2\pi\im}{\gamma}\frac{1}{\Gamma(\nu)} & \frac{1}{\Gamma(\nu)}\big(\psi(1-\nu)-2\psi(1)-\im\frac{\pi}{2}\big)\\ -\frac{2\pi\im}{\gamma}\frac{1}{\Gamma(-\nu)} & -\frac{1}{\Gamma(-\nu)}\big(\psi(-\nu)-2\psi(1)-\im\frac{\pi}{2}\big) \end{pmatrix} \\
	& \hspace{.5in} +z\begin{pmatrix}\frac{2\pi\im}{\gamma}\frac{2\im\nu t}{\Gamma(\nu)}&\ast\\
	-\frac{2\pi\im}{\gamma}\frac{2\im(\nu+1)t}{\Gamma(-\nu)} & \ast\end{pmatrix} +\mathcal{O}\left(z^2\right)\Bigg\}\e^{-\frac{\im}{2}\z\sigma_3},
\end{align*}
where $*$ denotes entries that are irrelevant to us.
\end{remark}
The function $\Psi(\z)$ leads directly to a solution of RHP \ref{oripara} through the relation
\begin{equation}\label{o:6}
 {\bf P}^{(0)}(z)={\bf E}^{(0)}(z)\Psi\big(\z(z)\big)\left.\begin{cases}
                                          \e^{\frac{\im}{2}(\z(z)-\frac{4}{3}t)\sigma_3},&z\in\mathbb{D}_{\frac{1}{4}}(0)\backslash\Sigma_L:\ \textnormal{arg}\,z\in(0,\pi)\\
					  \e^{-\frac{\im}{2}(\z(z)-\frac{4}{3}t)\sigma_3},&z\in\mathbb{D}_{\frac{1}{4}}(0)\backslash\Sigma_L:\ \textnormal{arg}\,z\in(-\pi,0)
                                         \end{cases}\right\}\ \times
  \e^{\im\frac{\pi}{2}\nu\sigma_3},
\end{equation}
in which
\begin{equation}\label{o:7}
  z\in\mathbb{D}_{\frac{1}{4}}(0):\ \ \z(z):=-2\im t\,\textnormal{sgn}(\Im z)\int_0^z(\lambda-1)^{\frac{1}{2}}\d\lambda=2tz\left(1-\frac{z}{4}+\mathcal{O}\left(z^2\right)\right),\ \ z\rightarrow 0,
\end{equation}
and
\begin{equation}\label{o:8}
  {\bf E}^{(0)}(z):={\bf P}^{(\infty)}(z)\e^{-\im\frac{\pi}{2}\nu\sigma_3}\begin{cases}\e^{\im\frac{2}{3}t\sigma_3}\bigl(\begin{smallmatrix}\e^{-\im\frac{3\pi}{2}\nu} & 0\\ 0&\e^{\im\frac{\pi}{2}\nu}\end{smallmatrix}\bigr)\z(z)^{\nu\sigma_3},&z\in\mathbb{D}_{\frac{1}{4}}(0):\ \textnormal{arg}\,z\in(0,\pi)\smallskip\\ \e^{-\im\frac{2}{3}t\sigma_3}\bigl(\begin{smallmatrix}0&\e^{-\im\frac{\pi}{2}\nu}\\ -\e^{-\im\frac{\pi}{2}\nu} & 0\end{smallmatrix}\bigr)
  \z(z)^{\nu\sigma_3},&z\in\mathbb{D}_{\frac{1}{4}}(0):\ \textnormal{arg}\,z\in(-\pi,0)\end{cases}
\end{equation}
are both analytic at $z=0$. We choose the integration path in \eqref{o:7} in the separate half planes without crossing the cut $(-\infty,1]$. Using RHP \ref{conflu}, it is straightforward to verify 
the required properties of \eqref{o:6}.
\begin{remark} Note that from \eqref{o:5} as well as \eqref{l:0} and \eqref{l:1}, we obtain
\begin{equation*}
	{\bf L}(z)={\bf N}_0(z){\bf P}^{(0)}(z),\ \ \ z\in\mathbb{D}_{\frac{1}{4}}(0),
\end{equation*}
where ${\bf N}_0(z)$ is analytic at $z=0$.
\end{remark}
\begin{remark} Analyticity of ${\bf E}^{(0)}(z)$ at $z=0$ allows us to compute the following Taylor expansion which is used later on.  As $z\rightarrow 0$,
\begin{equation*}
	{\bf E}^{(0)}(z)=\e^{-\im\frac{\pi}{4}\sigma_3}|s|^{-\frac{1}{4}\sigma_3}{\bf N}\left\{{\bf I}-\frac{z}{4}\big(\sigma_2+3\nu\sigma_3\big)+\mathcal{O}\left(z^2\right)\right\}\e^{\im\frac{2}{3}t\sigma_3}\begin{pmatrix}\e^{-\im\pi\nu} & 0\\ 0& 1\end{pmatrix}(8t)^{\nu\sigma_3}.
\end{equation*}
\end{remark}
This concludes the local analysis for fixed $\gamma\in[0,1)$.
\subsection{Ratio transformation and small norm estimates.}
\label{sec:rhp-error}
With \eqref{o:1}, \eqref{o:3}, and \eqref{o:6}, this step amounts to the transformation
\begin{equation}\label{Rdef}
	{\bf R}(z):=\begin{pmatrix}
	1 & 0\\
	-2\im\nu|s|^{\frac{1}{2}} & 1
	\end{pmatrix}{\bf L}(z)\begin{cases}
	\big({\bf P}^{(0)}(z)\big)^{-1},&z\in\mathbb{D}_r(0)\\
	\big({\bf P}^{(1)}(z)\big)^{-1},&z\in\mathbb{D}_r(1)\\
	\big({\bf P}^{(\infty)}(z)\big)^{-1},&z\notin\big(\overline{\mathbb{D}_r(0)}\cup\overline{\mathbb{D}_r(1)}\big)
	\end{cases}
\end{equation}
in which $0<r<\frac{1}{4}$ is kept fixed. Recalling RHPs \ref{out:1}, \ref{Airy:1}, and \ref{oripara}, we obtain the following problem.
\begin{figure}[tbh]
\begin{center}
\resizebox{0.5\textwidth}{!}{\includegraphics{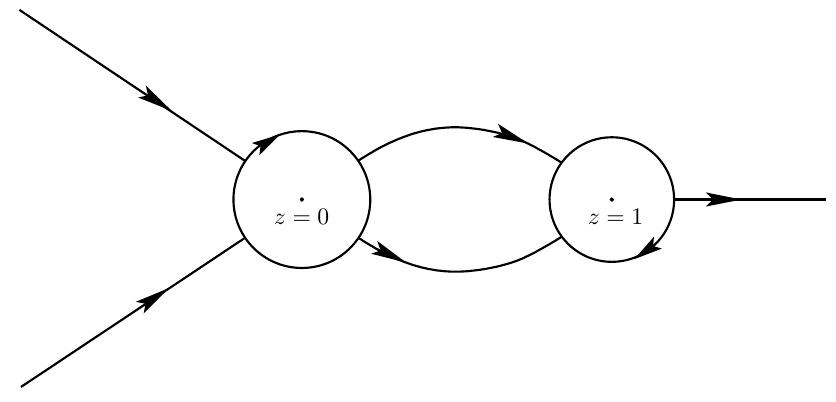}}
\caption{Jump contours for ${\bf R}(z)$ defined in \eqref{Rdef} in the complex $z$-plane.}
\label{figr2}
\end{center}
\end{figure}
\begin{problem}\label{ratio1} Determine ${\bf R}(z)={\bf R}(z;s,\gamma)\in\mathbb{C}^{2\times 2}$ such that
\begin{enumerate}
	\item $R(z)$ is analytic for $z\in\mathbb{C}\backslash\Sigma_R$ with square-integrable boundary values on the contour
	\begin{equation*}
		\Sigma_R:=\partial\mathbb{D}_r(0)\cup\partial\mathbb{D}_r(1)\cup(1+r,\infty)\cup\Big(\big(\Gamma_2\cup\Gamma_4\cup\gamma^+\cup\gamma^-\big)\cap\{z\in\mathbb{C}:|z|>r,|z-1|>r\}\Big),
	\end{equation*}
	which is shown in Figure \ref{figr2}.
	\item On the contour $\Sigma_R$, the boundary values ${\bf R}_{\pm}(z)$ are related via ${\bf R}_+(z)={\bf R}_-(z){\bf G}_R(z;s,\gamma)$ with
	\begin{align*}
		{\bf G}_R(z;s,\gamma)&={\bf P}^{(0)}(z)\big({\bf P}^{(\infty)}(z)\big)^{-1}, \hspace{1.34in} z\in\partial\mathbb{D}_r(0);\\ 
		{\bf G}_R(z;s,\gamma)&={\bf P}^{(1)}(z)\big({\bf P}^{(\infty)}(z)\big)^{-1}, \hspace{1.34in} z\in\partial\mathbb{D}_r(1);\\
		{\bf G}_R(z;s,\gamma)&={\bf P}^{(\infty)}(z)\begin{pmatrix}
		1 & \e^{-t(\varkappa+2g(z))}\\
		0 & 1
		\end{pmatrix}\big({\bf P}^{(\infty)}(z)\big)^{-1},\ \ \,z\in(1+r,\infty);\\
		{\bf G}_R(z;s,\gamma)&={\bf P}^{(\infty)}(z)\begin{pmatrix}
		1 & 0\\
		\e^{2tg(z)} & 1
		\end{pmatrix}\big({\bf P}^{(\infty)}(z)\big)^{-1}, \hspace{.45in} z\in\big(\Gamma_2\cup\Gamma_4\big)\cap\{z\in\mathbb{C}:\,|z|>r\};\\
		{\bf G}_R(z;s,\gamma)&={\bf P}^{(\infty)}(z)\begin{pmatrix}
		1 & 0\\
		\e^{t(\phi(z)+\varkappa)} & 1
		\end{pmatrix}\big({\bf P}^{(\infty)}(z)\big)^{-1}, \hspace{.26in} z\in\gamma^{\pm}\cap\{z\in\mathbb{C}:\,|z|>r,\ |z-1|>r\}.
	\end{align*}
	By construction, compare RHP \ref{out:1}, \ref{Airy:1}, and \ref{oripara}, there are no jumps inside $\mathbb{D}_r(0)\cup \mathbb{D}_r(1)$ and on $(-\infty,-r)\cup(r,1-r)$. Moreover, $R(z)$ is bounded at $z=0$.
	\item As $z\rightarrow\infty$,
	\begin{equation*}
		{\bf R}(z)=\begin{pmatrix}
	1 & 0\\
	-2\im\nu|s|^{\frac{1}{2}} & 1
	\end{pmatrix}(|s|z)^{-\frac{1}{4}\sigma_3}{\bf N}\left\{{\bf I}+{\bf X}_1(|s|z)^{-\frac{1}{2}}+\mathcal{O}\left(z^{-1}\right)\right\}
	\big({\bf P}^{(\infty)}(z)\big)^{-1}={\bf I}+\mathcal{O}\left(z^{-\frac{1}{2}}\right).
	\end{equation*}
\end{enumerate}
\end{problem}
Using standard small norm estimates, compare \eqref{es:1}, \eqref{es:2}, \eqref{es:3}, and \eqref{Aimatch}, we deduce
\begin{prop}\label{DZes:1} For any fixed $\gamma\in[0,1)$, there exist positive constants $t_0=t_0(\gamma)$ and $c=c(\gamma)$
such that
\begin{equation*}
  \|{\bf G}_R(\cdot;s,\gamma)-{\bf I}\|_{L^2\cap L^{\infty}(\Sigma_R\backslash\partial\mathbb{D}_r(0))}\leq\frac{c}{t^{\frac{2}{3}}}\text{ for all }t\geq t_0.
\end{equation*} 
\end{prop}
On the other hand, $\partial\mathbb{D}_1(0)$ has to be investigated more carefully.  From \eqref{confmatch},
\begin{equation*}
  {\bf P}^{(0)}(z)\big({\bf P}^{(\infty)}(z)\big)^{-1}={\bf I}+\mathcal{O}\left(t^{-\frac{2}{3}+2|\Re\nu|}\right),\ \ t\rightarrow\infty,
\end{equation*}
and 
\begin{equation*}
  \nu=\frac{t\varkappa}{2\pi\im} = -\frac{1}{2\pi\im}\ln(1-\gamma)\in\im\mathbb{R},\ \ \gamma\in[0,1).
\end{equation*}
Hence, we can conclude
\begin{prop}\label{DZes:2} For any fixed $\gamma\in[0,1)$, there exist positive constants $t_0=t_0(\gamma)$ and $c=c(\gamma)$
such that
\begin{equation*}
  \|{\bf G}_R(\cdot;s,\gamma)-{\bf I}\|_{L^2\cap L^{\infty}(\partial\mathbb{D}_r(0))}\leq\frac{c}{t^{\frac{2}{3}}}\text{ for all }t\geq t_0.
\end{equation*} 
\end{prop}
The last Proposition, together with Proposition \ref{DZes:1}, ensures solvability of RHP \ref{ratio1} by general theory (see \cite{DZ}).
\begin{theo}\label{mainDZ} For any fixed $\gamma\in[0,1)$, there exists $t_0=t_0(\gamma)>0$ and $c=c(\gamma)>0$ such that the ratio RHP \ref{ratio1} is uniquely solvable in $L^2(\Sigma_R)$ for all $t\geq t_0$. We can compute its solution iteratively via the integral equation
\begin{equation*}
	{\bf R}(z)={\bf I}+\frac{1}{2\pi\im}\int_{\Sigma_R}{\bf R}_-(\lambda)\big({\bf G}_R(\lambda)-{\bf I}\big)\frac{\d\lambda}{\lambda-z},\ \ \ \lambda\in\mathbb{C}\backslash\Sigma_R,
\end{equation*}
using that
\begin{equation*}
	\|{\bf R}_-(\cdot;s,\gamma)-{\bf I}\|_{L^2(\Sigma_R)}\leq \frac{c}{t^{\frac{2}{3}}}\text{ for all }t\geq t_0.
\end{equation*}
\end{theo}
This theorem will be central to the asymptotic evaluation of \eqref{F2-determinant} as $s\rightarrow-\infty$. Before we carry out the underlying computations, we shall first derive two differential identities.

\section{Differential identities}\label{diffidentity}
We connect the two logarithmic derivatives
\begin{equation*}
	\frac{\partial}{\partial s}\ln F_2(s,\gamma),\ \ \ \ \textnormal{respectively} \ \frac{\partial}{\partial\gamma}\ln F_2(s,\gamma),
\end{equation*}
with fixed $\gamma$, respectively $s$, to RHP \ref{unifRHP}. In this process it will be useful to recall the following well known
facts about the ``ring of integrable integral operators'' (see \cite{IIKS}). The given kernel displays the structure (compare Remark \ref{AiryConn})
\begin{equation*}
  \gamma K_{\textnormal{Ai}}(\lambda,\mu)=\frac{f^T(\lambda)h(\mu)}{\lambda-\mu};\ \ \ \ \
  f(\lambda)=\binom{\Phi^{11}(\lambda)}{\Phi^{21}(\lambda)},\ \ h(\lambda)=\frac{\gamma}{2\pi\im}\binom{-\Phi^{21}(\lambda)}{\ \ \,\Phi^{11}(\lambda)},
\end{equation*}
and we observe that
\begin{eqnarray}
	\frac{\partial}{\partial s}\ln F_2(s,\gamma)
  &=&-\textnormal{trace}\,\left(\big(1-\gamma K_{\textnormal{Ai}}\upharpoonright_{L^2(s,\infty)}\big)^{-1}\gamma\frac{\partial}{\partial s}\big(K_{\textnormal{Ai}}\upharpoonright_{L^2(s,\infty)}\big)\right)\nonumber\\
  &=&\int_s^{\infty}\big(1-\gamma K_{\textnormal{Ai}}\big)^{-1}(s,\lambda)\gamma K_{\textnormal{Ai}}(\lambda,s)\,\d\lambda=R(s,s);\nonumber\\
  \frac{\partial}{\partial\gamma}\ln F_2(s,\gamma)
  &=&-\textnormal{trace}\,\left(\big(1-\gamma K_{\textnormal{Ai}}\upharpoonright_{L^2(s,\infty)}\big)^{-1}K_{\textnormal{Ai}}\upharpoonright_{L^2(s,\infty)}\right)=-\frac{1}{\gamma}\int_s^{\infty}R(\lambda,\lambda)\,\d\lambda.\label{trace2}
\end{eqnarray}
Here, $1+R\upharpoonright_{L^2(s,\infty)}$ denotes the resolvent of the operator $\gamma K_{\textnormal{Ai}}\upharpoonright_{L^2(s,\infty)}$, i.e. 
\begin{equation*}
	1+R\upharpoonright_{L^2(s,\infty)}=\big(1-\gamma K_{\textnormal{Ai}}\upharpoonright_{L^2(s,\infty)}\big)^{-1},
\end{equation*}
and we have the
kernel representation
\begin{equation*}
  R(\lambda,\mu) = \frac{F^T(\lambda)H(\mu)}{\lambda-\mu},\ \ (\lambda,\mu)\in (s,\infty);\hspace{0.5cm}
  F=\big(1-\gamma K_{\textnormal{Ai}}\upharpoonright_{L^2(s,\infty)}\big)^{-1}f,\ \ H=\big(1-\gamma K_{\textnormal{Ai}}\upharpoonright_{L^2(s,\infty)}\big)^{-1}h.
\end{equation*}
Most importantly, the connection to RHP \ref{master} is of the form
\begin{equation*}
	{\bf Y}(z)={\bf I}-\int_s^{\infty}F(\lambda)h^T(\lambda)\frac{\d\lambda}{\lambda-z};\ \ \ \ F(z)={\bf Y}_+(z)f(z),\ \ H(z)=\big({\bf Y}_+^{-1}\big)^T(z)h(z),\ \ z\in(s,\infty).
\end{equation*}

\begin{prop}\label{diff1} For fixed $\gamma\in\mathbb{R}$, we have
\begin{equation*}
	\frac{\partial}{\partial s}\ln F_2(s,\gamma)=
  \frac{\gamma}{2\pi\im}\big({\bf X}^{-1}(z){\bf X}'(z)\big)^{21}\Big|_{z\rightarrow s},\ \ \ \ (')=\frac{\d}{\d z},
\end{equation*}
in terms of the solution ${\bf X}(z)$ to RHP \ref{unifRHP}, where the limit is taken with $\textnormal{arg}(z-s)\in(0,\frac{2\pi}{3})$.
\end{prop}
\begin{proof} Note that
\begin{equation*}
  F(z) = \binom{X^{11}(z)}{X^{21}(z)},\ \ \ \ H(z)=\frac{\gamma}{2\pi\im}\binom{-X^{21}(z)}{\ \ \,X^{11}(z)},
\end{equation*}
where $X^{11}(z)$ and $X^{21}(z)$ are analytically continued from $\textnormal{arg}(z-s)\in(0,\frac{2\pi}{3})$ to the full complex plane.
The stated identity follows now from
\begin{equation*}
  R(s,s)=\big(F^T\big)'(z)H(z)\Big|_{z\rightarrow s},
\end{equation*}
in which the limit is taken in the same sector in the upper half-plane.
\end{proof}
For the second identity, we start from the central expression used in the derivation of Proposition \ref{diff1},
\begin{equation}\label{d:1}
  R(\lambda,\lambda)=\frac{\gamma}{2\pi\im}\big({\bf X}^{-1}(z){\bf X}'(z)\big)^{21}\Big|_{z\rightarrow\lambda},\ \ \ \lambda\in (s,\infty);\hspace{0.75cm}(')=\frac{\d}{\d z},
\end{equation}
where again the limit is taken with $\textnormal{arg}(z-s)\in(0,\frac{2\pi}{3})$. Next, we replace the derivative
terms by recalling the well-known differential equation associated with the function ${\bf X}(z)$ characterized through
RHP \ref{unifRHP},
\begin{equation*}
  {\bf X}'(z) = \Bigg[\begin{pmatrix}
			    Y_1^{12} & 1\smallskip\\ z+Y_1^{22}-Y_1^{11} & -Y_1^{12} \end{pmatrix}+\frac{\gamma}{2\pi\im}
\begin{pmatrix}
 -\widehat{X}^{11}(s)\widehat{X}^{21}(s) & \big(\widehat{X}^{11}(s)\big)^2\smallskip\\
-\big(\widehat{X}^{21}(s)\big)^2 & \widehat{X}^{11}(s)\widehat{X}^{21}(s)
\end{pmatrix}\frac{1}{z-s}
\Bigg]{\bf X}(z).
\end{equation*}
This equation follows directly from the observation that ${\bf X}'{\bf X}^{-1}$ is single valued and analytic for 
$z\in\mathbb{C}\backslash\{s\}$, and hence can be computed from the asymptotics $z\rightarrow s,\infty$ by Liouville's 
Theorem. Back in \eqref{d:1}, we obtain
\begin{eqnarray}
  R(\lambda,\lambda)&=&\frac{\gamma}{2\pi\im}\Bigg\{-2X^{11}(\lambda)X^{21}(\lambda)\left[Y_1^{12}-\frac{\gamma}{2\pi\im}
  \frac{\widehat{X}^{11}(s)\widehat{X}^{21}(s)}{\lambda-s}\right]-\big(X^{21}(\lambda)\big)^2
  \left[1+\frac{\gamma}{2\pi\im}\frac{\big(\widehat{X}^{11}(s)\big)^2}{\lambda-s}\right]\nonumber\\
  &&\hspace{.4in}+\big(X^{11}(\lambda)\big)^2\left[\lambda+Y_1^{22}-Y_1^{11}-
  \frac{\gamma}{2\pi\im}\frac{\big(\widehat{X}^{21}(s)\big)^2}{\lambda-s}\right]\Bigg\},\label{trace3}
\end{eqnarray}
and we now introduce the auxiliary function
\begin{equation}\label{d:2}
  {\bf K}(z):=\frac{1}{2\pi\im}\int_s^{\infty}{\bf X}(\lambda)\bigl(\begin{smallmatrix}
                                                  0 & -1\\ 0 & 0
                                                 \end{smallmatrix}\bigr){\bf X}^{-1}(\lambda)\frac{\d\lambda}{\lambda-z},\ \ 
  z\in\mathbb{C}\backslash[s,\infty).
\end{equation}
In view of RHP \ref{unifRHP}, ${\bf K}(z)$ is well defined and we have the important identity
\begin{equation*}
  {\bf K}(z)=\bigg(\frac{\partial}{\partial\gamma}{\bf X}(z)\bigg){\bf X}^{-1}(z)=\bigg(\frac{\partial}{\partial\gamma}{\bf Y}(z)\bigg){\bf Y}^{-1}(z),\ \ \ z\in\mathbb{C}\backslash[s,\infty).
\end{equation*}
Thus, on one hand, compare again RHP \ref{unifRHP},
\begin{equation}\label{k:1}
  {\bf K}(z)={\bf Y}_{1,\gamma}\,z^{-1}+\big({\bf Y}_{2,\gamma}-{\bf Y}_{1,\gamma}{\bf Y}_{1}\big)z^{-2}+
  \mathcal{O}\left(z^{-3}\right),\ \ z\rightarrow\infty,\ z\notin[s,\infty);\ \ \ {\bf Y}_{\ell,\gamma}:=\frac{\partial}{\partial\gamma}{\bf Y}_{\ell},
\end{equation}
and, on the other hand, from \eqref{d:2} in the same limit,
\begin{equation}\label{k:2}
  {\bf K}(z)=\frac{1}{z}\frac{\im}{2\pi}\int_s^{\infty}{\bf X}(\lambda)\bigl(\begin{smallmatrix}
                                                                  0 & -1\\ 0 & 0
                                                                 \end{smallmatrix}\bigr){\bf X}^{-1}(\lambda)\,\d\lambda
+\frac{1}{z^2}\frac{\im}{2\pi}\int_s^{\infty}\lambda\,{\bf X}(\lambda)\bigl(\begin{smallmatrix}
                                                                  0 & -1\\ 0 & 0
                                                                 \end{smallmatrix}\bigr){\bf X}^{-1}(\lambda)\,\d\lambda
+\mathcal{O}\left(z^{-3}\right).
\end{equation}
Hence, comparing entries in \eqref{k:1} and \eqref{k:2}, one finds
\begin{equation*}
  Y_{1,\gamma}^{11}=\frac{\im}{2\pi}\int_s^{\infty}X^{11}(\lambda)X^{21}(\lambda)\,\d\lambda;\ \ \ \ \ \ \ \ 
  Y_{1,\gamma}^{21}=\frac{\im}{2\pi}\int_s^{\infty}\big(X^{21}(\lambda)\big)^2\,\d\lambda;
\end{equation*}
\begin{equation*}
  -Y_{1,\gamma}^{12}=\frac{\im}{2\pi}\int_s^{\infty}\big(X^{11}(\lambda)\big)^2\,\d\lambda;\ \ \ \ \ \ \ \ \ 
  \big({\bf Y}_{1,\gamma}{\bf Y}_{1}\big)^{12}-Y_{2,\gamma}^{12}=\frac{\im}{2\pi}
  \int_s^{\infty}\lambda\,\big(X^{11}(\lambda)\big)^2\,\d\lambda.
\end{equation*}
These four identities are useful in the computation of \eqref{trace2} once we use \eqref{trace3}. For the outstanding pieces we note that
\begin{equation*}
  {\bf X}(\lambda)\bigl(\begin{smallmatrix}
        0 & -1\\ 0 & 0
    \end{smallmatrix}\bigr){\bf X}^{-1}(\lambda)=\widehat{{\bf X}}(s)\bigl(\begin{smallmatrix}
0 & -1\\ 0 & 0 \end{smallmatrix}\bigr)\widehat{{\bf X}}^{-1}(s)+\mathcal{O}\left(\lambda-s\right)
\end{equation*}
as $\lambda\rightarrow s$ such that $\textnormal{arg}(\lambda-s)\in(0,\frac{2\pi}{3})\cup(\frac{4\pi}{3},2\pi)$. Hence, in the same limit,
\begin{equation*}
  2X^{11}(\lambda)X^{21}(\lambda)\widehat{X}^{11}(s)\widehat{X}^{21}(s)-\big(X^{21}(\lambda)\big)^2\big(\widehat{X}^{11}(s)\big)^2
-\big(X^{11}(\lambda)\big)^2\big(\widehat{X}^{21}(s)\big)^2=\mathcal{O}\left(\lambda-s\right).
\end{equation*}
Define now the auxiliary function
\begin{eqnarray*}
	\widehat{{\bf K}}(z)&:=&{\bf K}(z)-\frac{1}{2\pi\im}\ln(z-s)\widehat{{\bf X}}(s)\bigl(\begin{smallmatrix} 0 & 1\\ 0 & 0 \end{smallmatrix}\bigr)\widehat{{\bf X}}^{-1}(s),\ \ \ \textnormal{arg}(z-s)\in(0,2\pi)\\
	&=&\frac{1}{2\pi\im}\int_s^{\infty}
	\left[{\bf X}(\lambda)\bigl(\begin{smallmatrix} 0 & -1\\
	 0 & 0 \end{smallmatrix}\bigr){\bf X}^{-1}(\lambda)-\widehat{{\bf X}}(s)\bigl(\begin{smallmatrix} 0 & -1\\
	 0 & 0 \end{smallmatrix}\bigr)\widehat{{\bf X}}^{-1}(s)\right]\frac{\d\lambda}{\lambda-z},\ \ z\in\mathbb{C}\backslash(s,\infty)
\end{eqnarray*}
which is analytic at $z=s$, satisfying
\begin{equation*}
	\widehat{{\bf K}}(z)=\left(\frac{\partial}{\partial\gamma}\widehat{{\bf X}}(z)\right)\widehat{{\bf X}}^{-1}(z),\ \ z\rightarrow s.
\end{equation*}
Moreover, it provides us with the following identity:
\begin{align*}
	&\frac{1}{2\pi\im}\int_s^{\infty}\left[2X^{11}(\lambda)X^{21}(\lambda)\widehat{X}^{11}(s)\widehat{X}^{21}(s)-\big(X^{21}(\lambda)\big)^2\big(\widehat{X}^{11}(s)\big)^2-\big(X^{11}(\lambda)\big)^2\big(\widehat{X}^{21}(s)\big)^2\right]\frac{\d\lambda}{\lambda-s}\\
	=\,\,2&\widehat{X}^{11}(s)\widehat{X}^{21}(s)\widehat{K}^{11}(s)-\big(\widehat{X}^{11}(s)\big)^2\widehat{K}^{21}(s)+\big(\widehat{X}^{21}(s)\big)^2\widehat{K}^{12}(s);\ \ \ \widehat{K}^{jk}(s):=\left(\left(\frac{\partial}{\partial\gamma}\widehat{{\bf X}}(s)\right)\widehat{{\bf X}}^{-1}(s)\right)^{jk}.
\end{align*}
We summarize.
\begin{prop}\label{diff2} For fixed $s<0$, we have
\begin{eqnarray*}
  \frac{\partial}{\partial\gamma}\ln F_2(s,\gamma)&=& -Y_{1,\gamma}^{11}Y_1^{12}+Y_{1,\gamma}^{12}Y_1^{11}-Y_{1,\gamma}^{21}-Y_{2,\gamma}^{12}\\
&&+\frac{\gamma}{2\pi\im}\left(\big(\widehat{X}^{11}(s)\big)^2\widehat{K}^{21}(s)-\big(\widehat{X}^{21}(s)\big)^2\widehat{K}^{12}(s)
  -2\widehat{X}^{11}(s)\widehat{X}^{21}(s)\widehat{K}^{11}(s)\right)
\end{eqnarray*}
in terms of the asymptotic data ${\bf Y}_1$ and ${\bf Y}_2$ given in \eqref{masterASY} as well as the singular data $\widehat{{\bf X}}(z)$ from \eqref{Xsing}.
\end{prop}
\section{Extraction of asymptotics via Proposition \ref{diff1}}\label{sderivsec}
In order to obtain structural information on the large negative $s$ behavior of $F_2(s,\gamma)$, we first use Proposition \ref{diff1}. Recall to this end the relevant transformations 
\begin{equation}\label{AS:0}
  {\bf Y}(z)\mapsto {\bf X}(z)\mapsto {\bf T}(z)\mapsto {\bf S}(z)\mapsto {\bf L}(z)\mapsto {\bf R}(z),
\end{equation}
and obtain in the first step with \eqref{Tdef}, \eqref{Sdef}, and \eqref{l:1},
\begin{equation*}
 	 \frac{\partial}{\partial s}\ln F_2(s,\gamma)=\frac{\gamma}{2\pi\im|s|}\e^{-2tg_+(0)}\Big({\bf L}^{-1}(\lambda){\bf L}'(\lambda)\Big)^{21}\bigg|_{\lambda\rightarrow 0},
 \end{equation*}
where the limit is taken with $\textnormal{arg}\,\lambda\in(\frac{\pi}{3},\frac{2\pi}{3})$. After that, with \eqref{Rdef},
\begin{align*}
	\frac{\partial}{\partial s}\ln F_2(s,\gamma)= & \frac{\gamma}{2\pi\im|s|}\e^{-2tg_+(0)}\Big(\big({\bf P}^{(0)}(\lambda)\big)^{-1}\big({\bf R}(0)\big)^{-1}{\bf R}'(0){\bf P}^{(0)}(\lambda)\Big)^{21}\bigg|_{\lambda\rightarrow 0}\\
	&+\frac{\gamma}{2\pi\im|s|}\e^{-2tg_+(0)}\Big(\big({\bf P}^{(0)}(\lambda)\big)^{-1}\big({\bf P}^{(0)}(\lambda)\big)'\Big)^{21}\bigg|_{\lambda\rightarrow 0}
\end{align*}
 using the same limit convention in both terms. The second term is computed explicitly using \eqref{o:5}, \eqref{o:6}, and \eqref{o:7},
 \begin{equation}\label{2nd}
 	\frac{\gamma}{2\pi\im|s|}\e^{-2tg_+(0)}\Big(\big({\bf P}^{(0)}(\lambda)\big)^{-1}\big({\bf P}^{(0)}(\lambda)\big)'\Big)^{21}\bigg|_{\lambda\rightarrow 0}=\frac{v}{\pi}|s|^{\frac{1}{2}}+\frac{1}{|s|}\left[-\frac{3v^2}{8\pi^2}+\frac{v}{4\pi}\cos\big(2\phi(s,\gamma)\big)\right],
\end{equation}
where we use the parameter $v=-\ln(1-\gamma)\in[0,+\infty)$ and the phase function
\begin{equation}\label{phase}
	\phi(s,\gamma):=\frac{2}{3}|s|^{\frac{3}{2}}-\frac{v}{2\pi}\ln\big(8|s|^{\frac{3}{2}}\big)-\textnormal{arg}\,\Gamma\left(-\frac{iv}{2\pi}\right).
\end{equation}
For the first term we need Theorem \ref{mainDZ}.  More precisely, we write
\begin{equation}\label{Rfact}
	{\bf R}(z)=|s|^{-\frac{1}{4}\sigma_3}\widehat{{\bf R}}(z)|s|^{\frac{1}{4}\sigma_3},\ z\in\mathbb{C}\backslash\Sigma_R; \ \ \ \ \ \  \|\widehat{{\bf R}}_-(\cdot;s,\gamma)-{\bf I}\|_{L^2(\Sigma_R)}\leq\frac{c}{t}\text{ for all }t\geq t_0,
\end{equation}
and then deduce
\begin{equation}\label{1st}
	\frac{\gamma}{2\pi\im|s|}\e^{-2tg_+(0)}\Big(\big({\bf P}^{(0)}(\lambda)\big)^{-1}\big({\bf R}(0)\big)^{-1}{\bf R}'(0){\bf P}^{(0)}(\lambda)\Big)^{21}\bigg|_{\lambda\rightarrow 0}=\mathcal{O}\left(|s|^{-\frac{5}{2}}\right),\ \ s\rightarrow-\infty.
\end{equation}
At this point we can combine \eqref{2nd} with \eqref{1st} to find
\eq
\label{partial-s-F2}
\frac{\partial}{\partial s}\ln F_2(s,\gamma) = \frac{v}{\pi}|s|^{\frac{1}{2}}+\frac{1}{|s|}\left[-\frac{3v^2}{8\pi^2}+\frac{v}{4\pi}\cos\big(2\phi(s,\gamma)\big)\right] + \mathcal{O}\left(|s|^{-\frac{5}{2}}\right),\ \ s\rightarrow-\infty.
\endeq
We integrate indefinitely with respect to $s$ to obtain the next result.
\begin{prop}\label{Eres:1} For any fixed $\gamma\in[0,1)$, there exist positive constants $t_0=t_0(\gamma)$ and $c=c(v)$ such that
\begin{equation}\label{Eexp:1}
	\ln F_2(s,\gamma)=-\frac{2v}{3\pi}|s|^{\frac{3}{2}}+\frac{3v^2}{8\pi^2}\ln|s|+D(\gamma)+r(s,v)
\end{equation}
for $t=|s|^{\frac{3}{2}}\geq t_0$ with $v=-\ln(1-\gamma)\in[0,+\infty)$. The term $D(\gamma)$ is independent of $s$, and the error term $r(s,v)$ satisfies
\begin{equation*}
	\big|r(s,v)\big|\leq \frac{c(v)}{|s|^{\frac{3}{2}}}\ \ \ \ \ \ \forall\,t\geq t_0.
\end{equation*}
\end{prop}
Our next goal is the evaluation of the term $D(\gamma)$ in Proposition \ref{Eres:1}. This will be achieved in the following section.
\section{Extraction of asymptotics via Proposition \ref{diff2}}
\label{gamma-deriv-sec}
We start from Proposition \ref{diff2}.  For fixed $s<0$,
\begin{equation*}
	\frac{\partial}{\partial\gamma}\ln F_2(s,\gamma)=T_1(s,\gamma)+T_2(s,\gamma)+T_3(s,\gamma)
\end{equation*}
where we put for simplicity
\begin{equation*}
	T_1(s,\gamma):=-Y_{1,\gamma}^{11}Y_1^{12}+Y_{1,\gamma}^{12}Y_1^{11}-Y_{1,\gamma}^{21},\hspace{1cm}T_2(s,\gamma):=-Y_{2,\gamma}^{12},
\end{equation*}
and
\begin{equation}\label{T3def}
	T_3(s,\gamma):=\frac{\gamma}{2\pi\im}\left(\big(\widehat{X}^{11}(s)\big)^2\widehat{K}^{21}(s)-\big(\widehat{X}^{21}(s)\big)^2\widehat{K}^{12}(s)
  -2\widehat{X}^{11}(s)\widehat{X}^{21}(s)\widehat{K}^{11}(s)\right).
 \end{equation}
 The idea is now to first evaluate all three $T_j(s,\gamma)$ asymptotically as $s\rightarrow-\infty$ with fixed $\gamma\in[0,1)$ and after that perform a definite integration with respect to $\gamma$,
 \begin{equation}\label{gamidea}
 	\gamma\in[0,1):\ \ \ \ \ \ln F_2(s,\gamma)=\int_0^{\gamma}\frac{\partial}{\partial\gamma'}\ln F_2(s,\gamma')\d\gamma'=\int_0^{\gamma}\big(T_1(s,\gamma')+T_2(s,\gamma')+T_3(s,\gamma')\big)\d\gamma'.
\end{equation}
Comparing the so-obtained expansion for $\ln F_2(s,\gamma)$ to \eqref{Eexp:1}, we will obtain the unknown $D(\gamma)$. 
\subsection{Computation of \texorpdfstring{$T_1(s,\gamma)$}{T1}}
We begin with the asymptotic evaluation of $T_1(s,\gamma)$ for which we trace back the  transformations 
\begin{equation*}
	{\bf Y}(z)\mapsto {\bf X}(z)\mapsto {\bf T}(z)\mapsto {\bf S}(z)\mapsto {\bf L}(z).
\end{equation*}
This provides us with the following explicit identity for ${\bf Y}_1$ (compare \eqref{masterASY}):
\begin{equation}\label{Y1exp}
{\bf Y}_1=|s|\begin{pmatrix}-2\nu^2 &-2\im\nu|s|^{-\frac{1}{2}}\\ \frac{7}{48}|s|^{-1}+\frac{2}{3}\im\nu|s|^{\frac{1}{2}}(1-4\nu^2) & 2\nu^2\end{pmatrix}+|s|\begin{pmatrix}1&0\\ 2\im\nu|s|^{\frac{1}{2}} & 1\end{pmatrix}
	{\bf J}_1\begin{pmatrix}1&0\\ -2\im\nu|s|^{\frac{1}{2}} & 1\end{pmatrix},
\end{equation}
where
\begin{equation*}
	{\bf J}_1:=\frac{\im}{2\pi}\int_{\Sigma_R}{\bf R}_-(\lambda)\big({\bf G}_R(\lambda)-{\bf I}\big)\d\lambda.
\end{equation*}
The last integral is then computed asymptotically with the help of Theorem \ref{mainDZ}, which amounts to a standard residue computation using \eqref{Aimatch} and \eqref{confmatch}. We summarize the results below.
\begin{lem}\label{lem1} As $s\rightarrow-\infty$,
\begin{equation*}
	{\bf J}_1=|s|^{-\frac{1}{4}\sigma_3}\left\{\frac{1}{48t}\begin{pmatrix}-24\im\nu & -48\nu^2\\ -7&24\im\nu\end{pmatrix}+\frac{1}{2t}\begin{pmatrix}\im\nu\sin(2\phi)& -\nu^2-\im\nu\cos(2\phi)\\ \nu^2-\im\nu\cos(2\phi) & -\im\nu\sin(2\phi)\end{pmatrix}+\mathcal{O}\left(t^{-2}\right)\right\}|s|^{\frac{1}{4}\sigma_3}
\end{equation*}
with $\phi=\phi(s,\gamma)$ as in \eqref{phase}, $t=(-s)^{\frac{3}{2}}$.  The error term is uniform with respect to $\gamma$ chosen from compact subsets of $[0,1)$, and it can be differentiated with respect to $\gamma$. In fact, after differentiation (with respect to $\gamma$), the error term is of order $\mathcal{O}\left(t^{-2}\ln t\right)$.
\end{lem}
Substituting the result of this lemma into \eqref{Y1exp} gives:
\begin{lem}\label{lem2} As $s\rightarrow-\infty$, with the same statements about the error terms as in Lemma \ref{lem1},
\begin{equation*}
	Y_1^{11}=-2\nu^2|s|-\frac{\im\nu}{2|s|^{\frac{1}{2}}}+\frac{3\im\nu^3}{|s|^{\frac{1}{2}}}+\frac{1}{2|s|^{\frac{1}{2}}}\big(\im\nu\sin(2\phi)-2\nu^2\cos(2\phi)\big)+\mathcal{O}\left(|s|^{-2}\right),
\end{equation*}
followed by
\begin{equation*}
	Y_1^{12}=-2\im\nu|s|^{\frac{1}{2}}-\frac{3\nu^2}{2|s|}-\frac{\im\nu}{2|s|}\cos(2\phi)+\mathcal{O}\big(|s|^{-\frac{5}{2}}\big),
\end{equation*}
and concluding with
\begin{equation*}
	Y_1^{21}=\frac{2}{3}\im\nu|s|^{\frac{3}{2}}(1-4\nu^2)+\frac{5\nu^2}{2}-6\nu^4-2\nu^2\sin(2\phi)-\frac{\im\nu}{2}\cos(2\phi)-2\im\nu^3\cos(2\phi)+\mathcal{O}\big(|s|^{-\frac{3}{2}}\big).
\end{equation*}
\end{lem}
At this point, we only need to combine the results of Lemma \ref{lem2} with the definition of $T_1(s,\gamma)$.  
\begin{prop}\label{propT1} As $s\rightarrow-\infty$,
\begin{align*}
	T_1(s,\gamma)=&-\frac{2}{3}\im\nu_{\gamma}|s|^{\frac{3}{2}}+4\im\nu^2\nu_{\gamma}|s|^{\frac{3}{2}}+12\nu^3\nu_{\gamma}-5\nu\nu_{\gamma}+\cos(2\phi)\left\{2\nu^2\phi_{\gamma}+3\im\nu^2\nu_{\gamma}+\frac{1}{2}\im\nu_{\gamma}\right\}\\
	&+\sin(2\phi)\left\{-2\im\nu^3\phi_{\gamma}+4\nu\nu_{\gamma}-\im\nu\phi_{\gamma}\right\}+\mathcal{O}\left(|s|^{-\frac{3}{2}}\ln|s|\right),
\end{align*}
and the error term is uniform with respect to $\gamma$ chosen from compact subsets of $[0,1)$.
\end{prop}
\subsection{Computation of \texorpdfstring{$T_2(s,\gamma)$}{T2}} Our strategy is the same as in the computation of $T_1(s,\gamma)$, however certain steps are more involved. First, after tracing back transformations, we obtain the exact identity
\begin{align}
	{\bf Y}_2=&\,|s|^2\begin{pmatrix}\frac{2}{3}\nu^2(\nu^2-1)-\frac{7\im\nu}{24}|s|^{-\frac{3}{2}} & \frac{4\im\nu}{3}(\nu^2-1)|s|^{-\frac{1}{2}}-\frac{5}{48}|s|^{-2}\\ \frac{4\im\nu}{15}(1-\nu^2)(1-4\nu^2)|s|^{\frac{1}{2}}+\frac{7}{48}(1+2\nu^2)|s|^{-1}&-\frac{2}{3}\nu^2(\nu^2-1)\end{pmatrix}-|s|{\bf Y}_1\nonumber\\
	&+|s|^2\begin{pmatrix}1&0\\ 2\im\nu|s|^{\frac{1}{2}}\end{pmatrix}{\bf J}_1\begin{pmatrix}1&0\\ -2\im\nu|s|^{\frac{1}{2}} & 1\end{pmatrix}\begin{pmatrix}-2\nu^2&-2\im\nu|s|^{-\frac{1}{2}}\\ \frac{7}{48}|s|^{-1}+\frac{2}{3}\im\nu|s|^{\frac{1}{2}}(1-4\nu^2) & 2\nu^2\end{pmatrix}\nonumber\\
	&+|s|^2\begin{pmatrix}1&0\\ 2\im\nu|s|^{\frac{1}{2}} & 1\end{pmatrix}{\bf J}_2\begin{pmatrix}1&0\\ -2\im\nu|s|^{\frac{1}{2}} & 1\end{pmatrix},\label{Y2exp}
\end{align}
where
\begin{equation*}
	{\bf J}_2:=\frac{\im}{2\pi}\int_{\Sigma_R}{\bf R}_-(\lambda)\big({\bf G}_R(\lambda)-{\bf I}\big)\lambda\,\d\lambda,
\end{equation*}
and ${\bf Y}_1$ and ${\bf J}_1$ have appeared previously in \eqref{Y1exp}. After that, we have the following analogue of Lemma \ref{lem1}.
\begin{lem} As $s\rightarrow-\infty$,
\begin{equation*}
	{\bf J}_2=|s|^{-\frac{1}{4}\sigma_3}\left\{\frac{1}{48t}\begin{pmatrix}-24\im\nu & 5-48\nu^2\\ -7& 24\im\nu\end{pmatrix}+\mathcal{O}\left(t^{-2}\right)\right\}|s|^{\frac{1}{4}\sigma_3},
\end{equation*}
and the error term is uniform with respect to $\gamma$ chosen from compact subsets of $[0,1)$.
\end{lem}
Now we substitute all obtained formul\ae\, into \eqref{Y2exp}, which leads to
\begin{lem} As $s\rightarrow-\infty$,
\begin{equation*}
	Y_2^{12}=\frac{4}{3}\im\nu^3|s|^{\frac{3}{2}}+\frac{2}{3}\im\nu|s|^{\frac{3}{2}}-\frac{1}{2}\nu^2+3\nu^4+\cos(2\phi)\left\{\frac{1}{2}\im\nu+\im\nu^3\right\}+\nu^2\sin(2\phi)+\mathcal{O}\left(|s|^{-\frac{3}{2}}\right),
\end{equation*}
where the error term is uniform with respect to $\gamma$ chosen from compact subsets of $[0,1)$ and can be differentiated with respect to $\gamma$.
\end{lem}
At this point, we are left to summarize our current results.
\begin{prop}\label{propT2} As $s\rightarrow-\infty$,
\begin{align*}
	T_2(s,\gamma)=&-\frac{2}{3}\im\nu_{\gamma}|s|^{\frac{3}{2}}-4\im\nu^2\nu_{\gamma}|s|^{\frac{3}{2}}-12\nu^3\nu_{\gamma}+\nu\nu_{\gamma}+\cos(2\phi)\left\{-2\nu^2\phi_{\gamma}-3\im\nu^2\nu_{\gamma}-\frac{1}{2}\im\nu_{\gamma}\right\}\\
	&+\sin(2\phi)\left\{2\im\nu^3\phi_{\gamma}-2\nu\nu_{\gamma}+\im\nu \phi_{\gamma}\right\}+\mathcal{O}\left(|s|^{-\frac{3}{2}}\ln|s|\right),
\end{align*}
and the error term is uniform with respect to $\gamma$ chosen from compact subsets of $[0,1)$.
\end{prop}
\subsection{Computation of \texorpdfstring{$T_3(s,\gamma)$}{T3}} For the last part we require the following exact identity, compare \eqref{Xsing}, \eqref{Tdef}, \eqref{Sdef}, \eqref{l:1}, \eqref{o:5}, \eqref{o:6}, \eqref{Rdef}:
\begin{equation}\label{Xhats}
	\widehat{{\bf X}}(s)=\begin{pmatrix}1&0\\ 2\im\nu|s|^{\frac{1}{2}} & 1\end{pmatrix}{\bf R}(0){\bf E}^{(0)}(0)\widehat{\Psi}(0)\begin{pmatrix}1&\frac{\gamma}{2\pi\im}\ln(2|s|^{\frac{1}{2}})\\ 0&1\end{pmatrix}.
\end{equation}
This formula implies at once that $\det\widehat{{\bf X}}(s)=1$, which in turn leads to a simplified identity for $T_3(s,\gamma)$.
\begin{lem}\label{lemsim} We have for $T_3(s,\gamma)$ as defined in \eqref{T3def}, see also Proposition \ref{diff2},
\begin{equation*}
	T_3(s,\gamma)=\frac{\gamma}{2\pi\im}\Big(\widehat{X}^{11}(s)\widehat{X}^{21}_{\gamma}(s)-\widehat{X}^{11}_{\gamma}(s)\widehat{X}^{21}(s)\Big).
\end{equation*}
\end{lem}
\begin{proof} Note that
\begin{equation*}
	\widehat{K}^{11}(s)=\widehat{X}_{\gamma}^{11}(s)\widehat{X}^{22}(s)-\widehat{X}_{\gamma}^{12}(s)\widehat{X}^{21}(s),\ \ \ \ \ \ \widehat{K}^{12}(s)=-\widehat{X}^{11}_{\gamma}(s)\widehat{X}^{12}(s)+\widehat{X}^{12}_{\gamma}(s)\widehat{X}^{11}(s),
\end{equation*}
and
\begin{equation*}
	\widehat{K}^{21}(s)=\widehat{X}^{21}_{\gamma}(s)\widehat{X}^{22}(s)-\widehat{X}^{22}_{\gamma}(s)\widehat{X}^{21}(s).
\end{equation*}
If we now use $\widehat{X}^{11}(s)\widehat{X}^{22}(s)-\widehat{X}^{12}(s)\widehat{X}^{21}(s)\equiv 1$, which leads to 
\begin{equation*}
	\widehat{X}_{\gamma}^{11}(s)\widehat{X}^{22}(s)-\widehat{X}_{\gamma}^{12}(s)\widehat{X}^{21}(s)=-\big(\widehat{X}^{11}(s)\widehat{X}_{\gamma}^{22}(s)-\widehat{X}^{12}(s)\widehat{X}_{\gamma}^{21}(s)\big),
\end{equation*}
then the identity follows directly from \eqref{T3def} after simplification.
\end{proof}
We now evaluate the outstanding matrix elements $\widehat{X}^{11}(s)$ and $\widehat{X}^{21}(s)$ through \eqref{Xhats} by referring once more to Theorem \ref{mainDZ}:
\begin{lem} For any fixed $\gamma\in[0,1)$, as $s\rightarrow-\infty$,
\begin{equation*}
	{\bf R}(0)=|s|^{-\frac{1}{4}\sigma_3}\left\{{\bf I}+\mathcal{O}\left(t^{-1}\right)\right\}|s|^{\frac{1}{4}\sigma_3}.
\end{equation*}
\end{lem}
Thus, after simplification, using in particular the phase function $\phi=\phi(s,\gamma)$ in \eqref{phase},
\begin{lem}\label{lemXhat} As $s\rightarrow-\infty$,
\begin{equation*}
	\widehat{X}^{11}(s)=|s|^{-\frac{1}{4}}\left\{\sqrt{\frac{\im\pi\nu}{\gamma}}\,\big(\e^{\im\phi}-\im\e^{-\im\phi}\big)+\mathcal{O}\left(t^{-1}\right)\right\}
\end{equation*}
and
\begin{equation*}
	\widehat{X}^{21}(s)=|s|^{\frac{1}{4}}\left\{\sqrt{\frac{\im\pi\nu}{\gamma}}\,\big(\im(2\nu-1)\e^{\im\phi}+(2\nu+1)\e^{-\im\phi}\big)+\mathcal{O}\left(t^{-1}\right)\right\},
\end{equation*}
where all error terms are uniform with respect to $\gamma$ chosen from compact subsets of $[0,1)$. Again, they are also differentiable with respect to $\gamma$, subject to error correction.
\end{lem}
The final step of our computation consists in combining Lemma \ref{lemXhat} with Lemma \ref{lemsim}.  The result is as follows.
\begin{prop}\label{propT3} As $s\rightarrow-\infty$,
\begin{equation*}
	T_3(s,\gamma)=-2\nu\nu_{\gamma}\sin(2\phi)+2\nu\nu_{\gamma}-2\im\nu\phi_{\gamma}+\mathcal{O}\left(|s|^{-\frac{3}{2}}\ln|s|\right)
\end{equation*}
with an error term that is uniform with respect to $\gamma$ chosen from compact subsets of $[0,1)$.
\end{prop}
\subsection{Evaluation of \texorpdfstring{$D(\gamma)$}{D}} The final expression for $D(\gamma)$ is obtained by combining Propositions \ref{propT1}, \ref{propT2}, and \ref{propT3}.  We have
\eq
\label{partial-gamma-F2}
\begin{split}
\frac{\partial}{\partial\gamma}\ln F_2(s,\gamma) = & T_1(s,\gamma)+T_2(s,\gamma)+T_3(s,\gamma) \\ 
        =&-\frac{4\im}{3}\nu_{\gamma}|s|^{\frac{3}{2}}-2\nu\nu_{\gamma}-2\im\nu\phi_{\gamma}+\mathcal{O}\left(|s|^{-\frac{3}{2}}\ln|s|\right)\\
	=&-\frac{2v_{\gamma}}{3\pi}|s|^{\frac{3}{2}}+\frac{vv_{\gamma}}{2\pi^2}\ln\big(8|s|^{\frac{3}{2}}\big)+\frac{vv_{\gamma}}{2\pi^2}+\frac{v}{\pi}\frac{\d}{\d\gamma}\textnormal{arg}\,\Gamma\left(-\frac{\im v}{2\pi}\right)+\mathcal{O}\left(|s|^{-\frac{3}{2}}\ln|s|\right),
\end{split}
\endeq
where we used the definitions of $\phi=\phi(s,\gamma)$ from \eqref{phase} and $v=-\ln(1-\gamma)=2\pi\im\nu$ in the last equality. Since $v|_{\gamma=0}=0$, we now integrate and derive, compare \eqref{gamidea},
\begin{equation*}
	\ln F_2(s,\gamma)=-\frac{2v}{3\pi}|s|^{\frac{3}{2}}+\frac{v^2}{4\pi^2}\ln\big(8|s|^{\frac{3}{2}}\big)+\frac{v^2}{4\pi^2}-\frac{1}{\pi}\int_0^{\gamma}v(\gamma')\left\{\frac{\d}{\d\gamma'}\,\textnormal{arg}\,\Gamma\left(\frac{\im v(\gamma')}{2\pi}\right)\right\}\d\gamma'+\mathcal{O}\left(|s|^{-\frac{3}{2}}\ln|s|\right).
\end{equation*}
This last identity can then be further simplified by referring to the Barnes $G$-function \cite{NIST}. This special function satisfies the following useful identity
\begin{equation}\label{barnes}
	z\in\mathbb{C}:\ \Re z>-1,\ \ \ \  \ \int_0^z\ln\Gamma(1+x)\,\d x=\frac{z}{2}\ln(2\pi)-\frac{z}{2}(z+1)+z\ln\Gamma(1+z)-\ln G(1+z),
\end{equation}
and allows us to obtain the next proposition.
\begin{prop} For any fixed $\gamma\in[0,1)$, with $v=v(\gamma)=-\ln(1-\gamma)$,
\begin{equation*}
	\frac{v^2}{4\pi^2}-\frac{1}{\pi}\int_0^{\gamma}v(\gamma')\left\{\frac{\d}{\d\gamma'}\,\textnormal{arg}\,\Gamma\left(\frac{\im v(\gamma')}{2\pi}\right)\right\}\d\gamma'=\ln\big(G\Big(1+\frac{\im v}{2\pi}\Big)G\Big(1-\frac{\im v}{2\pi}\Big)\big).
\end{equation*}
\end{prop}
\begin{proof} Using that
\begin{equation*}
	\Gamma\left(\frac{\im v}{2\pi}\right)\left\{\Gamma\left(-\frac{\im v}{2\pi}\right)\right\}^{-1} = \exp\left[2\im\,\textnormal{arg}\,\Gamma\left(\frac{\im v}{2\pi}\right)\right],
\end{equation*}
we need to evaluate
\begin{equation*}
	\int_0^{\gamma}v(\gamma')\left\{\frac{\d}{\d\gamma'}\,\textnormal{arg}\,\Gamma\left(\frac{\im}{2\pi} v(\gamma')\right)\right\}\d\gamma'=\frac{1}{2\im}\int_0^{\gamma}v(\gamma')\left\{\frac{\d}{\d\gamma'}\ln\frac{\Gamma(1+\frac{\im}{2\pi}v(\gamma'))}{\Gamma(1-\frac{\im}{2\pi}v(\gamma'))}\right\}\d\gamma'.
\end{equation*}
The last integral is exactly of the form \eqref{barnes} after integration by parts, so the stated identity follows at once.
\end{proof}
The last proposition concludes our asymptotic evaluation of $F_2(s,\gamma)$.  We have
\begin{equation}\label{F2final}
	\ln F_2(s,\gamma)=-\frac{2v}{3\pi}|s|^{\frac{3}{2}}+\frac{v^2}{4\pi^2}\ln\big(8|s|^{\frac{3}{2}}\big)+\ln\left(G\Big(1+\frac{\im v}{2\pi}\Big)G\Big(1-\frac{\im v}{2\pi}\Big)\right)+\mathcal{O}\left(|s|^{-\frac{3}{2}}\ln|s|\right).
\end{equation}
Comparing with \eqref{Eexp:1} gives the following.
\begin{cor} For any fixed $\gamma\in[0,1)$, 
\begin{equation*}
	D(\gamma)=\frac{3v^2}{4\pi^2}\ln 2+\ln\left(G\Big(1+\frac{\im v}{2\pi}\Big)G\Big(1-\frac{\im v}{2\pi}\Big)\right),\ \ \ v=-\ln(1-\gamma)\in[0,+\infty).
\end{equation*}
\end{cor}
In the end, combining \eqref{F2final} with Proposition \ref{Eres:1} gives the 
desired result.
\begin{theo}
\label{Eres:2} 
For any fixed $\gamma\in[0,1)$, there exist positive constants $t_0=t_0(\gamma)$ and $c=c(v)$ such that
\begin{equation}\label{Eexp:2}
	\ln F_2(s,\gamma)=-\frac{2v}{3\pi}|s|^{\frac{3}{2}}+\frac{v^2}{4\pi^2}\ln\big(8|s|^{\frac{3}{2}}\big)+\ln\left(G\Big(1+\frac{\im v}{2\pi}\Big)G\Big(1-\frac{\im v}{2\pi}\Big)\right)+r(s,v)
\end{equation}
for $t=|s|^{\frac{3}{2}}\geq t_0$ with $v=-\ln(1-\gamma)\in[0,+\infty)$ and in terms of the Barnes $G$-function. The error term is differentiable with respect to $s$ and satisfies
\begin{equation*}
	\big|r(s,v)\big|\leq\frac{c(v)}{|s|^{\frac{3}{2}}}\ \ \ \ \ \ \forall\,t\geq t_0.
\end{equation*}
\end{theo}
The last theorem was derived under the assumption that $\gamma\in[0,1)$ is kept fixed throughout. However, as we shall prove in the next section, the leading order behavior \eqref{Eexp:2} is still valid as $\gamma\uparrow 1$ at a certain (not too fast) rate.
\section{Extension of Theorem \ref{Eres:2}}
\label{sec:extension}
In order to allow for certain values of $v=-\ln(1-\gamma)\rightarrow+\infty$ we repeat all steps leading to the ratio problem RHP \ref{ratio1}.  However, now special care has to be given to the underlying error estimates. 
\subsection{Preliminary estimates}
Consider ${\bf G}_R(z;s,\gamma)$ for $z\in\partial\mathbb{D}_r(1)$, as $s\rightarrow-\infty$, from \eqref{o:3}, see also \eqref{Aimatch},
\begin{equation}\label{ext:1}
	{\bf G}_R(z;s,\gamma)-{\bf I}={\bf P}^{(\infty)}(z)\e^{-\frac{t}{2}\varkappa\sigma_3}\left\{\frac{1}{48\z^{\frac{3}{2}}(z)}\begin{pmatrix}1&6\im\\ 6\im & -1\end{pmatrix}+{\bf \mathcal{E}}^{(1)}(z;s)\right\}\e^{\frac{t}{2}\varkappa\sigma_3}\big({\bf P}^{(\infty)}(z)\big)^{-1},
\end{equation}
where the error term $\mathcal{E}^{(1)}(z;s)$ is $v$-independent. To estimate the behavior with respect to $v$, we thus only need to recall \eqref{o:1}:
\begin{equation*}
	z\in\partial\mathbb{D}_r(1),\ \ 0<r<\frac{1}{4}:\ \ \ \ \big|\mathcal{D}(z)\e^{\frac{t}{2}\varkappa}\big|=\exp\left[\frac{v}{2\pi}\arctan\left(\frac{2\sqrt{r}\cos(\frac{1}{2}\textnormal{arg}(z-1))}{1-r}\right)\right].
\end{equation*}
This shows that a contracting radius $r=r_t=t^{-\frac{2}{3}+\epsilon}$ with $0<\epsilon<\frac{2}{3}$ in the scaling region
\begin{equation}\label{scale}
	t\geq t_0,\ \ 0\leq v<t^{\frac{1}{3}-\frac{\epsilon}{2}}
\end{equation}
ensures on one end that
\begin{equation*}
	z\in\partial\mathbb{D}_{r_t}(1):\ \ \ \ \big|\mathcal{D}(z)\e^{\frac{t}{2}\varkappa}\big|^{\pm 1}=\exp\left[\pm\frac{v\sqrt{r}}{\pi}\cos\left(\frac{1}{2}\textnormal{arg}(z-1)\right)\right]\big(1+\mathcal{O}\big(vr^{\frac{3}{2}}\big)\big)=\mathcal{O}(1).
\end{equation*}
On the other end, we (still) have for $z\in\partial\mathbb{D}_{r_t}(1)$ that 
\begin{equation*}
	|\z(z)|\geq c_3 t^{\frac{2}{3}}r_t= c_3t^{\epsilon}\rightarrow +\infty\ \ \ \textnormal{as}\ \ t\rightarrow+\infty.
\end{equation*}
Combining these bounds in \eqref{ext:1} gives the following.
\begin{prop}\label{shrink1} For every fixed $0<\epsilon<\frac{2}{3}$, there 
exist $t_0>0$ and $c>0$ such that
\begin{equation*}
	\|{\bf G}_R(\cdot;s,\gamma)-{\bf I}\|_{L^2\cap L^{\infty}(\partial\mathbb{D}_{r_t}(1))}\leq\frac{c}{t^{\epsilon}}\text{ for all }t\geq t_0,\ \ 0\leq v<t^{\frac{1}{3}-\frac{\epsilon}{2}}.
\end{equation*}
\end{prop}
The corresponding argument for $z\in\partial\mathbb{D}_r(0)$ is similar but more involved since $\Psi(\z)$ in \eqref{o:4} already depends on $\gamma$ while $\Phi(\z)$, see \eqref{e:3}, did not. This difference requires us to work with the full asymptotic series in condition (4) of RHP \ref{conflu}.  From \cite{NIST}, as $\z\rightarrow\infty$,
\begin{align*}
	\Psi(\z)\sim&\left[{\bf I}+\e^{\im\frac{\pi}{2}\nu\sigma_3}\sum_{k=1}^{\infty}\begin{pmatrix}((\nu)_k)^2\e^{\im\frac{\pi}{2}k} & ((1-\nu)_{k-1})^2k\,\e^{-\im\frac{\pi}{2}k}\frac{\Gamma(1-\nu)}{\Gamma(\nu)}\\
	((1+\nu)_{k-1})^2k\,\e^{\im\frac{\pi}{2}k}\frac{\Gamma(1+\nu)}{\Gamma(-\nu)} & ((-\nu)_k)^2\e^{-\im\frac{\pi}{2}k}\end{pmatrix}\frac{\z^{-k}}{k!}\e^{-\im\frac{\pi}{2}\nu\sigma_3}\right]\\
	&\times\,\z^{-\nu\sigma_3}\e^{-\frac{\im}{2}\z\sigma_3}\begin{cases}
                                                                                                      \Bigl(\begin{smallmatrix}
                                                                                                             \e^{\im\frac{3\pi}{2}\nu} & 0\\ 0 & \e^{-\im\frac{\pi}{2}\nu}
                                                                                                            \end{smallmatrix}\Bigr),& \textnormal{arg}\,\z\in(0,\pi)\smallskip\\
      \Bigl(\begin{smallmatrix}
             0 & -\e^{\im\frac{\pi}{2}\nu}\\ \e^{\im\frac{\pi}{2}\nu} & 0
            \end{smallmatrix}\Bigr)
,&\textnormal{arg}\,\z\in(-\pi,0),
                                                                                    \end{cases}
\end{align*}
where $(\nu)_k:=\nu(\nu+1)(\nu+2)\cdot\ldots\cdot(\nu+k-1)$ is the Pochhammer symbol and $\nu=2\pi\im v\in\im\mathbb{R}$. Now we use \eqref{o:6} and \eqref{o:8} so that first for $z\in\partial\mathbb{D}_r(0):\textnormal{arg}\,z\in(0,\pi)$, as $s\rightarrow-\infty$,
\begin{align}
	{\bf N}^{-1}&(|s|(z-1))^{\frac{1}{4}\sigma_3}\Big({\bf G}_R(z;s,\gamma)-{\bf I}\Big)(|s|(z-1))^{-\frac{1}{4}\sigma_3}{\bf N}\nonumber\\
	&\sim\,\sum_{k=1}^{\infty}\begin{pmatrix}((\nu)_k)^2\e^{\im\frac{\pi}{2}k} & -\nu((1-\nu)_{k-1})^2k\,\e^{-\im\frac{\pi}{2}k}(f^+(z;s,\gamma))^{-1}\smallskip\\
	\nu((1+\nu)_{k-1})^2k\,\e^{\im\frac{\pi}{2}k}f^+(z;s,\gamma) & ((-\nu)_k)^2\e^{-\im\frac{\pi}{2}k}\end{pmatrix}\frac{\z^{-k}}{k!}\label{low1}
\end{align}
with
\begin{equation*}
	f^+(z;s,\gamma):=\e^{-\im\frac{4}{3}t-2\im\textnormal{arg}\,\Gamma(-\nu)}\e^{2\pi\im\nu}\big(\z(z)\big)^{-2\nu}\big(\mathcal{D}(z)\big)^2,\ \ z\in\partial\mathbb{D}_r(0):\ \textnormal{arg}\,z\in(0,\pi).
\end{equation*}
Note that for $z\in\partial\mathbb{D}_r(0)$,
\begin{align}
	\big|f^+(z;s,\gamma)\big|=&\,\e^v\exp\left[-\frac{v}{\pi}\textnormal{arg}\,\z(z)+\frac{v}{\pi}\textnormal{arg}\,\left(\frac{(z-1)^{\frac{1}{2}}-\im}{(z-1)^{\frac{1}{2}}+\im}\right)\right]\nonumber\\
	=&\,\e^v\exp\left[-\frac{v}{\pi}\textnormal{arg}\,\z(z)-\frac{v}{\pi}\arccos\left(\frac{1}{r}\big\{|z-1|-1\big\}\right)\right]\label{f+1},
\end{align}
and likewise
\begin{equation*}
	|z-1|=\sqrt{1-2r\cos(\textnormal{arg}\,z)+r^2}.
\end{equation*}
Thus, as $r\downarrow 0$ with $\textnormal{arg}\,z\in(0,\pi)$,
\begin{equation*}
	\arccos\left(\frac{1}{r}\big(|z-1|-1\big)\right)=\pi-\textnormal{arg}\,z-\frac{r}{2}\sin(\textnormal{arg}\,z)+\mathcal{O}\left(r^2\right),
\end{equation*}
leading to (compare also \eqref{o:7})
\begin{align}
	\big|f^+(z;s,\gamma)\big|=&\,\exp\left[-\frac{v}{\pi}\textnormal{arg}\left(\frac{\z(z)}{z}\right)+\frac{vr}{2\pi}\sin(\textnormal{arg}\,z)\right]\big(1+\mathcal{O}\left(vr^2\right)\big)\nonumber\\
	=&\,\exp\left[\frac{3vr}{4\pi}\sin(\textnormal{arg}\,z)\right]\big(1+\mathcal{O}\left(vr^2\right)\big)\label{f+2},
\end{align}
provided that $vr\downarrow 0$. The situation is similar in the lower half plane; instead of \eqref{low1} we have now for $z\in\partial\mathbb{D}_r(0):\ \textnormal{arg}\,z\in(-\pi,0)$, as $s\rightarrow-\infty$,
\begin{align}
	{\bf N}^{-1}&(|s|(z-1))^{\frac{1}{4}\sigma_3}\Big({\bf G}_R(z;s,\gamma)-{\bf I}\Big)(|s|(z-1))^{-\frac{1}{4}\sigma_3}{\bf N}\nonumber\\
	&\sim\,\sum_{k=1}^{\infty}\begin{pmatrix}((-\nu)_k)^2\e^{-\im\frac{\pi}{2}k} & -\nu((1+\nu)_{k-1})^2k\,\e^{\im\frac{\pi}{2}k}(f^-(z;s,\gamma))^{-1}\smallskip\\
	\nu((1-\nu)_{k-1})^2k\,\e^{-\im\frac{\pi}{2}k}f^-(z;s,\gamma) & ((\nu)_k)^2\e^{\im\frac{\pi}{2} k}\end{pmatrix}\frac{\z^{-k}}{k!}\label{low2}
\end{align}
with
\begin{equation*}
	f^-(z;s,\gamma):=\e^{\im\frac{4}{3}t+2\im\textnormal{arg}\,\Gamma(-\nu)}\e^{2\pi\im\nu}\big(\z(z)\big)^{2\nu}\big(\mathcal{D}(z)\big)^2,\ \ z\in\partial\mathbb{D}_r(0):\ \textnormal{arg}\,z\in(-\pi,0).
\end{equation*}
Applying the same geometrical reasoning as above, we deduce this time
\begin{equation*}
	\big|f^-(z;s,\gamma)\big|=\exp\left[-\frac{3vr}{4\pi}\sin(\textnormal{arg}\,z)\right]\big(1+\mathcal{O}\left(vr^2\right)\big),\ \ vr\downarrow 0.
\end{equation*}
Thus, here we choose a contracting radius $r=\hat{r}_t=t^{-\frac{1}{3}}$ so that, subject to \eqref{scale},
\begin{equation}\label{rconseq}
	0\leq |vr|<t^{-\frac{\epsilon}{2}}\downarrow 0;\ \ \ \ \ \ \ \ \ \ \ z\in\partial\mathbb{D}_{\hat{r}_t}(0):\ \ |\z(z)|\geq ctr=ct^{\frac{2}{3}}\rightarrow+\infty.
\end{equation}
Hence, with \eqref{low1} and \eqref{low2},
\begin{equation*}
	\Big\||s|^{\frac{1}{4}\sigma_3}\big({\bf G}_R(z;s,\gamma)-{\bf I}\big)|s|^{-\frac{1}{4}\sigma_3}\Big\|\leq c\left\{\frac{v}{rt}\left\|\begin{pmatrix}v & |f^{\pm}(z;s,\gamma)|^{-1}\\ 
	|f^{\pm}(s;z,\gamma)| & v\end{pmatrix}\right\|+\mathcal{E}^{(0)}(z;s,\gamma)\right\},\ \ z\in\partial\mathbb{D}_{\hat{r}_t}(0).
\end{equation*}
Estimates for the error term $\mathcal{E}^{(0)}(z;s,\gamma)$ follow from known error estimates for the confluent hypergeometric function, see e.g. \cite{NIST}: there exist $t_0>0$ and constants $c_j>0$ such that
\begin{equation*}
	\big\|\mathcal{E}^{(0)}(z;s,\gamma)\big\|\leq c_1\e^{c_2v-c_3rt}\leq c_1\e^{-c_4t^{\frac{2}{3}}}\text{ for all }t\geq t_0,\ \ 0\leq v<t^{\frac{1}{3}-\frac{\epsilon}{2}}\ \ \textnormal{with}\ \ \epsilon\in\left(0,\frac{2}{3}\right)\ \textnormal{fixed}.
\end{equation*}
But in the same scaling regime, compare \eqref{rconseq}, $|f^{\pm}(z;s,\gamma)|^{\pm 1}\rightarrow 1$.  Together these show the following.
\begin{prop}\label{shrink2} For every fixed $0<\epsilon<\frac{2}{3}$, there exist $t_0>0$ and $c>0$ such that
\begin{equation*}
	\Big\||s|^{\frac{1}{4}\sigma_3}\big({\bf G}_R(\cdot;s,\gamma)-{\bf I}\big)|s|^{-\frac{1}{4}\sigma_3}\Big\|_{L^2\cap L^{\infty}(\partial\mathbb{D}_{\hat{r}_t}(0))}\leq c\,v^2\,t^{-\frac{2}{3}}\text{ for all }t\geq t_0,\ 0\leq v<t^{\frac{1}{3}-\frac{\epsilon}{2}}.
\end{equation*}
\end{prop}
At this point we have modified our construction of $\partial\mathbb{D}_{r_t}(1)\cup\partial\mathbb{D}_{\hat{r}_t}(0)$ in RHP \ref{ratio1} to \eqref{scale}. We now have to analyze the remaining five jump contours. First, on the three infinite branches, starting with
\begin{align}
	z&\in\big(\Gamma_2\cup\Gamma_4)\cap\{z\in\mathbb{C}:\ |z|>\hat{r}_t\}:\ \ \ \ \ t\geq t_0,\ \ 0\leq v<t^{\frac{1}{3}-\frac{\epsilon}{2}}\nonumber\\
	&\big|(\mathcal{D}(z))^2\e^{2tg(z)}\big|=\exp\left[-\frac{v}{\pi}\arccos\left\{\frac{1}{\hat{r}_t}\big(|z-1|-1\big)\right\}+\frac{4t}{3}|z-1|^{\frac{3}{2}}\cos\left(\frac{3}{2}\textnormal{arg}(z-1)\right)\right]\nonumber\\
	&\hspace{2.3cm}\leq c_1\exp\left[-\frac{v}{3}+c_2v\hat{r}_t-c_3t\hat{r}_t|z-1|^{\frac{3}{2}}\right]\leq c_4\,\e^{-c_5t^{\frac{2}{3}}},\ \ \ \ c_j>0\label{inf1}
\end{align}
and 
\begin{align}
	z&\in(1+r_t,+\infty):\ \ \ \ \ t\geq t_0,\ \ 0\leq v<t^{\frac{1}{3}-\frac{\epsilon}{2}}\label{inf2}\\
	&\big|\e^{-t(\varkappa+2g(z))}(\mathcal{D}(z))^{-2}\big|=\exp\left[-\frac{v}{\pi}\arctan\left(\frac{2\sqrt{r_t}}{1-r_t}\right)-\frac{4t}{3}r_t^{\frac{3}{2}}\right]\leq c_1\e^{-c_2tr_t^{\frac{3}{2}}-c_3v\sqrt{r_t}}\leq c_4\e^{-c_5t^{\frac{3\epsilon}{2}}}, \ c_j>0.\nonumber
\end{align}
We summarize \eqref{inf1} and \eqref{inf2}.
\begin{prop}\label{shrink3} For every fixed $0<\epsilon<\frac{2}{3}$, there exist $t_0>0$ and $c_j>0$ such that
\begin{equation*}
	\|{\bf G}_R(\cdot;s,\gamma)-{\bf I}\|_{L^2\cap L^{\infty}(R_{\infty})}\leq c_1 t^{\frac{1}{3}}\exp\left[-c_2\,t^{\min\{\frac{2}{3},\frac{3\epsilon}{2}\}}\right]\text{ for all }t\geq t_0\text{ and }0\leq v<t^{\frac{1}{3}-\frac{\epsilon}{2}},
\end{equation*}
where $R_{\infty}$ denotes the three jump contours in RHP \ref{ratio1} that extend to infinity.
\end{prop}
Now we are left with the finite lens boundaries
\begin{align*}
	z&\in\gamma^{\pm}\cap\big\{z\in\mathbb{C}:\ |z|>\hat{r}_t,\ |z-1|>r_t\big\}:\ \ \ \ \ \ t\geq t_0,\ \ 0\leq v<t^{\frac{1}{3}-\frac{\epsilon}{2}}\\
	& \big|(\mathcal{D}(z))^2\e^{t(\varkappa+\phi(z))}\big|=\e^v\exp\left[\frac{v}{\pi}\textnormal{arg}\left(\frac{(z-1)^{\frac{1}{2}}-\im}{(z-1)^{\frac{1}{2}}+\im}\right)+\frac{4t}{3}|z-1|^{\frac{3}{2}}\cos\left(\frac{3}{2}\textnormal{arg}(z-1)\right)\right].
\end{align*}
For any point $z$ on the two boundaries that are away from the four endpoints, we obtain at once that
\begin{equation*}
	\big|(\mathcal{D}(z))^2\e^{t(\varkappa+\phi(z))}\big|\leq c_1\e^{c_2v-c_3t}\leq c_4\e^{-c_5t},\ \ t\geq t_0,\ \ 0\leq v<t^{\frac{1}{3}-\frac{\epsilon}{2}}.
\end{equation*}
At the endpoints, i.e. for $z\in\gamma^{\pm}\cap(\partial\mathbb{D}_{r_t}(1)\cup\partial\mathbb{D}_{\hat{r}_t}(0))$, we obtain quantitatively different estimates, namely,
\begin{equation*}
	z\in\gamma^{\pm}\cap\partial\mathbb{D}_{r_t}(1):\ \ \ \big|(\mathcal{D}(z))^2\e^{t(\varkappa+\phi(z))}\big|\leq c_1\e^{-c_2t^{\frac{3\epsilon}{2}}},\ \ \ t\geq t_0,\ \ 0\leq v<t^{\frac{1}{3}-\frac{\epsilon}{2}}
\end{equation*}
and
\begin{equation*}
	z\in\gamma^{\pm}\cap\partial\mathbb{D}_{\hat{r}_t}(0):\ \ \ \big|(\mathcal{D}(z))^2\e^{t(\varkappa+\phi(z))}\big|\leq c_1\e^{c_2v-c_3t^{\frac{5}{6}}}\leq c_4\e^{-c_5t^{\frac{5}{6}}},\ \ \ t\geq t_0,\ \ 0\leq v<t^{\frac{1}{3}-\frac{\epsilon}{2}}.
\end{equation*}
All together we obtain the next proposition.
\begin{prop}\label{shrink4} For every fixed $0<\epsilon<\frac{2}{3}$, there exist $t_0>0$ and $c_j>0$ such that
\begin{equation*}
	\|{\bf G}_R(\cdot;s,\gamma)-{\bf I}\|_{L^2\cap L^{\infty}(R_{\textnormal{lens}})}\leq c_1 t^{\frac{1}{3}}\exp\left[-c_2\,t^{\min\{\frac{5}{6},\frac{3\epsilon}{2}\}}\right]\text{ for all }t\geq t_0,\ \ 0\leq v<t^{\frac{1}{3}-\frac{\epsilon}{2}},
\end{equation*}
where $R_{\textnormal{lens}}$ denotes the two finite lens boundaries in RHP \ref{ratio1}.
\end{prop}
The content of Propositions \ref{shrink1}--\ref{shrink4} in summary allows us to solve the corresponding ratio RHP \ref{ratio1} iteratively.
\subsection{Iterative solution and expansion of \texorpdfstring{$F_2(s,\gamma)$}{F2}}
As indicated before, we are working with the ratio function
\begin{equation*}
	{\bf R}(z)=\begin{pmatrix}
	1 & 0\\
	-2\im\nu|s|^{\frac{1}{2}} & 1
	\end{pmatrix}{\bf L}(z)\begin{cases}
	\big({\bf P}^{(0)}(z)\big)^{-1},&z\in\mathbb{D}_{\hat{r}_t}(0)\\
	\big({\bf P}^{(1)}(z)\big)^{-1},&z\in\mathbb{D}_{r_t}(1)\\
	\big({\bf P}^{(\infty)}(z)\big)^{-1},&z\notin\big(\overline{\mathbb{D}_{\hat{r}_t}(0)}\cup\overline{\mathbb{D}_{r_t}(1)}\big)
	\end{cases}
\end{equation*}
in which $0<\hat{r}_t=t^{-\frac{1}{3}}$ and $0<r_t=t^{-\frac{2}{3}+\epsilon}$ with fixed $\epsilon\in(0,\frac{2}{3})$ both depend on $t$. Recalling our estimates in the previous subsection, we already have the following result.
\begin{prop} 
For every fixed $0<\epsilon<\frac{2}{3}$, there exist $t_0>0$ and $c_j>0$ such 
that
\begin{equation*}
	\Big\||s|^{\frac{1}{4}\sigma_3}\big({\bf G}_R(\cdot;s,\gamma)-{\bf I}\big)|s|^{-\frac{1}{4}\sigma_3}\Big\|_{L^2\cap L^{\infty}(\Sigma_R)}\leq c_1v^2\,t^{-\frac{2}{3}}+c_2t^{-\frac{1}{3}-\epsilon}\ \ \ \ \ \forall\,t\geq t_0,\ \ \ 0\leq v<t^{\frac{1}{3}-\frac{\epsilon}{2}}.
\end{equation*}
\end{prop}
We can now refer to general theory \cite{DZ} (see also \cite{BK} for the required modifications when working with contracting disks) and obtain the following result for the function (see also \eqref{Rfact})
\begin{equation*}
	\widehat{{\bf R}}(z)=|s|^{\frac{1}{4}\sigma_3}{\bf R}(z)|s|^{-\frac{1}{4}\sigma_3},\ \ z\in\mathbb{C}\backslash\Sigma_R;\ \ \ \ \ \ \ \widehat{{\bf R}}_+(z)=\widehat{{\bf R}}_-(z)|s|^{\frac{1}{4}\sigma_3}{\bf G}_R(z;s,\gamma)|s|^{-\frac{1}{4}\sigma_3},\ \ z\in\Sigma_R.
\end{equation*}
\begin{theo}
\label{improvedDZ} 
For every fixed $0<\epsilon<\frac{1}{3}$, there exist $t_0>0$ and $c_j>0$ 
such that the RHP for $\widehat{{\bf R}}(z;s,\gamma)$ is uniquely solvable for all $t\geq t_0$ and $0\leq v<t^{\frac{1}{3}-\epsilon}$. We can compute its solution iteratively through the integral equation
\begin{equation*}
	\widehat{{\bf R}}(z)={\bf I}+\frac{1}{2\pi\im}\int_{\Sigma_R}\widehat{{\bf R}}_-(\lambda)|s|^{\frac{1}{4}\sigma_3}\big({\bf G}_R(\lambda)-{\bf I}\big)|s|^{-\frac{1}{4}\sigma_3}\frac{\d\lambda}{\lambda-z},\ \ \lambda\in\mathbb{C}\backslash\Sigma_R,
\end{equation*}
using that
\begin{equation*}
	\|\widehat{{\bf R}}_-(\cdot;s,\gamma)-{\bf I}\|_{L^2(\Sigma_R)}\leq c_1v^2t^{-\frac{2}{3}}+c_2t^{-\frac{1}{3}-2\epsilon}\ \ \ \ \forall\,t\geq t_0,\ \ 0\leq v<t^{\frac{1}{3}-\epsilon}.
\end{equation*}
\end{theo}
In order to obtain the resulting expansion for $F_2(s,\gamma)$, we repeat and adjust the steps of Section \ref{sderivsec}. First, all steps leading to the exact identity \eqref{2nd} naturally carry over to the scaling regime \eqref{scale}. The influence of Theorem \ref{improvedDZ} manifests itself only in \eqref{1st}.  We now have
\begin{equation*}
	\frac{\gamma}{2\pi\im|s|}\e^{-2tg_+(0)}\Big(\big({\bf P}^{(0)}(\lambda)\big)^{-1}\big({\bf R}(0)\big)^{-1}{\bf R}'(0){\bf P}^{(0)}(\lambda)\Big)^{21}\bigg|_{\lambda\rightarrow 0}=\mathcal{O}\left(v^3|s|^{-\frac{5}{2}}\right)+\mathcal{O}\left(v|s|^{-2}\right),
\end{equation*}
and thus, after indefinite integration with respect to $s$,
\begin{equation}\label{panul}
	\ln F_2(s,\gamma)=-\frac{2v}{3\pi}|s|^{\frac{3}{2}}+\frac{3v^2}{8\pi^2}\ln|s|+C(v)+r(s,v),
\end{equation}
where the error term $r(s,v)$ is differentiable with respect to $s$ and for any $0<\epsilon<\frac{1}{3}$ there exist $t_0>0$ and $c_j>0$ such that
\begin{equation*}
	\big|r(s,v)\big|\leq c_1\frac{v^3}{|s|^{\frac{3}{2}}}+c_2\frac{v}{|s|}\ \ \ \ \text{ for all }t\geq t_0\text{ and }0\leq v<t^{\frac{1}{3}-\epsilon}.
\end{equation*}
The term $C(v)$ in \eqref{panul} is $s$-independent, so we can simply read it off from Theorem \ref{Eres:2} which, in turn, completes the proof of Theorem \ref{Eres:3}.

\section{Proof of the GOE and GSE expansions (Theorem \ref{OStheo})}
\label{OSsec}
We now determine the leading-order asymptotic expansion as $s\to-\infty$ of the thinned GOE and GSE Tracy-Widom distributions $F_1(s,\gamma)$ and $F_4(s,\gamma)$ for any fixed $\gamma\in[0,1)$. Note that Theorem \ref{OStheo} follows from Theorem \ref{Eres:3}, Proposition \ref{detformu}, and the following lemma.
\begin{lem}
For any fixed $\gamma\in[0,1)$, as $s\rightarrow-\infty$,
\eq
\label{as-partial-integral}
\int_s^\infty u_{_\textnormal{AS}}(t,\gamma)\,\d t = \frac{1}{2}\ln\left(\frac{1+\sqrt{\gamma}}{1-\sqrt{\gamma}}\right)+\mathcal{O}\left((-s)^{-\frac{3}{4}}\right).\endeq
\end{lem}
\begin{proof}
The total integral for the Ablowitz-Segur solution is known from \cite{BaikBDI:2009}, Theorem $2.1$ (with $s_1=-\im\sqrt{\gamma}$, compare \eqref{AS:behavior} and \cite{BaikBDI:2009}(27)) as
\eq
\label{total-integral}
\int_{-\infty}^\infty u_{_\textnormal{AS}}(t,\gamma)\,\d t = \frac{1}{2}\ln\left(\frac{1+\sqrt{\gamma}}{1-\sqrt{\gamma}}\right).
\endeq
A slight modification of the proof in \cite{BaikBDI:2009} yields the 
necessary result.  Equation \eqref{total-integral} was proven by analyzing 
the function $\Psi_2(\lambda;x)$ in \cite{BaikBDI:2009}, Equation (3) satisfying the Flaschka-Newell Lax pair 
for the Painlev\'e-II equation (corresponding to
$u_{_\textnormal{AS}}(x,\gamma)\sim\sqrt{\gamma}\text{Ai}(x)$ as $x\to+\infty$) for 
$\lambda\in\Omega_2$ (see \cite{BaikBDI:2009}, Figure $1$).  Here $\lambda$ is the auxiliary 
spectral variable and $x$ the Painlev\'e-II space variable.  As the 
$x$-equation in the underlying Lax pair (see \cite{BaikBDI:2009}, Equations (8) and (9)) simplifies when $\lambda=0$, one considers 
instead $P(x)=\lim_{\lambda\to 0}\Psi_2(\lambda;x)$, where the limit is taken 
in $\Omega_2$.  From \cite{BaikBDI:2009}, Equation (36), we find
\eq
\label{P21-left}
P^{21}(x) = \im\sinh\mu(x,\gamma)-\im\sqrt{\gamma}\cosh\mu(x,\gamma).
\endeq
Now \cite{BaikBDI:2009}, Equation (37) gives $\lim_{x\to-\infty}P^{21}(x)=0$.  
This can be strengthened to 
\eq
\label{P21-right}
P^{21}(x) = \mathcal{O}\left((-x)^{-\frac{3}{4}}\right) \text{ as } x\to-\infty
\endeq
by relating $\Psi_2(\lambda;x)$ to $\Psi_1(\lambda;x)$ in an adjacent 
$\lambda$-sector by the change of variables found in \cite{BaikBDI:2009}, 
Equation (37), and then computing the leading order asymptotics of 
$\Psi_2(0;x)$ via \cite{BaikBDI:2009}, Equations (81), (82), (78), (84), and 
(85).  Combining \eqref{P21-left} and \eqref{P21-right} shows 
\eq
\im\sinh\mu(x,\gamma)-\im\sqrt{\gamma}\cosh\mu(x,\gamma) = \mathcal{O}\left((-x)^{-\frac{3}{4}}\right) \text{ as } x\to-\infty.
\endeq
Taking into account that $\mu(x,\gamma)$ is uniformly 
bounded in $x$, the last expansion can be solved for $\mu(x,\gamma)$, and we obtain 
\eqref{as-partial-integral} upon replacing $x$ with $s$.
\end{proof}

\end{document}